\documentclass[pdftex,10pt]{article}
\usepackage[margin=1in]{geometry}
\usepackage{bm,bbm,amsmath,subcaption, amssymb,amsthm,mathtools,style_ec,booktabs,multirow,multicol,natbib,mathrsfs,enumitem}
\usepackage[table]{xcolor}
\usepackage{libertine}
\usepackage{hyperref}
\usepackage[nameinlink,capitalise]{cleveref}
\usepackage{setspace}
\definecolor{linkc}{rgb}{0.1, 0.5, 0.7}
\definecolor{citec}{rgb}{0, 0.8, 0.2}
\definecolor{urlc}{rgb}{0.5, 0.1, 0.2}
\hypersetup{
    colorlinks=true,
    linkcolor=linkc,
    citecolor=citec,
    urlcolor=urlc
}
\usepackage{makecell}

\providecommand{\keywords}[1]{\textbf{\textit{Keywords---}} #1}

\title{Metric Distortion Under Probabilistic Voting}

\author{
    Mohak Goyal\\
    \textit{Stanford University}\\
    \texttt{mohakg@stanford.edu}
    \and 
    Sahasrajit Sarmasarkar\\
    \textit{Stanford University}\\
    \texttt{sahasras@stanford.edu}
}

\date{} 

\begin{document}

\maketitle

\begin{abstract}
Metric distortion in social choice is a framework for evaluating how well voting rules minimize social cost when both voters and candidates exist in a shared metric space, with a voter's cost defined by their distance to a candidate. Voters submit rankings, and the rule aggregates these rankings to determine a winner. We extend this framework to incorporate probabilistic voting, recognizing that real-world voters exhibit randomness in how they vote. Our extension includes various probability functions, notably the widely studied Plackett-Luce (PL) model.

We show that the distortion results under probabilistic voting better correspond with conventional intuitions regarding popular voting rules such as \textsc{Plurality}, \textsc{Copeland}, \textsc{Random Dictator} and \textsc{Borda} than those under deterministic voting. For example, in the PL model with candidate strength inversely proportional to the square of their metric distance from a voter, we show that \textsc{Copeland}'s distortion is at most 2, whereas that of \textsc{Random Dictator} is $\Omega(\sqrt{m})$ in large elections (i.e., number of voters $n \rightarrow \infty$), where $m$ is the number of candidates. This contrasts sharply with the classical model, where \textsc{Random Dictator} beats \textsc{Copeland} with a distortion of 3 versus 5. In the PL model where the candidate strength is inversely proportional to the distance raised to the power $\theta$, the distortion under \textsc{Borda} is $\Theta(m^{1-2/\theta})$ when $\theta>2$ and $\Theta(1)$ otherwise. This generalizes the classical deterministic voting model where the distortion of \textsc{Borda} is $2m-1$. The proof uses a novel variant of asymptotic duality where we choose the Lagrange multiplier by asymptotically maximizing the derivative of the objective function. Overall, our work opens a new frontier for analyzing voting rules.

\end{abstract}


\makeatletter
\renewcommand\tableofcontents{%
    \@starttoc{toc}%
}
\makeatother
\tableofcontents

\newpage

\section{Introduction}


Societies make collective decisions despite conflicting interests, and the choice of the aggregation mechanism significantly impacts outcomes. 
Over the last century, there has been increasing interest in using computational tools to analyze and design voting rules \cite{arrow1950difficulty, sen1986social, arrow2010handbook, brandt2016handbook}. One prominent framework for evaluating voting rules is that of \textit{distortion} \cite{procaccia2006distortion}, where the voting rule has access only to the \textit{ordinal} preferences of the voters. However, the evaluation criterion is the sum of all voters' \textit{cardinal} utilities (or costs). The distortion of a voting rule is the worst-case ratio of the cost of the alternative it selects to the cost of the socially optimal alternative.

In \textit{metric distortion} \cite{anshelevich2015approximating}, an added assumption is that voters and candidates exist within a common but unknown metric space, where costs are represented by distances. These distances satisfy non-negativity and the triangle inequality. Each voter ranks candidates in order of increasing distance. A social choice function then combines these individual rankings to select a winner.

This model is inspired by a well-studied spatial voting model in economics \cite{enelow1984spatial, merrill1999unified}, 
and has a natural interpretation of voters liking candidates with a similar ideological position across many dimensions. 
While metric distortion is a powerful framework and has led to the discovery of interesting voting rules (e.g., Plurality Veto \cite{kizilkaya2022plurality} and the study of Maximal Lotteries \cite{kreweras1965aggregation} for metric distortion by \cite{charikar2024breaking}), its outcomes often do not align with conventional intuitions about popular voting rules.
For example, the overly simple \textsc{Random Dictator} (where the winner is the top choice of a uniformly randomly selected voter) outperforms \textsc{Copeland} (Definition~\ref{def:copeland}) (which satisfies the Condorcet Criterion \cite{arrow2010handbook} 
and other desirable properties), with metric distortion 3 versus 5 \cite{anshelevich2015approximating}.

While not yet adopted in the metric distortion framework, there is a mature line of work on \textit{probabilistic voting} \cite{coughlin1992probabilistic, quinn1999voter, mckelvey2006theory}. Here, the focus is on behavioral models of voters and on accounting for the randomness of their votes. Two often postulated sources of this randomness are the bounded rationality of voters and noise in their estimates of candidates' positions.
Similarly, in random utility models (RUMs) \cite{Pfeiffer2012adaptive, parkes2012random, soufiani2013preference} in social choice, it is assumed that candidates have ground-truth strengths and that voters make noisy observations of these strengths and vote accordingly.
%
%
The questions we ask in this paper are:


\emph{1. How can we incorporate probabilistic voting into the metric distortion framework? }

\emph{2. How does the metric distortion of popular voting rules under probabilistic voting differ from that under the traditional deterministic model? }

\subsection{Key Contributions}


We present the \emph{first} extension of metric distortion to probabilistic voting (Definition~\ref{def:distortion-probabilistic}), which is a much-needed generalization because, in practice, voters have been shown to vote with some degree of randomness \cite{mckelvey2006theory}. In this model, an adversary selects a metric space in which voters and candidates reside, and each voter probabilistically ranks all the candidates based on their distances to the candidates.
We present an axiomatic characterization of probability distributions $\mathcal{P}$ over submitted rankings to assess their suitability within the metric distortion framework. These axioms include \textit{scale-freeness} (Axiom~\ref{axiom:scalefree}), which states that scaling all distances by a constant should not affect probabilities; \textit{independence of other candidates} (Axiom~\ref{axiom:ioc}), which ensures that the relative ranking of two candidates is unaffected by the position of other candidates; and \textit{monotonicity of pairwise marginal probabilities in costs} (Axiom~\ref{axiom:monotonicity}), which asserts that the probability of ranking candidate \(X\) over \(Y\) increases if \(X\) moves closer to the voter while \(Y\) remains at the same distance.

We focus on the induced marginal probability of pairwise orderings in the submitted rankings. Specifically, we study a class of functions \(\G\) (Definition~\ref{def:pairwise_prob_function_class}) that characterize these marginal probabilities while adhering to our axioms.
We define metric distortion under probabilistic voting as follows: given a function \( g \) from the class \( \G \), we consider the worst case over all probability distributions \(\mathcal{P}\) that induce the marginal probability of candidate \( A \) being ranked above candidate \( B \) by voter \( i \) as 
$
g\left(\frac{d(i,B)}{d(i,A)}\right)
$
for every pair of candidates \( A \) and \( B \), where \( d(\cdot) \) represents the metric distance function. We then consider the worst case over all possible distance functions \( d(\cdot) \) to determine the metric distortion associated with \( g \). 
An example of such a function is 
$
g(r) = \frac{1}{1 + r^{-\theta}}
$
for \( \theta > 1 \), which arises naturally in the well-studied Plackett-Luce model.

  We study various voting rules, including \textsc{Random Dictator} (randomized), \textsc{Plurality} (first-past-the-post), \textsc{Borda} (a positional scoring rule that can also be interpreted as a weighted tournament rule), \textsc{Copeland} (from the class of unweighted tournament rules \cite{condorcetsocialchoice}), Weighted-Uncovered-Set~\cite{munagala2019improved} (from the class of weighted tournament rules \cite{condorcetsocialchoice}), and \textsc{Plurality Veto} \cite{kizilkaya2022plurality}, which attains optimal metric distortion among deterministic voting rules. 
Our analysis for \textsc{Borda} and \textsc{Plurality Veto} is specifically for the PL model ($g(r) = \frac{1}{1+r^{-\theta}}$ for $\theta >1$), whereas our analysis for the other rules is for any $g \in \G.$
  %

For our technical results, we derive the following quantities from function $g(\cdot)$:
\begin{align}{\label{eq:gm_go_defn}}
    \gm = \sup\limits_{x \in (0,1)} \frac{1}{x} g_{\textsc{mid}}(x) \text{ where } g_{\textsc{mid}}(x) = g\left(\frac{x}{1-x}\right) \text{ for } x \in [0,1]
\end{align}
An intuition for $\gm$ can be seen from the following setup. Consider a one-dimensional Euclidean metric with two candidates positioned at $0$ and $1$, where the candidate at $0$ is socially optimal. The adversary, whose goal is to maximize distortion, places a voter at position $x$. The probability that this voter ranks the candidate at $1$ above the candidate at $0$ is $g\left(\frac{x}{1-x}\right)$. The adversary has to trade off this probability against $x$, since increasing $x$ lowers the distortion.

The second one-edge obstruction is captured by a two-type optimization. For \(\alpha\in(0,1/2]\), let
\[
\Delta_g(\alpha)
:=
\max\left\{
\frac{\gm}{\alpha}-1,\,
R(g,\alpha)^{-1}
\right\},
\]
where \(R(g,\alpha)\) is the value of the explicit two-point system in Equation~\eqref{eq:R_C_defn}. Equivalently, \(\Delta_g(\alpha)=\Opt(\mathcal E_\alpha)^{-1}\), where \(\Efrac{\alpha}\) (equation~\eqref{eqn:optim_formulation_two_type}) is the one-edge fractional voter-type program used in the \textsc{Plurality} analysis. We also use the upper envelope
\[
\overline{\Delta}_g(\alpha)
:=
\max\left\{
\frac{\gm}{\alpha}-1,\,
\overline L(g,\alpha)^{-1}-1
\right\},
\]
where \(\overline L(g,\alpha)\) is defined by the single-point system in Equation~\eqref{eq:L_upper_bound_defn}; it satisfies \(\Delta_g(\alpha)\leq\overline{\Delta}_g(\alpha)\).
For example, for $g(r)=1/(1+r^{-2})$, \(\gm\approx1.21\), \(\overline L(g,1/2)\approx0.4\), and \(R(g,1/2)^{-1}\approx1.41\). Our key results for general functions $g \in \G$ are in Table \ref{tab:distortion_prob_voting}, and those specifically on the PL model are in Table \ref{tab:dist_PL_voting}. 


\begin{table}[t]
\centering
\caption{Comparing distortion under probabilistic voting and deterministic voting for large elections.}
\label{tab:distortion_prob_voting}
\resizebox{\textwidth}{!}{%
\begin{tabular}{lccc}
\toprule
\multirow{2}{*}{Voting rules} & \multirow{2}{*}{\makecell{Deterministic\\voting}} & \multicolumn{2}{c}{Probabilistic Voting ($g(\cdot)$)} \\
\cmidrule(lr){3-4}
 & & Upper Bound & Lower Bound \\
\midrule
\textsc{Random Dictator} & $3$~\cite{anshelevich2017randomized} & $1+ (m-1)\gm$ & $2+ [g^{-1}(\frac{1}{m-1})]^{-1}$ \\
\textsc{Plurality} & $2m-1$~\cite{anshelevich2015approximating} & $\Delta_g(1/m)$ & $\Delta_g(1/m)$ \\
\textsc{Copeland} & $5$~\cite{anshelevich2015approximating} & $\max\{\Delta_g(1/2),\Delta_g(1/2)^2\}$ & $\Omega(1)$ \\
\bottomrule
\end{tabular}%
}
\end{table}

\begin{informaltheorem}[Theorems \ref{theorem:thm_plurality_distortion_m}, \ref{theorem:thm_plurality_distortion_m_lower_bound}]
  The metric distortion of \textsc{Plurality} under probabilistic voting is exactly \(\Delta_g(1/m)=\max\{m\gm-1,R(g,1/m)^{-1}\}\) for large elections. In particular, it scales linearly with \(m\).
\end{informaltheorem}

\begin{informaltheorem}[Theorem \ref{theorem:RD_distortion_lower_bound}]
    The metric distortion of \textsc{Random Dictator} under probabilistic voting increases monotonically with the number of candidates and is lower bounded by $2 + \left(g^{-1}\left(\frac{1}{m-1}\right)\right)^{-1}$ for large elections. For example, for the PL model, this is $\Omega(m^{\frac{1}{\theta}})$ for $\theta >1.$
\end{informaltheorem}

This contrasts with deterministic voting, where the distortion of \textsc{Random Dictator} is 3. We now show that the distortion of \textsc{Plurality Veto} scales with the number of candidates under probabilistic voting.
\begin{informaltheorem}[Theorem \ref{thm:plurality_veto_lb}]
    The metric distortion of \textsc{Plurality Veto} under the PL model ($g(r) = \frac{1}{1+r^{-\theta}}$ for $\theta >1$) is $\Omega(m^{\frac{1}{1+\theta}})$. 
\end{informaltheorem}

This result presents a \textit{stark} contrast between deterministic voting and probabilistic voting. However, this behavior is intuitively consistent with the rule: \textsc{Plurality Veto}, like \textsc{Plurality}, discards most of the information in the ranking and only looks at the top-ranked candidate and the bottom-ranked candidate among the not-yet-eliminated candidates.

\begin{table}[t]
\centering
\caption{Comparing distortion under the PL model ( $g(r) = \frac{1}{1+r^{-\theta}}$) and deterministic voting for large elections.}
\label{tab:dist_PL_voting}
\resizebox{\textwidth}{!}{%
\begin{tabular}{lccc}
\toprule
\multirow{2}{*}{Voting rules} & \multirow{2}{*}{\makecell{Deterministic\\voting}} & \multicolumn{2}{c}{PL model of voting (parameter $\theta$)} \\
\cmidrule(lr){3-4}
 & & Upper Bound & Lower Bound \\
\midrule
\textsc{Random Dictator} & $3$~\cite{anshelevich2017randomized} & $1+ (m-1)\gm$ & $1+ \frac{(m-1)^{\frac{1}{\theta}}}{2}$ \\
\textsc{Plurality} & $2m-1$~\cite{anshelevich2015approximating} & $O(m)$ & $\Omega(m)$ \\
\textsc{Copeland} & $5$~\cite{anshelevich2015approximating} & $\max\{\Delta_g(1/2),\Delta_g(1/2)^2\}$ & $\Omega_\theta(1)$ \\
\textsc{Borda} & $2m-1$~\cite{anshelevich2015approximating} & $O(m^{\max(1-\frac{2}{\theta},0)})$ & $\Omega(m^{\max(1-\frac{2}{\theta},0)})$ \\
\textsc{Plurality Veto} & $3$~\cite{kizilkaya2022plurality} & - & $\Omega(m^{\frac{1}{1+\theta}})$ \\
\bottomrule
\end{tabular}%
}
\end{table}

We show that \textsc{Copeland} is robust to randomness in the votes, unlike \textsc{Plurality Veto}, and its distortion under probabilistic voting is even lower than that under deterministic voting.

\begin{informaltheorem}[Theorem \ref{theorem:Copeland_distrotion_m}]
The metric distortion of \textsc{Copeland} under probabilistic voting for large elections is at most
$\max\left\{\Delta_g(1/2),\Delta_g(1/2)^2\right\}$ which is a constant independent of \(m\).
\end{informaltheorem}

We provide upper bounds on metric distortion for every finite $n$ in Theorems \ref{theorem:thm_plurality_distortion_m}, \ref{theorem:Copeland_distrotion_m}, and \ref{theorem:RD_distortion_upper_bound}.
We further observe that our distortion bounds in Table~\ref{tab:distortion_prob_voting} match those of deterministic voting when we choose $g(\cdot)$ as an indicator function that equals one if and only if the input is at least one. However, we restrict our analysis to doubly differentiable functions for simplicity.

\begin{informaltheorem}[Theorem \ref{theorem:borda_distortion}]
    The metric distortion of \textsc{Borda} under the PL model ($g(r) = \frac{1}{1+r^{-\theta}}$ for $\theta >1$) for large elections is of order $\Theta(m^{1-\frac{2}{\theta}})$ if $\theta>2$ and $\Theta(1)$ otherwise. 
\end{informaltheorem}
This result shows that classical metric distortion is too pessimistic about \textsc{Borda} and looks only at the extreme case of no randomness in the votes. Further observe that when $\theta<\sqrt{3}$, the distortion of \textsc{Borda} is asymptotically smaller than that of \textsc{Plurality Veto}.
Our results demonstrate a transition toward deterministic voting as \(\theta \to \infty\).

Therefore, our framework provides a ``complete'' view of metric distortion, in contrast to prior work, which has only examined the ``asymptotic'' behavior.

\subsection{Main Techniques}

The following technique forms the basis of the analysis for \textsc{Copeland}, \textsc{Plurality}, and \textsc{WeightedUncovered}. For the case of two candidates, we define a linear-fractional program that finds the highest possible ratio of social costs of the candidates subject to the constraint that the ``bad'' candidate gets at least an $\alpha$ fraction of votes in expectation. We linearize this program via the sublevel-set technique \cite{boyd2004convex} and find a feasible solution using KKT constraints. Concentration inequalities on this solution provide an upper bound on the distortion. For \textsc{Plurality}, we set $\alpha$ to $\frac{1}{m}$ since there can be $m-1$ ``equally good'' candidates, each getting a $\frac{1}{m}$ share of the total votes in the worst case. For \textsc{Copeland}, we set $\alpha$ to $\frac{1}{2}$ and then show that the distortion can be at most the square of the optimal value of the program. For \textsc{WeightedUncovered}, we solve the program for two different values of $\alpha$ and then combine the solutions to find a bound on the distortion. We find a matching lower bound for \textsc{Plurality} by constructing an example in the 3-D Euclidean space.
The lower bound for \textsc{Random Dictator} is obtained via an example
in the 1-D Euclidean space. 

The technique for 
analyzing the distortion of \textsc{Borda} under the PL model is quite different. We need to show that a candidate \( W \) whose social cost is at least \( \omega(m^{\max(1-\frac{2}{\theta},0)}) \) times the social cost of another candidate \( B \) will lose the election with high probability. We normalize the metric space by setting \( d(B,W) = 1 \). We define a ball \( \C \) of radius \( \frac{1}{2} \) and another ball \( \V \) of a smaller radius around \(B\). We identify a candidate \( B^{(b)} \) that attains the highest expected `truncated' Borda score among the candidates in \( \C \), using only voters in \(\V\). We then show that the expected Borda score of \( B^{(b)} \) exceeds that of \( W \) by at least \( cn \) for some constant \( c \).
We prove this by analyzing each candidate outside the ball \( \C \) separately. When its normalized distance from \(B^{(b)}\) is below a fixed threshold, the triangle inequality gives the required bound. When this normalized distance is above the threshold, we formulate an optimization problem for the candidate's contribution to the expected Borda-score difference between \(B^{(b)}\) and \(W\). The problem for a specific value of \(m\) is difficult since there may be multiple inflection points of the objective function. To address this, we analyze the asymptotic maximum of the derivative of a carefully chosen lower bound of the objective function as \(m\) approaches infinity. We use this maximum as the Lagrange multiplier to obtain an order-wise tight bound. 



\subsection{Related Work}
\textbf{Metric Distortion~~ } 
The work in \citep{anshelevich2015approximating} initiated the study of metric distortion and showed that any deterministic voting rule 
has a distortion of at least 3, and \textsc{Copeland} has a distortion of 5.
The Plurality Veto rule attains the optimal distortion of 3 for deterministic voting rules \cite{kizilkaya2022plurality}. 
The work in \citep{pulyassary2021randomizedmetricdistortionconjecture} showed that any randomized voting rule has a distortion of at least 2.06, later improved to 2.11 in \citep{charikar2022metric}.
The work in \citep{charikar2024breaking} gave a randomized voting rule with a distortion of at most 2.75, breaking the barrier of 3.
The work in \citep{kempe2020analysis} gave an LP duality framework for the study of distortion.
See \citep{anshelevich2021distortion} for a useful survey on distortion in social choice.

\textbf{Distortion with Alternate Information~~}
The work in \citep{abramowitz2019awareness} showed that deterministic voting rules achieve a distortion of 2 when voters provide preference strengths as ratios of distances. The work in \citep{amanatidis2021peeking} demonstrated that even a few queries from each voter can significantly improve distortion in non-metric settings. The work in \citep{anshelevich2024improved} examined threshold approval voting, where voters approve candidates with utilities above a threshold. The work in \citep{goel2025metric} obtains better distortion bounds when voters submit their preferences after aggregating them in small groups. Our work relates to these studies since, in probabilistic voting, the probability of a voter switching the order of two candidates in their vote relative to their true order depends on the \textit{relative strength} of their preference. However, we may also lose useful information if the degree of randomness is too high.   


\textbf{Probabilistic Voting and Random Utility Models (RUMs)~~}
The work in \citep{hinich1977equilibrium} showed that the celebrated Median Voter Theorem of \cite{black1948rationale} does not hold under probabilistic voting, underlining the fact that important results in voting theory attained under the assumption of `deterministic' voters do not hold more generally. Classical work on probabilistic voting has focused on studying the equilibrium positions of voters and/or candidates in game-theoretic models \cite{patty2005local, coughlin1981electoral, coughlin1981directional, mckelvey2006theory}, and empirical validation from national elections in Europe  \cite{alvarez2000issues, schofield1998multiparty, quinn1999voter}.
The work in \citep{mckelvey2006theory} adopts the quantal response model, a popular way to model agents' bounded rationality. In \cite{Ghodsi_Latifian_Seddighin_2019}, the authors study a setting in which voters may abstain from voting with a probability that depends on their utilities for the candidates.


RUMs have been studied in social choice \cite{Pfeiffer2012adaptive, soufiani2013preference} with the hypothesis that candidates have \textit{universal} ground-truth strengths, which voters make noisy observations of. Our model is the same as RUM regarding the voters' behavior; however, voters have \textit{independent} costs from candidates. The Plackett-Luce (PL) model \cite{plackett1975analysis, luce2005individual} has been widely studied in social choice \cite{gormley2006analysis, azari2012random,gormley2004grade}. For probabilities on pairwise orders, PL reduces to the Bradley-Terry (BT) model \cite{bradley1952rank}. These probabilities are proportional to candidates' strengths (which we define as the inverse of powers of costs). 


The widely studied Mallows model \cite{mallows1957non}, based on \citep{condorcet1785essay}, is obtained as follows: the voter flips the order of each candidate pair (relative to a ground truth ranking) with a constant probability $p \in (0,\frac{1}{2})$ \cite{caragiannis2016noisy, liu2023robust}. The process is repeated if a linear order is not attained. In the context of metric distortion, a limitation of this model is that it does not account for the relative distance of candidates to the voter. Therefore, we do not include this model in our work. See \citep{marden1996analyzing} for a comprehensive review of RUMs.
The work in \citep{critchlow1991probability} does an axiomatic study of RUMs; our axioms are grounded in the metric distortion framework and are distinct from those therein.

\textbf{Comparison with Smoothed Analysis~~} Recently, there has been significant interest in smoothed analysis \cite{spielman2004smoothed} of social choice. Here, a small amount of randomness is added to problem instances and its effect is studied on the satisfiability of axioms \cite{baumeister2020towards, flanigan2023smoothed, xia2020smoothed, xia2023semi} and the computational complexity of voting rules \cite{liu2022semi, xia2021smoothed, xia2022beyond}. In \citep{baumeister2020towards}, this model is termed as being `towards reality,' highlighting the need to study the randomness in the election instance generation processes. Unlike smoothed analysis where the voter and candidate positions are randomized, we consider these positions fixed, but the submitted votes are random given these positions. The technical difference appears in the benchmark (the ``optimal'' outcome in the denominator of the distortion is unchanged in our framework and changes due to randomness in smoothed analysis).  The work in \citep{caragiannis2023beyond} does an average-case analysis of distortion in impartial culture electorates.

\subsection{Preliminaries and Notation}

Let $\N$ be a set of $n$ voters and $\A$ be the set of $m$ candidates. Let $\s$ be the set of total orders on $\A$. Each voter $i \in \N$ has a preference ranking $\sigma_i \in \s$.
A vote profile is a tuple of preference rankings $\sigma_{\N} = (\sigma_1,\ldots,\sigma_n) \in \s^n$ for all voters. The tuple $(\N,\A,\sigma_{\N})$ defines an election instance.
 Let $\Delta(\A)$ denote the set of all probability distributions over the set of candidates.
\begin{definition} [Voting Rule]
    A voting rule $f : \s^n \rightarrow \Delta(\A)$ takes a vote profile $\sigma_{\N}$ and outputs a probability distribution $\textbf{p}$ over the candidates.
\end{definition}

For deterministic voting rules, we overload notation by saying that the rule's output is a candidate rather than a distribution.
 Let $\mathbb{I}$ denote the indicator function. We use $|AB|$ to denote the number of voters who rank candidate $A$ above candidate $B$, that is, $|AB| = \sum_{i\in \mathcal{N}} \mathbb{I}(A \succ_{\sigma_i} B)$. We now present some voting rules from the existing literature.




%
\begin{definition}[\textsc{Random Dictator}]\label{def:rd}
    Select a voter uniformly at random and output their top choice, i.e.,
    $\textsc{Random Dictator}(\sigma_{\N}) = \textbf{p}$ such that ${p}_j = \frac{1}{n}\sum_{i \in \N} \mathbb{I}(\sigma_{i,1} = j).$
\end{definition}

%

\begin{definition}[\textsc{Plurality}] \label{def:plurality}
Choose the candidate who is the top choice of the most voters, i.e., $\textsc{Plurality}(\sigma_{\N}) = \argmax_{j \in \A} \sum_{i \in \N} \mathbb{I}(\sigma_{i,1} = j).$ Ties are broken arbitrarily.
\end{definition}

\begin{definition}[\textsc{Copeland}] \label{def:copeland}
   Choose the candidate who wins the most pairwise comparisons, i.e.,
  $\textsc{Copeland}(\sigma_{\N}) = \argmax_{j \in \A} \sum_{j' \in \A \setminus \{j\}} \mathbb{I}(|jj'| > \frac{n}{2})$. 
   Ties are broken arbitrarily.
\end{definition}

\begin{definition} [Borda Rule and Borda score]
  The Borda score of a candidate \( j \) is the total number of votes they receive in pairwise contests against all other candidates. That is, $\textsc{Borda score}(j) = \sum_{j' \in \A \setminus \{j\}} |jj'|.$ The candidate with the highest Borda score wins, breaking ties arbitrarily.
\end{definition}

\begin{definition}[ $\lambda$-weighted uncovered set and \textsc{WeightedUncovered} Rule \cite{munagala2019improved}] \label{def:wu}
    For $\lambda \in [0,1],$ the $\lambda$-weighted uncovered set $A_{\lambda} \subseteq \A$ has the following property: for any $j \in A_{\lambda}$ and $j' \neq j$, we have:
    \begin{itemize}
        \item either $|jj'| \geq (1-\lambda) n$; or
        \item there is another candidate $j'' \notin \{j,j'\}$ such that $|jj''| \geq (1-\lambda)n$ and $|j''j'| \geq \lambda n$.
    \end{itemize}
     The \textsc{WeightedUncovered} rule picks the alphabetically first candidate in the $\varphi$-weighted uncovered set, where $\varphi=\frac{\sqrt{5}-1}{2}$ is the inverse golden ratio.
\end{definition}

It is shown in \citep[Lemma 3.2]{munagala2019improved} that the $\lambda$-weighted uncovered set is always non-empty for $\lambda \in [0.5,1]$.
 We bound the distortion of \textsc{WeightedUncovered} in Corollary \ref{corollary:weighted_uncovered_set} using techniques similar to those used for \textsc{Copeland}. More generally, we can similarly analyze a family of voting rules that pick the alphabetically smallest candidate in the $\lambda$-weighted uncovered set for any $\lambda \in [0.5,1]$. We do not define the voting rule \textsc{Plurality Veto} \cite{kizilkaya2022plurality} in this section; instead, we refer the reader to \cite[Section 3]{kizilkaya2022plurality} for a detailed description and include a brief discussion in Section \ref{sec:Plurality_Veto}.
We now formally present the concept of metric distortion.
The distance function $d: (\N \cup \A)^2 \rightarrow \mathbb{R}_{\geq 0}$ satisfies the triangle inequality ($d(x,y)\leq d(x,z) + d(z,y)$) and symmetry ($d(x,y) = d(y,x)$). 
The distance between voter $i \in \N$ and candidate $j \in \A$ is also referred to as the \textit{cost} voter $i$ experiences from candidate $j$. 
We consider the most commonly studied social cost: the sum of the costs of all voters, defined as $\SC(j,d) := \sum_{i \in \N} d(i,j).$

In deterministic voting, the preference ranking $\sigma_i$ of voter $i$ is consistent with the distances. That is, $d(i,j) > d(i,j') \Leftrightarrow j' \succ_{\sigma_i} j$ for all voters $i \in \N$ and candidates $j,j' \in \A.$ Ties in the distance are broken arbitrarily. Let $\rho(\sigma_{\N})$ be the set of distance functions $d$ consistent with the vote profile $\sigma_{\N}.$
The metric distortion of a voting rule is defined as follows: 

\begin{definition} [Metric Distortion\cite{anshelevich2015approximating}] \label{def:distortion}
    $\textsc{dist}(f) = \sup\limits_{\N,\A,\sigma_{\N}} ~\sup\limits_{d \in \rho(\sigma_{\N})} ~\frac{\E[\SC(f(\sigma_{\N}),d)]}{\min\limits_{j\in \A} \SC(j,d)}.$
\end{definition}

\section{Axioms and Model} \label{sec:model}
Under probabilistic voting, the submitted preferences may no longer be consistent with the underlying distances. 
For a distribution $\Rho(d)$ over $\sigma_{\N}$, let $q^{\Rho(d)}(i,j,j')$ denote the induced marginal probability that voter $i$ ranks candidate $j$ higher than $j'.$  
We focus on these marginal probabilities of pairwise orders and provide axioms for classifying which $q^{\Rho(d)}(\cdot)$ are suitable for studying metric distortion. Our three axioms are as follows.


\begin{axiom}[Scale-Freeness (SF)] \label{axiom:scalefree}
    The probability distribution $\Rho(d)$ has induced marginal probabilities $q^{\Rho(d)}(\cdot)$ that are invariant to scaling of $d$.
    That is, for any tuple $(i,j,j')$ and any constant $\kappa > 0$, we must have $q^{\Rho(d)}(i,j,j') = q^{\Rho(\kappa d)}(i,j,j').$
\end{axiom}

Note that the metric distortion (Definition~\ref{def:distortion}) for deterministic voting is scale-free. We want to retain this property in the probabilistic model. Conceptually, one may think of a voter's preferences as being a function of the relative, rather than absolute, distances to the candidates.

\begin{axiom}[Independence of Other Candidates (IOC)] \label{axiom:ioc}
    For any $i \in \N$ and $j, j' \in \A$, the induced marginal probability $q^{\Rho(d)}(i,j,j')$ is not a function of the distance of voter $i$ to any other candidate $j'' \notin \{j,j'\}$.
\end{axiom}


IOC extends Luce's classical choice axioms \cite{luce2005individual} of consumer behavior, which were defined only for selecting the top choice. The deterministic voting framework satisfies IOC.
IOC is also reminiscent of the \textit{independence of irrelevant alternatives} axiom for voting rules. IOC is somewhat restrictive and may not always hold true in real-world data; however, it makes analysis much easier by allowing us to `factor' distributions over rankings into marginals over pairwise orders.

\begin{axiom} [Strict Monotonicity (SM)] \label{axiom:monotonicity}
   For every tuple  $(i,j,j')$, for fixed distance $d(i,j) > 0,$ the probability $q^{\Rho(d)}(i,j,j')$ is strictly increasing in $d(i,j')$.  
\end{axiom}

The monotonicity in $d(i,j)$ follows since $q^{\Rho(d)}(i,j',j) = 1- q^{\Rho(d)}(i,j,j').$ This axiom is natural.

\begin{example} [Mallows model fails Axioms~\ref{axiom:ioc} and~\ref{axiom:monotonicity}.]
In the Mallows model \cite{mallows1957non}, $q^{\Rho(d)}(\cdot)$, as derived in \citep{busa2014preference}, is given by:
  $ q^{\Rho(d)}(i,j,j') = h(r_{j'} - r_j+1, \phi) - h(r_{j'} - r_j, \phi).$

Here, $h(k,\phi) = \frac{k}{(1-\phi^k)}$, while $r_j$ and $r_{j'}$ are the positions of $j$ and $j'$ in the ground-truth (noiseless) ranking, and the constant $\phi$ is a dispersion parameter.

The Mallows model fails Axiom~\ref{axiom:ioc} since it depends on the number of candidates between $j$ and $j'$ in the noiseless ranking, which depends on the distances of the candidates other than $j$ and $j'$ from voter $i$.

It also fails Axiom~\ref{axiom:monotonicity} since it does not depend on the distances but only on their order.
\end{example}

\begin{definition}[Plackett-Luce Model \cite{luce2005individual,plackett1975analysis}] \label{def:pl}
    A ranking in the PL model is constructed as follows: for each voter $i \in \N$, each candidate $j \in \A$ has a `strength' $s_{i,j}$. The voter chooses their top choice with probability proportional to the strengths. Similarly, for every subsequent rank, they choose a candidate from among the \textit{remaining} ones with probability proportional to their strengths.
\end{definition}
 In most of the literature on RUMs, a common assumption is that the strength of a candidate is the same for all voters. However, we choose the more general model where a candidate's strength differs across voters, since it is more natural in our setting.
 In terms of the pairwise order probabilities, the PL model is the same as the Bradley-Terry (BT) model \cite{bradley1952rank}, that is:
\begin{equation}{\label{eq:plackett_luce}}
  \text{PL/BT:}~~~~~~~ q^{\Rho(d)}(i,j,j') = \frac{s_{i,j}}{s_{i,j} + s_{i,j'}}.
\end{equation}
Prima facie, any decreasing function of distance $d(i,j)$ would be a natural choice for $s_{i,j}$. However, not all such functions satisfy Axiom~\ref{axiom:scalefree}. The exponential function is a popular choice in the literature using BT or PL models. However, in general, $\frac{e^{-d(i,j)}}{e^{-d(i,j)} + e^{-d(i,j')}} \neq \frac{e^{-2d(i,j)}}{e^{-2d(i,j)} + e^{-2d(i,j')}},$ so it does not satisfy scale-freeness.

On the other hand, observe that all functions $s = d^{-\theta}$ for any $\theta \in (1, \infty)$ satisfy our axioms. Therefore, we adopt this form of the PL model for the analysis in this paper. When a voter is at distance zero from one or more candidates, we use the natural limiting convention: zero-distance candidates are ranked above all positive-distance candidates, and ties among zero-distance candidates are broken uniformly at random.

\begin{table}[t]
    \centering
    \caption{Axioms satisfied by commonly studied models of probabilistic voting}
    \begin{tabular}{lccc}
        \toprule
        & Axiom 1: SF & Axiom 2: IOC & Axiom 3: SM \\
        \midrule
        Mallows &  \checkmark & $\bm{\times}$ & $\bm{\times}$ \\
        PL with exponential in distance &  $\bm{\times}$ &  \checkmark & \checkmark \\
        PL with powers of distance & \checkmark & \checkmark & \checkmark \\
        \bottomrule
    \end{tabular}
\end{table}

\begin{figure}[t]
    \centering
    \begin{subfigure}[b]{0.45\textwidth}
        \centering
        \includegraphics[width=0.95\textwidth, trim=0.8cm 0.4cm 2cm 1.8cm, clip]{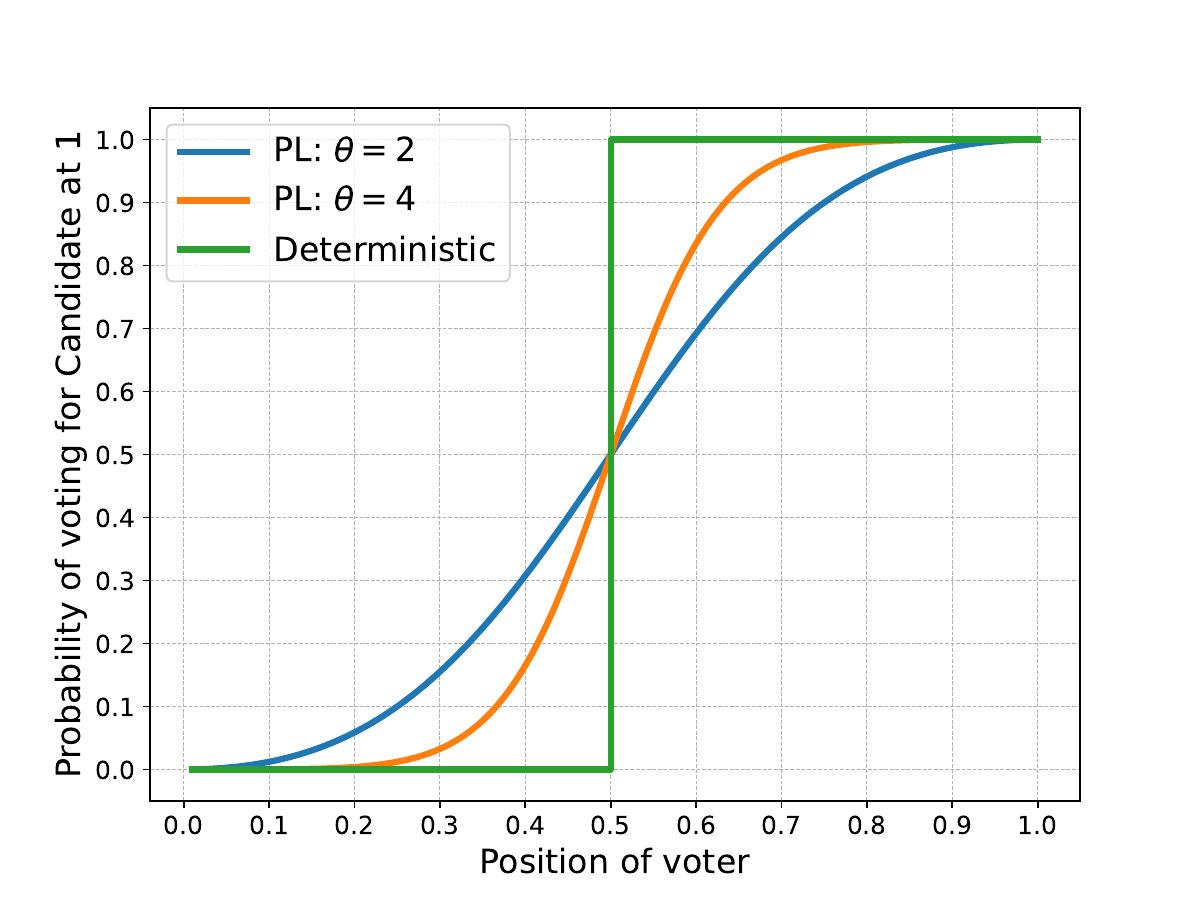}
    \end{subfigure}
    \hfill
    \begin{subfigure}[b]{0.45\textwidth}
        \centering
        \includegraphics[width=0.95\textwidth, trim=0.6cm 0.4cm 2cm 1.5cm, clip]{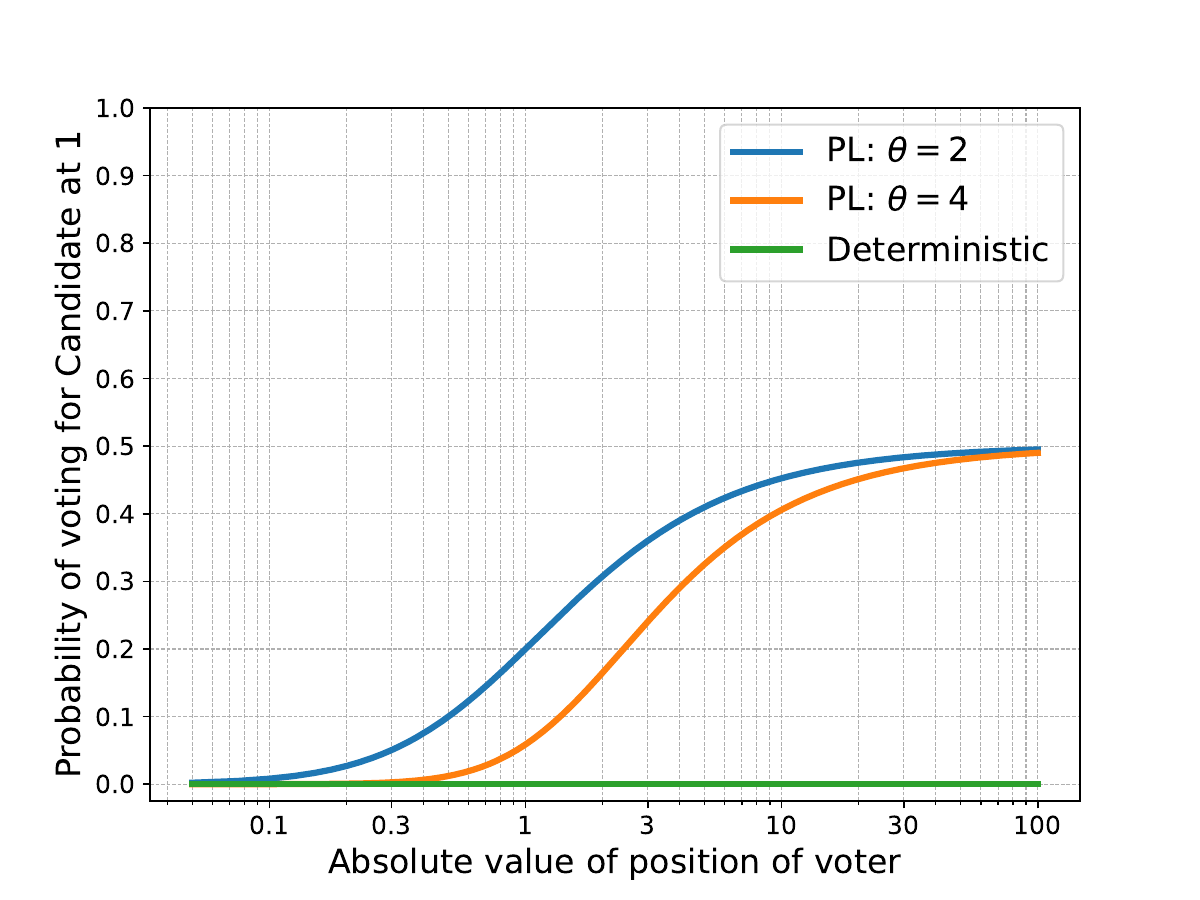}
    \end{subfigure}
    \caption{A 1-D Euclidean example of voting probabilities. There are two candidates, positioned at 0 and 1. The left figure shows the voter positioned between 0 and 1, while the right figure shows the case where the voter is to the left of 0. As the distance increases, both candidates look similar to the voter in the probabilistic model but not in deterministic voting. The scenario where the voter is positioned to the right of 1 is symmetric.}
    \label{fig:prob_of_voting}
\end{figure}

We now define a general class of functions $\G$ for pairwise order probabilities in terms of the relative preference (ratio of distances) $r$. Formally, for $g \in \G$, we have the marginal probability $q^{\Rho(d)}(i,j,j')=g\left(d(i,j')/d(i,j)\right)$

\begin{definition}[Function class for pairwise order probability]\label{def:pairwise_prob_function_class}
   Let \( \G \) be a class of functions where each \( g \in \G \) is a mapping \( g: [0, \infty) \cup \{\infty\} \to [0,1] \) that satisfies the following properties.

\begin{enumerate}
\item $g$ is continuous and twice differentiable with $g'(r) > 0 \text{ }\forall r \in (0,\infty)$, i.e., $g(r)$ is strictly increasing on $[0,\infty)$. 
    \item Define $\frac{1}{r}$ as $\infty$ when $r =0$.
 We have $g(r) + g(\frac{1}{r}) = 1 ~~\forall r \geq 0$ with $g(0)=0$.   
    \item \( g_{\textsc{mid}}(x) \) is strictly convex on the open interval \( (0,1/2) \), and strictly concave on \( (1/2,1) \) with \( g_{\textsc{mid}}(x) \) regularly varying at 0 with index \(\beta>1\). 
\end{enumerate}
\end{definition}

%

The first condition is a smoothness-regularity strengthening of Strict Monotonicity (SM) [Axiom~\ref{axiom:monotonicity}]. The second condition imposes the natural reciprocity symmetry $q^{\Rho(d)}(i,j',j)=1-q^{\Rho(d)}(i,j,j')$. The boundary condition $g(0)=0$ ensures that candidates who overlap with a voter are ranked deterministically above other non-overlapping candidates.

The third condition is a technical shape and tail-regularity assumption used to simplify the analysis. 
The regular variation assumption is used in the definition of $\Gamma_{g,\eta}$ in \eqref{eq:Gamma_g_eta_defn}, while the convex-to-concave transition of $g_{\textsc{mid}}$ is used to establish the unique existence of maximizer in Lemma~\ref{lemma:existence_maxima}.

These conditions are satisfied by several standard comparison models, including the Bradley--Terry-type model $g(r)=\frac{1}{1+r^{-\theta}}; \theta>1$ for which the regular variation index is $\beta=\theta$.\footnote{When $\theta=1$, this model satisfies strict monotonicity but falls outside the standing tail-regularity assumption \(\beta>1\). The distortion guarantees can be recovered for this case by essentially the same calculations.}

For an arbitrary function \(g\in\G\), such pairwise marginals need not be realizable by a distribution over full rankings for every metric \(d\). Distributions on rankings that generate pairwise order probabilities according to the BT model are known from prior work \cite{azari2012random} and can be directly derived from the constructive definition in Definition~\ref{def:pl}.

We assume $g \in \G$ in the rest of the paper. Let $\M(\N \cup \A)$ denote the set of valid distance functions on $(\N, \A)$. For any $g$ and $d \in \M(\N \cup \A)$, let $\hat{\Rho}^{(g)}(d)$ denote the set of probability distributions on $\sigma_{\N}$ for which the marginal pairwise order probabilities are $g(\frac{d(i,j')}{d(i,j)}).$
That is,
\begin{align}{\label{eq:pairwise_probability}}
    \forall \Rho \in \hat{\Rho}^{(g)}(d), \sigma_{\N} \sim \Rho  \implies \mathbb{P}[A \succ_{\sigma_i} B] = g\left(\frac{d(i,B)}{d(i,A)}\right).
\end{align}
We assume that all voters vote independently of each other.
We now define metric distortion under probabilistic voting as a function of $g$ for a given $m$ and $n$.



\begin{definition}[Metric Distortion under Probabilistic Voting] \label{def:distortion-probabilistic}
%
\begin{align}{\label{eq:distortion_defn_1}}
    \textsc{dist}^{(g)}(f,n,m)  := \sup_{\substack{\N: |\N| =n \\ \A: |\A| =m}} \sup_{\substack{d \in \mathcal{M}(\N \cup \mathcal{A})\\ \hat{\Rho}^{(g)}(d)\neq\emptyset}} \sup_{\Rho \in \hat{\Rho}^{(g)}(d)} \frac{\mathbb{E}_{\sigma_{\N} \sim \Rho}[\SC(f(\sigma_{\N}),d)]}{\min\limits_{j \in \A}\SC(j,d)}.
\end{align}
%

The expectation is over both the randomness in the votes and the voting rule $f$.
\end{definition}
Since we focus on large elections, we define $\textsc{dist}^{(g)}$ as a function of the number of voters \(n\) and the number of candidates \(m\).

Note that distortion is defined as the supremum over all realizable metrics \(d\) and all distributions in $\hat{\Rho}^{(g)}(d)$; metrics for which $\hat{\Rho}^{(g)}(d)$ is empty are excluded from the supremum. Additionally, we provide bounds specifically for the PL model, where we modify the definition to account for the specific distribution of the vote profile.

As in Figure \ref{fig:prob_of_voting}, consider the 1-D Euclidean space with two candidates positioned at 0 and 1. The function \( g_{\textsc{mid}}(x) \) represents the probability that a voter located at distance \( x \) from 0 votes for the candidate at 1 when the voter is between 0 and 1.
%
%

\newcommand{\lemmauniquemaxima}
{
$\frac{g_{\textsc{mid}}(x)}{x}$ has a unique global maximum in $(0,1)$.
}

\begin{lemma}{\label{lemma:existence_maxima}}
    \lemmauniquemaxima
\end{lemma}


Denote the unique maximizers of $\frac{g_{\textsc{mid}}(x)}{x}$ by $x^{\ast}_{\textsc{mid}}$. Recall from Equation \eqref{eq:gm_go_defn} that $\gm$ is defined as $\gm = \frac{ g_{\textsc{mid}}(x^{\ast}_{\textsc{mid}}) }{x^{\ast}_{\textsc{mid}}}$. 

In the rest of the analysis, $\gm$ appears many times, so we note these quantities for the PL model here. For the PL model with $\theta =2$, $\gm = \frac{\sqrt{2}+1}{2} \approx 1.21,\overline L(g,1/2) \approx 0.40$ and $R(g,1/2)\approx 0.707$. When $\theta =4,$ $\gm \approx 1.42,$ and $R(g,0.5) = + \infty$ as there is no solution. When $\theta \rightarrow \infty$, $\gm \rightarrow 2$ and $R(g,0.5) = + \infty$. This limit is where PL resembles deterministic voting.

We now consider a modified probabilistic voting setup where the rankings are not generated adversarially subject to a constraint on the marginal distribution of pairwise probabilities but are generated directly from the PL model with parameter $\theta$.

\begin{definition}[Metric Distortion under PL model]
    \[
        \textsc{dist}^{\theta}_{\text{PL}}(f,n,m) = \sup_{\substack{\N: |\N| =n \\ \A: |\A| =m}} \sup_{d \in \mathcal{M}(\N \cup \mathcal{A})} \frac{\mathbb{E}_{\sigma_{\N} \sim \text{PL}(\theta)}[\SC(f(\sigma_{\N}),d)]}{\min\limits_{j \in \A}\SC(j,d)}.
    \]
    Here, $\text{PL}(\theta)$ denotes a ranking from the PL model with parameter $\theta$. The expectation is over both the randomness in the votes and the voting rule $f$.
\end{definition}

Observe that $\textsc{dist}^{\theta}_{\text{PL}}(f,n,m) \leq \textsc{dist}^{(g)}(f,n,m)$ where $g(r) = \frac{1}{1+r^{-\theta}}$ (drawn from the BT model).

In the following sections, we derive distortion bounds for various voting rules.

\section{Metric Distortion of \textsc{Plurality} Under Probabilistic Voting}\label{sec:metric_distortion_Plurality}

In this section, we establish upper and lower bounds on the metric distortion of \textsc{Plurality}. 
In the limit as the number of voters approaches infinity (``large election''), our upper and lower bounds match and scale linearly with the number of candidates $m.$
Let \( B \) be the candidate that minimizes social cost (referred to as the ``best'' candidate), and let \( \{A_j\}_{j \in [m-1]} \) be the set of all other candidates.

Before we state the main result for the upper and lower bounds of plurality, we first set some notations. Define 
\begin{equation}{\label{eq:R_C_defn}}
	R(g,\alpha) = \inf\limits_{(x_1, x_2) \in \mathcal{P}^{(g, \alpha)}} \frac{\C^{(g,\alpha)}(x_2)}{1- \C^{(g,\alpha)}(x_2)} \text{ where } \C^{(g,\alpha)}(x_2) = \frac{\alpha - g_{\textsc{mid}}(x_2) + x_2(1-2x_2) g'_{\textsc{mid}}(x_2)}{2(\alpha-g_{\textsc{mid}}(x_2)) +  (1-2x_2) g'_{\textsc{mid}}(x_2)}
\end{equation}
 
where the constraint set $\mathcal{P}^{(g,\alpha)}$ is defined below.\footnote{If the constraint set is empty, define \(R(g,\alpha)=+\infty\).}

 \begin{align}
 \mathcal{P}^{(g,\alpha)}
 :=
 \{(x_1,x_2)\in[0,1]\times(0,1/2):
 &\; g'_{\textsc{mid}}(x_1)
 =
 g'_{\textsc{mid}}(x_2)(1-2x_2)
 +2\left(\alpha-g_{\textsc{mid}}(x_2)\right),\nonumber\\
 &\; g_{\textsc{mid}}(x_1)-\alpha
 =
 g'_{\textsc{mid}}(x_1)
 \left(x_1-\C^{(g,\alpha)}(x_2)\right),\nonumber\\
 &\; x_2 < \C^{(g,\alpha)}(x_2),\quad
 x_1 > \C^{(g,\alpha)}(x_2),\quad
 0<\C^{(g,\alpha)}(x_2)<1/2\}.\label{eq:constraint_set_P_def}
 \end{align}
 
 Finally, define
\begin{equation}{\label{eq:L_upper_bound_defn}}
\overline{L}(g,\alpha)
:=
\min\left\{
\frac12,\,
\inf_{x\in \mathcal X(g,\alpha)}
\C^{(g,\alpha)}(x)
\right\} \text{ where }  \mathcal X(g,\alpha)
:=
\left\{
x\in(0,1/2):
g_{\textsc{mid}}(x)< \alpha
\right\}.
\end{equation}
	 The value \(1/2\) in the empty-set case is a harmless convention, since whenever \(\mathcal X(g,\alpha)\neq\emptyset\), the values \(\C^{(g,\alpha)}(x)\) lie in \((0,1/2)\).
	 
\[
\Delta_g(\alpha)
:=
\max\left\{
\frac{\gm}{\alpha}-1,\,
R(g,\alpha)^{-1}
\right\},
\qquad
\overline{\Delta}_g(\alpha)
:=
\max\left\{
\frac{\gm}{\alpha}-1,\,
\overline L(g,\alpha)^{-1}-1
\right\}.
\]
We will show below that \(\Delta_g(\alpha)=\Opt(\Efrac{\alpha})^{-1}\) for the values of \(\alpha\) used in the distortion bounds, and that \(\Delta_g(\alpha)\leq\overline{\Delta}_g(\alpha)\).
Geometrically, the two terms in \(\Delta_g(\alpha)\) correspond to two ways in which a fixed challenger \(A_j\) can obtain expected pairwise support against the socially optimal candidate \(B\). The term \(\gm/\alpha-1\) comes from a one-active-location obstruction: after adding inactive mass near the optimal side, the active voters lie between \(B\) and \(A_j\), and the adversary chooses the location that maximizes probability gained per unit distance from \(B\). The term \(R(g,\alpha)^{-1}\) comes from a genuinely two-location obstruction, where the active mass is split between two points on the line through \(B\) and \(A_j\), one on the segment and one on the outside extension. The two-point system defining \(R(g,\alpha)\) captures exactly when this mixed geometry is better for the adversary. The lower-bound statement in Theorem~\ref{theorem:thm_plurality_distortion_m_lower_bound}, and its proof in Appendix~\ref{sec:thm_plurality_distortion_m_lower_bound_proof}, realize these same one-location and two-location geometries, which is why the upper bound is tight in the large-election limit.

For finite-\(n\) bounds, define
\begin{equation}{\label{eq:Gamma_g_eta_defn}}
    \Gamma_g(\eta):=\sup_{0<t\leq\eta} t\,\Delta_g(t).
\end{equation}

This quantity is finite for the values of \(\eta\) used below. Indeed,
\[
t\,\Delta_g(t)
\leq
\max\left\{\gm-t,\frac{t}{R(g,t)}\right\}.
\]
The first term is bounded by \(\gm\). For the second term, Claim~\ref{claim:asymptotic_r_scale} gives \(R(g,t)=\Omega(g_{\textsc{mid}}^{-1}(t))\) as \(t\downarrow0\). Under the standing local regular-variation assumption on \(g_{\textsc{mid}}\) near zero, with index \(\beta>1\), the ratio \(t/g_{\textsc{mid}}^{-1}(t)\) is bounded near zero. Away from zero, the supremum is over a compact interval on which the two-point system is finite wherever it is active. Hence \(\Gamma_g(\eta)<\infty\).

\subsection{Upper bound on the metric distortion of \textsc{Plurality}}\label{sec:metric_distortion_Plurality_upper_bound}

\newcommand{\thmplulabel}{\label{eq:distortion_plurality_m}\gdef\thmplulabel{}}
\newcommand{\thmplu}{
For every \( \epsilon \in (0,1/2) \), and for all \( m \geq 2 \) and \( n \geq m^2 \), we have:
	    \begin{align}\thmplulabel
	        \textsc{dist}^{(g)}(\textsc{Plurality},n,m)
	        &\leq
	        m(m-1)\Gamma_g\left(\frac{1-n^{-(\frac{1}{2}-\epsilon)}}{m}\right)
	        \exp\Bigg(
	            \frac{- n^{(\frac{1}{2}+\epsilon)} +2m}
	                 {(2n^{(\frac{1}{2}-\epsilon)}-1)m}
	        \Bigg) \nonumber \\
	        &\quad
	        + \max\bigg(
	            \frac{m\gm}{1- n^{-(\frac{1}{2}-\epsilon)}}-1,
	            R\left(g,\frac{1- n^{-(\frac{1}{2}-\epsilon)}}{m}\right)^{-1}
	        \bigg). 
	    \end{align}
	    Furthermore, $\limsup_{n \to \infty} \textsc{dist}^{(g)}(\textsc{Plurality},n,m)
	    \leq
	    \Delta_g(1/m)
	    \leq
	    \overline{\Delta}_g(1/m)$.
	}

\begin{theorem}{\label{theorem:thm_plurality_distortion_m}}
\thmplu
\end{theorem}

\newcommand{\claimasymptoticrscale}{
As $m \to \infty$, we have $\left(R(g,1/m)\right) =\Omega\left(g_{\textsc{mid}}^{-1}(1/m)\right)$ asymptotically.
In particular, if \(g_{\textsc{mid}}(x)\sim c x^\beta\) as \(x\downarrow0\) for some \(c>0\), then \(\left(R(g,1/m)\right)^{-1}=O(m^{1/\beta})\) as \(m\to\infty\).
}

The following claim shows that the distortion grows at most linearly in $m$. 

\begin{claim}{\label{claim:asymptotic_r_scale}}
\claimasymptoticrscale
\end{claim}

For the $\text{PL}$ model with parameter $\theta$, \(g_{\textsc{mid}}\) is regularly varying with index $\theta$ and thus, \(\left(R(g,1/m)\right)^{-1}=O(m^{1/\theta})\). The proof of Claim~\ref{claim:asymptotic_r_scale} appears in Appendix~\ref{sec:asymptotic_r_scale_proof}.

Table~\ref{tab:comp_distortion_terms} illustrates the two competing terms in \(\Delta_g(\alpha)\), together with the single-point upper envelope \(\overline{\Delta}_g(\alpha)\), for the Plackett--Luce model. For fixed \(m\), the term \(m\gm-1\) dominates when \(\theta\) is larger, while the two-point term \(R(g,\alpha)^{-1}\) can dominate for smaller \(\theta\). As guaranteed by Lemma~\ref{lemma:optimizer_mu_alpha_combined}, the single-point quantity \(\overline L(g,\alpha)^{-1}-1\) upper-bounds \(R(g,\alpha)^{-1}\).

\begin{table}[t]
	\centering
	\begin{tabular}{c|ccc|ccc}
		\hline
		& \multicolumn{3}{c|}{\(m=2,\ \alpha=1/2\)}
		& \multicolumn{3}{c}{\(m=4,\ \alpha=1/4\)} \\
		$\theta$ & $\overline L^{-1}-1$ & $R^{-1}$ & $\gm/\alpha-1$
		& $\overline L^{-1}-1$ & $R^{-1}$ & $\gm/\alpha-1$ \\
		\hline
		1.2 & 2.12 & 1.46 & 1.10 & 3.23 & 3.18  & 3.21 \\
		1.5 & 1.73 & 1.43 & 1.24 & 2.44 & -- & 3.47 \\
		2   & 1.48 & 1.41 & 1.41 & 1.91 & -- & 3.83 \\
		4   & 1.20 & -- & 1.84 & 1.37 & -- & 4.68 \\
		\hline
	\end{tabular}
	\caption{Limiting distortion terms for the Plackett--Luce model. Here \(R^{-1}\) denotes \(R(g,\alpha)^{-1}\), \(\overline L^{-1}-1\) denotes \(\overline L(g,\alpha)^{-1}-1\), and ``--'' means that the set of equations corresponding to the two-point term has no solution.}
	\label{tab:comp_distortion_terms}
\end{table}



To prove this theorem, we first establish an upper bound on $\frac{\SC(W,d)}{\SC(B,d)}$ under the constraint that the expected fraction of voters ranking candidate $W$ over $B$ is at least \(\alpha\). This ratio helps bound the contribution of the non-optimal candidate $W$ to the metric distortion of $\textsc{Plurality}$. To express this bound as a function of \( \alpha \), we formulate the following optimization problem \eqref{eqn:optim_formulation_first}. The optimization program is defined for all \( \alpha \in [0,1] \). Let \(\Opt\) denote the optimal value of an optimization program.

\begin{equation}{\label{eqn:optim_formulation_first}}
\Efin{\alpha}=
    \left\{
    \begin{aligned}
   \text{minimize}  \quad & \frac{\sum_{i \in \mathcal{N}} b_i}{\sum_{i \in \mathcal{N}} w_i}\\
    \text{subject to} \quad & \sum_{i \in \mathcal{N}} g\left(\frac{b_i}{w_i}\right) \geq n\alpha\\
      & \max_{i \in \mathcal{N}}|w_i-b_i| \leq \min_{i\in \mathcal{N}} (w_i+b_i) \\
      & b_i, w_i \geq 0 \quad \forall i \in \mathcal{N}
\end{aligned}
\right.
\end{equation}

It is useful to view \eqref{eqn:optim_formulation_first} through voter types. After normalizing \(d(B,W)=1\), a voter type is a pair \((b,w)\) satisfying \(|b-w|\leq1\) and \(b+w\geq1\). Appendix~\ref{sec:two_active_voter_types_proof} gives the full fractional voter-type formulation and proves that its optimum is attained using at most two active voter types. We therefore use \(\Efrac{\alpha}\) for this fractional relaxation, written below in its equivalent two-location form.\footnote{For the objective ratio in \(\Efrac{\alpha}\), a positive numerator divided by zero is interpreted as \(+\infty\), and \(0/0\) is interpreted as \(1\).} The variable \(p\) in \(\Efrac{\alpha}\) is a real mixture weight, not necessarily an integer multiple of \(1/n\). Every feasible finite electorate for \(\Efin{\alpha}\) induces a feasible fractional voter-type distribution, so $\Opt(\Efrac{\alpha}) \leq \Opt(\Efin{\alpha})$
\begin{equation}{\label{eqn:optim_formulation_two_type}}
    \Efrac{\alpha}=
    \left\{
    \begin{aligned}
    \text{minimize} \quad
        & \frac{p b_1+(1-p)b_2}{p w_1+(1-p)w_2} \\
    \text{subject to} \quad
        & p\,g\left(\frac{b_1}{w_1}\right)
          +(1-p)g\left(\frac{b_2}{w_2}\right) \geq \alpha,\\
        & |b_k-w_k|\leq 1,\quad b_k+w_k\geq 1 \quad k=1,2,\\
        & p\in[0,1],\quad b_k,w_k\geq0 \quad k=1,2.
    \end{aligned}
    \right.
\end{equation}

\begin{observation}{\label{lemma:bounding_social_cost_ratio}}
For any two candidates \(W,B\in\mathcal A\) satisfying
\(\sum_{i\in\mathcal N}\mathbbm P[W\succ_{\sigma_i}B]\geq n\alpha\), we have
\[
    \frac{\SC(W,d)}{\SC(B,d)}
    \leq
    \frac{1}{\Opt(\Efin{\alpha})}
    \leq
    \frac{1}{\Opt(\Efrac{\alpha})}.
\]
\end{observation}

\begin{proof}
Variables \(b_i\) and \(w_i\) in \eqref{eqn:optim_formulation_first} represent the distances \(d(i,B)\) and \(d(i,W)\). The constraint \(\max_i|w_i-b_i|\leq \min_i(w_i+b_i)\) follows from the triangle inequalities
\[
    |d(i,B)-d(i,W)|
    \leq d(B,W)
    \leq d(i,B)+d(i,W)
    \quad \forall i\in\mathcal N.
\]
Also, \(g(b_i/w_i)=\mathbbm P[W\succ_{\sigma_i}B]\). Thus every such candidate pair gives a feasible solution of the finite-voter program \eqref{eqn:optim_formulation_first}, and therefore
$\Opt(\Efin{\alpha})
    \leq
    \frac{\SC(B,d)}{\SC(W,d)}$
which gives the first inequality after taking reciprocals. The second inequality follows from
\(\Opt(\Efrac{\alpha})\leq\Opt(\Efin{\alpha})\).
\end{proof}

To obtain the needed reciprocal bound from Observation~\ref{lemma:bounding_social_cost_ratio}, we linearize the objective in \eqref{eqn:optim_formulation_two_type} by introducing an extra non-negative parameter \(\mu\). The following program is defined for all \(\mu \in [0,1]\)\footnote{The restriction \(\mu\leq1\) is without loss for the values \(\alpha\leq1/2\) used below: setting \(b_k=w_k\) gives objective value \(1\) and satisfies the probability constraint since \(g(1)=1/2\).} and \(\alpha \in [0,1]\).
\begin{equation}{\label{eqn:optim_formulation_final}}
    \Efrac{\mu,\alpha}=
    \left\{
    \begin{aligned}
    \text{minimize} \quad
        & p(b_1-\mu w_1)+(1-p)(b_2-\mu w_2)\\
    \text{subject to} \quad
        & p\,g\left(\frac{b_1}{w_1}\right)
          +(1-p)g\left(\frac{b_2}{w_2}\right) \geq \alpha,\\
        & |b_k-w_k|\leq 1,\quad b_k+w_k\geq 1 \quad k=1,2,\\
        & p\in[0,1],\quad b_k,w_k\geq0 \quad k=1,2.
\end{aligned}
\right. 
\end{equation}

Define
\[
    \mu^\ast(\alpha)
    :=
    \sup\left\{
        \mu\in[0,1]:
        \Opt(\Efrac{\mu,\alpha})\geq 0
    \right\}.
\]

\newcommand{\lemmasecondopt}{
For \(\alpha\leq 1/2\), we have $Opt(\Efrac{\alpha})= \mu^\ast(\alpha)$.
In the explicit solution below, the critical value is obtained by solving the threshold condition
\(\Opt(\Efrac{\mu,\alpha})=0\), whenever the threshold is attained.
}

\begin{observation}{\label{lemma:bouding_second_optimization_problem}}
\lemmasecondopt
\end{observation}

\begin{proof}
Let
\[
    B:=p b_1+(1-p)b_2,
    \qquad
    W:=p w_1+(1-p)w_2 .
\]
The feasible sets of \(\Efrac{\alpha}\) and \(\Efrac{\mu,\alpha}\) are the same. If
\(\Opt(\Efrac{\mu,\alpha})\geq0\), then every feasible solution of \(\Efrac{\alpha}\) satisfies
\[
    B-\mu W\geq0,
\]
and hence \(B/W\geq\mu\) whenever \(W>0\). If \(W=0\), then the ratio objective is either \(+\infty\) or, in the \(0/0\) case, equals \(1\) by convention; since \(\mu\leq1\), the same conclusion holds. Thus
\(\Opt(\Efrac{\alpha})\geq\mu\), so
\[
    \Opt(\Efrac{\alpha})\geq \mu^\ast(\alpha).
\]
Conversely, consider any admissible \(\mu\in[0,1]\) with \(\mu>\Opt(\Efrac{\alpha})\). By the definition of infimum, there exists a feasible solution of \(\Efrac{\alpha}\) with objective value less than \(\mu\). This solution must have \(W>0\), since feasible solutions with \(W=0\) have ratio \(+\infty\) or \(1\) under the convention, and \(\mu\leq1\). Hence \(B/W<\mu\), and for this same feasible solution,
\[
    B-\mu W<0,
\]
which implies \(\Opt(\Efrac{\mu,\alpha})<0\). Therefore no such \(\mu\) belongs to the set defining \(\mu^\ast(\alpha)\). Taking the supremum over admissible \(\mu\in[0,1]\) gives
\[
    \mu^\ast(\alpha)\leq \Opt(\Efrac{\alpha}).
\]
Combining both inequalities proves the equality.
\end{proof}

We now state the following lemma which characterises $\mu^\ast(\alpha)$. A detailed proof is provided in Appendix~\ref{sec:optimizer_e_alpha_soln}.

\newcommand{\optimizerepsilonalpha}{
\[
\mu^\ast(\alpha)
=
\min\left\{
R(g,\alpha),
\left(\frac{\gm}{\alpha}-1\right)^{-1}
\right\}
\ge
\min\left\{
\frac{\overline{L}(g,\alpha)}{1-\overline{L}(g,\alpha)},
\left(\frac{\gm}{\alpha}-1\right)^{-1}
\right\}.
\]
Furthermore, \(\left(\frac{\gm}{\alpha}-1\right)^{-1}\) dominates this lower bound whenever $\alpha \le \gm\,\overline{L}(g,\alpha)$.
}

\begin{lemma}\label{lemma:optimizer_mu_alpha_combined}
		\optimizerepsilonalpha
\end{lemma}

The two alternatives in Lemma~\ref{lemma:optimizer_mu_alpha_combined} are the formal optimizer counterparts of the geometric picture above: \(\left(\frac{\gm}{\alpha}-1\right)^{-1}\) is the one-active-location value, while \(R(g,\alpha)\) is the value of the two-active-location system. The reciprocal form in \(\Delta_g(\alpha)\) is the corresponding upper bound on the social-cost ratio for a challenger \(A_j\) against \(B\).

Under the Plackett--Luce model with parameter \(\theta\) and \(m=4\) (\(\alpha=1/4\)), the condition \(\alpha \le \gm\,\overline{L}(g,\alpha)\) is satisfied for \(\theta \geq 1.21\). For \(m=2\), the corresponding threshold is \(\theta \geq 2.11\). Table~\ref{tab:comp_distortion_terms} reports the relevant quantities for representative values of \(\theta\).

We now provide an intuition for the proof of Theorem~\ref{theorem:thm_plurality_distortion_m}, with the full proof presented in Appendix~\ref{sec:thm_plurality_distortion_m_proof}.  
The key technique involves analyzing two cases for each candidate \( A_j \neq B \), based on whether the expected number of voters ranking \( A_j \) above \( B \) (denoted as \( \alpha_j n \)) exceeds 
$\left(\frac{n}{m} - \frac{n^{\epsilon + 1/2}}{m}\right).$ 

\textbf{Case 1:} If \( \alpha_j n \geq \frac{n}{m} - \frac{n^{\epsilon + 1/2}}{m}\), we apply Observation~\ref{lemma:bounding_social_cost_ratio} to bound the ratio of social costs between \( A_j \) and \( B \), which in turn provides a bound on \( A_j \)'s contribution to the distortion.  

\textbf{Case 2:} If \( \alpha_j n  < \frac{n}{m} - \frac{n^{\epsilon + 1/2}}{m}\), we use the Chernoff bound on the probability of \( A_j \) being the winner. We then multiply this probability by the social cost ratio of \( A_j \) and \( B \) to bound the distortion.  

\subsection{Lower bound on the metric distortion of \textsc{Plurality}}


We now establish a matching lower bound on the metric distortion of \textsc{Plurality} for any \( m \) in the limit as \( n \to \infty \). A full proof is provided in Appendix~\ref{sec:thm_plurality_distortion_m_lower_bound_proof}.

\newcommand{\thmplulb}{
For every \(m\geq2\), $\liminf_{n \to \infty}
\textsc{dist}^{(g)}(\textsc{Plurality},n,m)
\geq
\Delta_g(1/m)$.
}

\begin{theorem}{\label{theorem:thm_plurality_distortion_m_lower_bound}}
\thmplulb
\end{theorem}

\begin{proof}[Proof Sketch]
The proof is by a single optimizer-based construction in a 3-D Euclidean space. Candidate \(W\) is positioned at \((1,0,0)\), while the other \(m-1\) candidates \(B_1,\ldots,B_{m-1}\) are placed on a circle of radius \(\epsilon\) around the origin in the \(y\)-\(z\) plane.

Take an approximately optimal two-type solution \((p,b_1,w_1,b_2,w_2)\) of \(\Efrac{1/m}\). The appendix shows that such active types may be chosen on the line through \(W\) and the center of the \(B_j\)'s, so we place type-\(k\) voters with distance \(b_k\) from that center and distance \(w_k\) from \(W\). A vanishing fraction \(\beta_\zeta=\zeta/(m-1)\) of voters is placed at \(W\); the remaining voters are split according to \(p\) and \(1-p\) between the two active types. This raises the limiting expected plurality score of \(W\) to at least \((1+\zeta)/m\), so \(W\) wins with probability tending to one. Letting \(n\to\infty\), \(\epsilon\to0\), the approximation error go to zero, and then \(\zeta\to0\), the social-cost ratio tends to \(\Opt(\Efrac{1/m})^{-1}=\Delta_g(1/m)\). \qedhere
\end{proof}

This result shows that even under probabilistic voting, the metric distortion of \textsc{Plurality} grows linearly with $m$, making it a poor choice when $m$ is even moderately large.

Theorem \ref{theorem:thm_plurality_distortion_m_lower_bound} also has the following implication for \textit{any} deterministic rule since any well-behaved deterministic voting rule cannot achieve a lower distortion than by selecting the majority winner when \( m = 2 \). We define a well-behaved voting rule and prove Observation~\ref{obs:generic_lb} in Appendix \ref{sec:proof_obs_generic_lb}.

	 \newcommand{\obsgenericlb}{
	 For the case of two candidates, the metric distortion of any well-behaved deterministic voting rule (Definition~\ref{def:wellbehaved}), for large elections ($n \rightarrow \infty$), is at least \(\Delta_g(1/2)\). 
	 }

 \begin{observation}\label{obs:generic_lb}
     \obsgenericlb
 \end{observation}


Note, however, that the lower bound in Theorem \ref{theorem:thm_plurality_distortion_m_lower_bound} does not apply to the PL model. This is because the PL model has a specific distribution over rankings, whereas Theorem \ref{theorem:thm_plurality_distortion_m_lower_bound} is for an \textit{adversarial} distribution subject to the constraint that its marginals on pairwise orders are given by $g$. In the PL model, $\mathbb{P}[A_j \text{~is top-ranked in~} \sigma_i]  = \frac{d(i,A_j)^{-\theta}}{\sum_{A_k \in \A} d(i,A_k)^{-\theta} }$ \cite{azari2012random}. 
Nonetheless, even under the PL model, we can show that the metric distortion of \textsc{Plurality} increases linearly with \( m \).


\newcommand{\thmplulbpl}{
For all $m\geq 2, ~~ \liminf_{n \to \infty} \textsc{dist}^{\theta}_{\textsc{PL}}(\textsc{Plurality},n,m) = \Omega(m). $
}

\begin{theorem}\label{theorem:thm_plurality_distortion_m_PL_lower_bound}
\thmplulbpl
\end{theorem}

\section{Metric Distortion of \textsc{Copeland} and \textsc{WeightedUncovered} Under Probabilistic Voting}\label{sec:distortion_Copeland}

We now derive an upper bound on the metric distortion of \textsc{Copeland}. Recall that a Copeland winner belongs to the uncovered set in the tournament graph \cite[Theorem 15]{anshelevich2015approximating}. This implies that for a socially optimal candidate \(B\) and Copeland winner \(W\), one of the following two scenarios holds:
\begin{itemize}
    \item the oriented edge from \(W\) to \(B\) is present;
    \item \(W\) has an oriented edge to some candidate \(Y\), and \(Y\) in turn has an oriented edge to \(B\). 
\end{itemize}

In the first scenario, the distortion bound for \textsc{Plurality} with \(m=2\) applies to \textsc{Copeland} for any \(m\). The analysis for the second scenario is more nuanced because the two pairwise-majority constraints are imposed on the same electorate. We therefore use a three-candidate optimization problem rather than multiplying two independent one-edge bounds.

As in the proof of \textsc{Plurality}, it is useful to work with voter types. Here a voter type is a triple \((b,y,w)\), representing the distances of the voter to \(B,Y,W\), respectively. For \(\alpha_1,\alpha_2\in[0,1]\), define the three-type program \(\mathcal{T}_{\alpha_1,\alpha_2}\) as follows:
\begin{equation}{\label{eqn:copeland_three_candidate_fractional}}
\mathcal{T}_{\alpha_1,\alpha_2}
=
    \left\{
    \begin{aligned}
    \text{minimize}\quad
        & \frac{\sum_{\ell=1}^{3}p_\ell b_\ell}
        {\sum_{\ell=1}^{3}p_\ell w_\ell}\\
    \text{subject to}\quad
        & \sum_{\ell=1}^{3}p_\ell g\left(\frac{y_\ell}{w_\ell}\right)\geq \alpha_1,\\
        & \sum_{\ell=1}^{3}p_\ell g\left(\frac{b_\ell}{y_\ell}\right)\geq \alpha_2,\\
        & (b_\ell,y_\ell,w_\ell)_{\ell=1}^{3}
        \text{ are jointly feasible three-candidate voter types},\\
        & \sum_{\ell=1}^{3}p_\ell=1,\quad p_\ell\geq0 \quad \ell\in\{1,2,3\}.
    \end{aligned}
    \right.
\end{equation}
The first probability constraint corresponds to the expected fraction of voters ranking \(W\) above \(Y\), and the second corresponds to the expected fraction ranking \(Y\) above \(B\). Joint feasibility means that the three distance triples can be realized with common candidate-candidate distances \(d(B,Y)\), \(d(Y,W)\), and \(d(B,W)\) satisfying the triangle inequalities. The weights \(p_1,p_2,p_3\) are three disjoint voter blocks, one for each active voter type. Appendix~\ref{sec:copeland_three_active_voter_types_proof} gives the full finite-electorate formulation and proves, using the same basic-feasible-solution argument as for \textsc{Plurality}, that the fractional relaxation can be represented using at most three active voter types.

\newcommand{\obscopelandthreecandidateratio}{
If candidates \(W,Y,B\in\A\) satisfy $\sum_{i\in\N}\mathbbm P[W\succ_{\sigma_i}Y]\geq n\alpha_1$ and 
    $\sum_{i\in\N}\mathbbm P[Y\succ_{\sigma_i}B]\geq n\alpha_2$ then, $\frac{\SC(W,d)}{\SC(B,d)} \leq
    \frac{1}{\Opt(\mathcal{T}_{\alpha_1,\alpha_2})}$.
}

\begin{observation}{\label{obs:copeland_three_candidate_ratio}}
\obscopelandthreecandidateratio
\end{observation}

\newcommand{\claimcopelandprogramproductlowerbound}{
For all \(\alpha_1,\alpha_2\in[0,1]\),
$\Opt(\mathcal{T}_{\alpha_1,\alpha_2})
    \geq
    \Opt(\Efrac{\alpha_1})\Opt(\Efrac{\alpha_2})$. Consequently, $\Opt(\mathcal{T}_{\alpha_1,\alpha_2})^{-1}
    \leq
    \Opt(\Efrac{\alpha_1})^{-1}\Opt(\Efrac{\alpha_2})^{-1}$.
}

\begin{claim}{\label{claim:copeland_program_product_lower_bound}}
\claimcopelandprogramproductlowerbound
\end{claim}
The proof is in Appendix~\ref{sec:copeland_three_active_voter_types_proof}. The claim shows that treating the two pairwise constraints independently gives the product-type Copeland bound as a relaxation of the shared-electorate three-candidate program.

Let \(\eta_n:=(1-n^{-(1/2-\epsilon)})/2\). Formally, we have the following result.

\newcommand{\thmcop}{
		For every \( \epsilon \in (0,1/2) \), \( m \geq 2 \), and \( n \geq 16 \), let \(\eta_n=(1-n^{-(1/2-\epsilon)})/2\). Then  
	\begin{align*} 
	    \textsc{dist}^{(g)}(\textsc{Copeland},n,m) 
		    &\leq  4m(m-1)
		    \exp\left(
	            \frac{- n^{(\frac{1}{2}+\epsilon)} +8}
	                 {2(2n^{(\frac{1}{2}-\epsilon)}-1)}
			    \right)
			    \max\left\{\Gamma_g(\eta_n),\Gamma_g(\eta_n)^2\right\}  \nonumber \\
		    &\quad +\max\left\{
		    \Opt(\Efrac{\eta_n})^{-1},
	        \Opt(\mathcal{T}_{\eta_n,\eta_n})^{-1}
		    \right\}.  
		\end{align*}

	For every \( m \geq 2 \), we have  
	\[
	\begin{aligned}
	\limsup_{n \to \infty} \textsc{dist}^{(g)}(\textsc{Copeland},n,m)
	&\leq
	\max\left\{
	    \Delta_g(1/2),
	    \Opt(\mathcal{T}_{1/2,1/2})^{-1}
	\right\} \\
	&\leq
	\max\left\{
	    \Delta_g(1/2),
	    \Delta_g(1/2)^2
	\right\},
	\end{aligned}
	\]
	where \(\Delta_g(\alpha)=\Opt(\Efrac{\alpha})^{-1}\), as computed in Lemma~\ref{lemma:optimizer_mu_alpha_combined}.
	}

\begin{theorem}{\label{theorem:Copeland_distrotion_m}}
\thmcop
\end{theorem}

\begin{proof}[Proof Sketch]

We analyze two cases for each candidate \( A_j \neq B \) and define the event \( E_j \) as the existence of a directed path of length at most two from \( A_j \) to \( B \) in the tournament graph, where the expected number of votes on all edges exceeds  
$\frac{n}{2} - \frac{n^{(1/2+\epsilon)}}{2}.$

If \( E_j \) holds through a direct edge, we use Observation~\ref{lemma:bounding_social_cost_ratio}. If it holds through a two-hop path, we use Observation~\ref{obs:copeland_three_candidate_ratio} and then Claim~\ref{claim:copeland_program_product_lower_bound}. Otherwise, we apply the union bound and Chernoff's bound to upper bound the probability of \( A_j \) winning. Finally, multiplying this probability bound with the social-cost bound yields the stated distortion bound. A detailed proof is provided in Appendix~\ref{sec:Copeland_distrotion_m_proof}.
\end{proof}

\newcommand{\claimcopelandlowerbound}{
\[
    \liminf_{n\to\infty}
    \textsc{dist}^{(g)}(\textsc{Copeland},n,3)
    \geq
    \Opt(\mathcal{T}_{1/2,1/2})^{-1}.
\]
}

\begin{claim}[Copeland lower bound]\label{claim:copeland_lower_bound}
\claimcopelandlowerbound
\end{claim}

The construction is taken directly from an approximately optimal solution of \(\mathcal{T}_{1/2+\zeta,1/2+\zeta}\). We realize its three active voter types as three disjoint voter blocks around candidates \(B,Y,W\). The positive slack \(\zeta\) ensures that \(W\) beats \(Y\) and \(Y\) beats \(B\) with probability tending to one, while the social-cost ratio converges to the reciprocal of the program objective. Appendix~\ref{sec:copeland_lower_bound_proof} gives the rounding, right-continuity, and limiting details.

The value of the new three-candidate program is the remaining object to solve if one wants a closed-form Copeland bound. Recall that the metric distortion of \textsc{Copeland} in the deterministic model is 5 \cite{anshelevich2015approximating}; recovering the right deterministic limit likely requires a sharper geometric analysis of the three-candidate program, which we leave for future work.

We now show that this technique extends seamlessly to \textsc{WeightedUncovered}. Recall from Definition~\ref{def:wu} that for a winner \( W \), either \( |WB| \geq n(1-\varphi) \) or there exists a candidate \( Y \) such that \( |WY| \geq n(1-\varphi) \) and \( |YB| \geq \varphi n \). Here, \(\varphi=(\sqrt{5}-1)/2\) is the inverse golden ratio.
At a high level, the direct-edge case is controlled by \(\mathcal{E}_{1-\varphi}\), while the two-hop case is controlled by the shared-electorate program \(\mathcal{T}_{1-\varphi,\varphi}\). Claim~\ref{claim:copeland_program_product_lower_bound} then gives the simpler product relaxation
\[
    \Opt(\mathcal{T}_{1-\varphi,\varphi})^{-1}
    \leq
    \Opt(\Efrac{1-\varphi})^{-1}\Opt(\Efrac{\varphi})^{-1}.
\]
Formally, we establish the following result:

\begin{corollary}{\label{corollary:weighted_uncovered_set}}
	     For every $m\geq 2$, we have 
	     \begin{align*}
	         \limsup_{n \to \infty}
	         \textsc{dist}^{(g)}(\textsc{WeightedUncovered},n,m)
	         &\leq
	         \max\left\{
	             \Opt(\Efrac{1-\varphi})^{-1},
	             \Opt(\mathcal{T}_{1-\varphi,\varphi})^{-1}
	         \right\} \\
	         &\leq
	          \max\left\{
	              \Opt(\Efrac{1-\varphi})^{-1},
	              \Opt(\Efrac{1-\varphi})^{-1}\Opt(\Efrac{\varphi})^{-1}
	          \right\}.
		     \end{align*}
\end{corollary}



Similarly, we can derive a lower bound for the entire family of voting rules that select an arbitrary candidate from the \( \lambda \)-weighted uncovered set for any \( \lambda \in [0.5, 1] \) by replacing \( \varphi \) with \( \lambda \). Extending our analysis to more voting rules from the class of weighted or unweighted tournament graph-based rules is an interesting direction for future research.

\section{Metric Distortion of \textsc{Random Dictator} Under Probabilistic Voting}


We now study a randomized voting rule that picks a random voter and chooses their favorite candidate as the winner. This simple scheme has a metric distortion of 3 in the classical model. However, under probabilistic voting, its metric distortion can be much larger. We first give an upper bound, which is linear in \(m\).

\newcommand{\thmrdub}{
    $\textsc{dist}^{(g)}(\textsc{Random Dictator},n,m) \leq (m-1)\gm +1.$
}

\begin{theorem}\label{theorem:RD_distortion_upper_bound}
\thmrdub
\end{theorem}

The proof is in Appendix \ref{sec:RD_distortion_upper_bound_proof}. 
We now give a lower bound with an example.

\begin{theorem}{\label{theorem:RD_distortion_lower_bound}}
For $m\geq 3$ and $n \geq 2,$ we have: 
    $\textsc{dist}^{(g)}(\textsc{Random Dictator},n,m) \geq  2+\frac{1}{g^{-1}(\frac{1}{m-1})} - \frac{2}{n}. $
\end{theorem}
\begin{proof}
We consider an example in the 1-D Euclidean space. Let \( B \) be positioned at \( 0 \), and all other candidates in \( \mathcal{A} \setminus \{B\} \) be positioned at \( 1 \). There are \( n-1 \) voters at \( 0 \) and one voter \( V \) at \( \tilde{x} = \frac{g^{-1}(\frac{1}{m-1})}{1 + g^{-1}(\frac{1}{m-1})} \).
The ranking for \( V \) is generated as follows:
\begin{enumerate}
    \item Select a candidate from \( \mathcal{A} \setminus \{B\} \) as the top choice uniformly at random.
    \item Place \( B \) in the second position.
    \item Randomly permute the remaining candidates in the remaining ranks.
\end{enumerate}
The marginal pairwise order probabilities are consistent with the distances of \( V \) from \( B \) and from each candidate in \( \mathcal{A} \setminus \{B\} \), together with the function \(g(\cdot)\). In particular, \(g\left(\frac{\tilde{x}}{1- \tilde{x}}\right) = \frac{1}{m-1}\).
The distortion for this instance is given by:  
\begin{align}
& \mathbb{P}[B \text{~wins}] \cdot 1
  + \mathbb{P}[B \text{~loses}] \cdot \frac{n-\tilde{x}}{\tilde{x}} \nonumber\\
&\qquad =
\frac{n-1}{n} + \frac{1}{n} \frac{n-\tilde{x}}{\tilde{x}}
 = 1+ \frac{1}{\tilde{x}} - \frac{2}{n}
 = 2+\frac{1}{g^{-1}(\frac{1}{m-1})} - \frac{2}{n}.
\end{align}
\end{proof}
For example, for \(g(r)=1/(1+r^{-\theta})\), we have \(g^{-1}(t)=\left(\frac{t}{1-t}\right)^{1/\theta}\). Hence \(g^{-1}\left(\frac{1}{m-1}\right)=(m-2)^{-1/\theta}\), so
\[
    \textsc{dist}^{(g)}(\textsc{Random Dictator},n,m)
    \geq
    2 + (m-2)^{1/\theta} -\frac{2}{n},
\]
and therefore
\[
    \lim_{n\rightarrow \infty}
    \textsc{dist}^{(g)}(\textsc{Random Dictator},n,m)
    \geq
    2 + (m-2)^{1/\theta}.
\]


However, this result does not apply to the PL model as the PL model has a specific distribution over rankings, whereas the above result is derived using an adversarial distribution subject to constraints on the marginals of pairwise order relations. 
Nonetheless, we can obtain the same order-wise lower bound for the PL model as well. 

\begin{theorem}{\label{theorem:RD_distortion_lower_bound_PL}}
\(\liminf_{n \rightarrow \infty} \textsc{dist}_{PL}^{\theta}(\textsc{Random Dictator},n,m) \geq  1+ \frac{(m-1)^{1/\theta}}{2}\).
\end{theorem}
\begin{proof}
   We have the same example in a 1-D Euclidean space as for the proof of Theorem~\ref{theorem:RD_distortion_lower_bound}.  Let $B$ be positioned at $0$ and all other candidates $\A \setminus \{B\}$ be positioned at 1. $n-1$ voters are at $0$, and one voter is at $t.$ We will set $t$ later by optimizing for the distortion. Recall that in the PL model, $\mathbb{P}[A_j \text{~is top-ranked in~} \sigma_i]  = \frac{d(i,A_j)^{-\theta}}{\sum_{A_k \in \A} d(i,A_k)^{-\theta} }$ \cite{azari2012random}. 
   The metric distortion for this instance is
   \[
   \begin{aligned}
   \mathbb{P}[B \text{~wins}] \cdot 1 + \mathbb{P}[B \text{~loses}] \cdot \frac{n-t}{t}
   &=  \frac{n-1}{n} +  \frac{1}{n} \frac{t^{-\theta}}{t^{-\theta} + (m-1) (1-t)^{-\theta}} \\
   &\quad +  \frac{1}{n}  \frac{(m-1) (1-t)^{-\theta}}{t^{-\theta} + (m-1) (1-t)^{-\theta}}  \frac{n-t}{t}.
   \end{aligned}
   \]

   To analyze the limit as \( n \to \infty \), we drop terms of order \( O(1/n) \) to obtain  
\( 1 + \frac{(m-1) (1-t)^{-\theta}}{t(t^{-\theta} + (m-1) (1-t)^{-\theta})} \).  
This simplifies to  
\( 1 + \frac{(m-1) t^{\theta-1}}{(1-t)^{\theta} + (m-1)t^{\theta}} \),  
which is further lower-bounded by  
\( 1 + \frac{(m-1) t^{\theta-1}}{1 + (m-1)t^{\theta}} \).  
Setting \( t = (m-1)^{-1/\theta} \), we obtain a metric distortion lower bound of  
\( 1 + \frac{(m-1)^{1/\theta}}{2} \).
\end{proof}

\section{Metric Distortion of \textsc{Plurality Veto} Under the PL Model}{\label{sec:Plurality_Veto}}

We now study the \textsc{Plurality Veto} voting rule \cite{kizilkaya2022plurality}, which attains the optimal distortion of 3 among deterministic voting rules. Under probabilistic voting with the PL model, however, we show that its metric distortion scales with the number of candidates \(m\). We restrict to the PL model for simplicity, but the construction can be adapted to other functional forms of \(g(\cdot)\).

For clarity, we recall the rule from \cite[Section 3]{kizilkaya2022plurality}. In the first stage, \textsc{Plurality Veto} assigns an initial score of $\textsc{plu}(c)$ to each candidate $c$, where $\textsc{plu}(c)$ is the number of voters who rank $c$ first. In the second stage, these scores are gradually decreased. When the score of a candidate \(c\) reaches zero, \(c\) is eliminated. Voters are processed one by one in an arbitrary order, which may be fixed in advance or adaptive to voters' preferences. When a voter \(v\) is processed, the rule decrements the score of \(v\)'s bottom choice among the candidates not yet eliminated. A detailed description and pseudocode appear in \cite[Algorithm 1]{kizilkaya2022plurality}. 

We now present a lower bound on the distortion of \textsc{Plurality Veto} by constructing an appropriate metric space.

\begin{theorem}{\label{thm:plurality_veto_lb}}
    For every fixed \(\theta>1\), there exists a constant \(c_\theta>0\) such that, for all sufficiently large \(m\),
    \[
    \liminf_{n \to \infty} \textsc{dist}^{\theta}_{\textsc{PL}}(\textsc{Plurality Veto},n,m)
    \geq
    c_\theta m^{\frac{1}{\theta+1}}.
    \]
\end{theorem}

\begin{proof}
    Consider the following metric space on the real line. The social-cost-minimizing candidate \(B\) is placed at \(x=0\), and all remaining candidates \(\{W_1,\ldots,W_{m-1}\}\) are placed at \(x=1\). A fraction \(p\) of the voters is placed at \(x=1\), and the remaining voters are placed at \(x=m^{-a}\). We set \(p=m^{-a}\) and choose \(a\) below. Let \(\top(B)\) denote the set of voters who rank \(B\) first, and let \(\bot(B)\) denote the set of voters who rank \(B\) last. Under probabilistic voting with the PL model, both sets are random. Recall that
    \[
        \mathbb{P}[A_j \text{ is top-ranked in } \sigma_i]
        =
        \frac{d(i,A_j)^{-\theta}}{\sum_{A_k \in \A} d(i,A_k)^{-\theta}}
    \]
    \cite{plackett1975analysis}. We thus have

    \begin{align}{\label{eq:vote_for_against_G}}
        \frac{\mathbb{E}[|\top(B)|]}{n}
        =
        \frac{(1-m^{-a}) (m^{-a})^{-\theta}}
             {(m^{-a})^{-\theta} + (m-1) (1-m^{-a})^{-\theta}},
        \qquad
        \frac{\mathbb{E}[|\bot(B)|]}{n} \geq m^{-a}.
    \end{align}

    
    Since each voter votes independently, we can view $|\top(B)|$ and $|\bot(B)|$ as sums of independent random variables; hence their variances are bounded by \(n\).

    Applying Cantelli's inequality, it suffices to ensure that \(\frac{1}{n} \mathbb{E}[|\bot(B)|] \geq \frac{1}{n} \mathbb{E}[|\top(B)|] + c_m\) (for some constant \( c_m \) independent of \( n \)) in order to guarantee that \( |\bot(B)|>|\top(B)| \) with constant probability. We now observe why this implies that \(B\) is not the \textsc{Plurality Veto} winner. Algorithm~1 of \cite{kizilkaya2022plurality} processes all \(n\) voters and returns the candidate whose score is decremented in the final iteration. Candidate \(B\)'s initial score is \(|\top(B)|\). Moreover, every voter in \(\bot(B)\), when processed while \(B\) still has positive score, decrements \(B\)'s score; if such a voter does not decrement \(B\), then \(B\)'s score has already reached zero. Hence, if \(|\bot(B)|>|\top(B)|\), candidate \(B\) reaches score zero before the end of the process and cannot be the final returned candidate.

    Now choose \(a = \frac{1}{1+\theta}\), and write \(s=m^{-a}\). Then the top-vote fraction in expectation is
    \[
        \frac{\mathbb{E}[|\top(B)|]}{n}
        =
        \frac{(1-s)s^{-\theta}}
             {s^{-\theta}+(m-1)(1-s)^{-\theta}}
        =
        \frac{1-s}
             {1+(m-1)(s/(1-s))^{\theta}}.
    \]
    Since \(m=s^{-(1+\theta)}\), the final expression is \(s-\Omega(s^2)\) as \(s\downarrow0\). On the other hand, \(\mathbb{E}[|\bot(B)|]/n\geq s\). Thus, for all sufficiently large \(m\), there is a positive number \(c_m\), independent of \(n\), such that
    \[
        \frac{1}{n}\mathbb{E}[|\bot(B)|]
        -
        \frac{1}{n}\mathbb{E}[|\top(B)|]
        \geq c_m.
    \]
    The ratio of social costs may be bounded as

    \begin{equation}
        \frac{\textsc{SC} (W_i,d)}{\textsc{SC} (B,d)}
        =
        \frac{(1-m^{-a})(1-m^{-a})}{m^{-a}(1-m^{-a}) + m^{-a}}
        =
        \frac{(m^a-1)^2}{2m^a-1}.
        \end{equation}

    Since the optimal candidate \(B\) is not the winner with constant probability, we have the following for some constant \(c\) independent of \(m,n\):

    \begin{equation}
        \liminf_{n \to \infty}
        \textsc{dist}^{\theta}_{\textsc{PL}}(\textsc{Plurality Veto},n,m)
        \geq
        c\frac{(m^a-1)^2}{2m^a-1}.
    \end{equation}

    This proves the desired theorem.


\end{proof}

\section{Metric Distortion of \textsc{Borda} Under the Plackett-Luce Model}
In this section, we derive the metric distortion of \textsc{Borda}. We focus on the PL choice model, where \(g(r)=1/(1+r^{-\theta})\) for \(\theta>1\). This specialization keeps the proof tractable while still exposing the main phenomenon. The proof ideas can also be adapted to other classes of choice functions.
Our main result is as follows:



For fixed \(\theta>1\), define the large-election distortion as a function of \(m\):
\[
D_m^\theta
:=
\limsup_{n \to \infty}
\textsc{dist}^{\theta}_{PL}(\textsc{Borda},n,m).
\]

\begin{theorem}{\label{theorem:borda_distortion}}
  Let $\distborda = m^{\max(1 - \frac{2}{\theta},0)}$. Then \(D_m^\theta=\Theta(\distborda)\) as \(m\to\infty\).
\end{theorem}

Observe that the distortion is bounded for $\theta\leq 2$ and grows with $\theta$ once the randomness in the votes becomes sufficiently small.
\textsc{Borda} is a natural and widely used voting rule because it accounts for the full ranking rather than only the top choice. However, in the classical model, its metric distortion matches that of \textsc{Plurality}, a rule that discards most of the ranking information. This surprising result suggests a limitation in how the classical model evaluates voting rules. Our findings resolve this issue by showing that, under probabilistic voting, \textsc{Borda}'s metric distortion better reflects the information contained in the ranking.
%
%
We first prove the lower bound with the help of an example. 

\subsection{Lower bound on the metric distortion of \textsc{Borda}}

\begin{lemma}
    \(D_m^\theta=\Omega(\distborda)\) as \(m\to\infty\), where $\distborda = m^{\max(1 - \frac{2}{\theta},0)}$.
\end{lemma}
\begin{proof}
   We construct an instance of an election on the 2-D Euclidean metric space. The `optimal' candidate B is at origin `O', and candidate W is at coordinates $(1,0)$. There are $(1-\delta)n $ voters at O and $\delta n$ voters at location L at coordinates $(t,0).$ The other $m-2$ candidates are at location M which is equidistant from the origin O and location L, with a distance of $t$ from each. That is, O, L, and M form an equilateral triangle. We will set $\delta$ of order $o(1)$ and $t$ of order $\omega(1)$ later to achieve the required result. 
     Consider the difference between the expected Borda scores of B and W.
    \begin{align*}
& \mathbb{E}[\textsc{Borda score of } B]
-\mathbb{E}[\textsc{Borda score of } W]\\
&= \sum_{i \in \mathcal{N}} \Bigg(
      \sum_{j \in \mathcal{A} \setminus \{B\}}
      \!\!\mathbbm{P}(B \succ_{\sigma_i} j)
      \;-\!\!\!\!
      \sum_{j \in \mathcal{A} \setminus \{W\}}
      \!\!\mathbbm{P}(W \succ_{\sigma_i} j)
    \Bigg) \\
&\le (1-\delta)n \left[(m-2) \left(1 - \frac{1}{1+ \left(\frac{1}{t}\right)^{\theta}}\right) + 1 \right] + \delta n \left[(m-2) \left(\frac{1}{2} - \frac{1}{1+ \left(\frac{t-1}{t}\right)^{\theta}} \right) +1\right].
\end{align*}

\noindent\textit{The first term corresponds to voters at O, and the second term corresponds to voters at L.}

\begin{align*}
&\leq n (m-2) \Bigg[\left(1 - \frac{1}{1+ \left(\frac{1}{t}\right)^{\theta}}\right) + \delta \left(\frac{1}{2} - \frac{1}{1+ \left(\frac{t-1}{t}\right)^{\theta}} \right) \Bigg]
   + n .
\end{align*}

\noindent\textit{Since $t$ is of order $\omega(1),$ we obtain upon simplification}
\(\leq n(m-2) \left[\frac{1}{t^{\theta}} - \frac{ \delta \theta}{8t} \right] + n\).
\textit{We choose $t = m^{1/\theta}$ and $\delta = 24 m^{1/\theta -1}$. Upon substitution, we obtain}
\(\leq n(m-2) \left[\frac{1}{m} - \frac{3\theta}{m} \right] + n \leq 2(1-\theta) n\).

      Recall that \( \theta > 1 \). Therefore, the expected Borda score of \( W \) exceeds that of \( B \) by at least a constant multiple of \( n \). Since the score difference is a sum of \( n \) independent random variables, each bounded in absolute value by \(m\), its variance is at most \(nm^2\).  
By Cantelli's inequality, for every fixed \(m\), this implies that \(B\) loses to \(W\) with probability tending to one as \(n\to\infty\). We now determine the ratio of the social costs of \( B \) and the other candidates:
\begin{align}
\frac{\SC(W,d)}{\SC(B,d)}  =  \frac{(1-\delta)n + \delta n (t-1)}{\delta n t} 
         & = \frac{(1-2\delta) + \delta t}{\delta t} \\
         \geq \frac{1}{\delta t} =\frac{1}{24} m^{1-\frac{2}{\theta}}.
\end{align}
       Whereas for any candidate $j$ at M, we have \( \frac{\SC(j,d)}{\SC(B,d)}  =  \frac{nt}{\delta n t}  = \frac{1}{\delta} =\frac{1}{24} m^{1-\frac{1}{\theta}} \geq \frac{1}{24} m^{1-\frac{2}{\theta}} \).
The distortion is at least $\frac{1}{24} m^{1-\frac{2}{\theta}}$ times the probability that $B$ loses, which we showed to be a constant. For the case of $\theta \in (2,\infty),$ the above result shows that the distortion is $\Omega (m^{1-\frac{2}{\theta}}).$ For $\theta \in (1,2],$ the distortion is at least a constant since it is at least 1 by definition.
\end{proof}

\subsection{Upper bound on the metric distortion of \textsc{Borda}}
We now give an order-wise matching upper bound.
We first establish the following result: any candidate whose social cost scales as $m^{\max(1 - \frac{2}{\theta},0)}$ times the minimum social cost, gets a ``much'' lower Borda score than some other candidate. Since we are interested in the asymptotics of the metric distortion in this section, we consider a sequence of metric spaces $\{d^{(n,m)}\}_{(n,m)}$ where $d^{(n,m)}$ denotes a metric space over $n$ voters and $m$ candidates. 
The precise result is as follows:




\newcommand{\lemmaboundingexp}{
Let $\distborda = m^{\max(1 - \frac{2}{\theta},0)}$. In each metric space $d^{(n,m)}$, consider the candidates $W^{(n,m)}$ and $B^{(n,m)}$. Suppose that $\lim\limits_{m \to \infty}\frac{\SC(W^{(n,m)},d^{(n,m)})}{\SC(B^{(n,m)},d^{(n,m)})} \frac{1}{ \distborda} = \infty. $ Then there exist constants $m_0>1,c > 0$ such that for all $m > m_0,$ the expected Borda score of $W^{(n,m)}$ is at least $cn$ lower than the expected Borda score of some other candidate. 
}

\begin{lemma}{\label{lemma:bounding_expected_score}}
\lemmaboundingexp
\end{lemma}

\begin{proof}
   For brevity, we omit the superscripts in the metric space $d^{(n,m)}$ and the candidate names when clear from the context. Given that our probability distributions satisfy scale-freeness (SF), we set $d(B,W)=1$.   Suppose the ratio of social costs of candidate $B$ and candidate $W$ is lower bounded by $\mff(m) = \omega(h_{\theta}(m))$. Then we obtain:
    \begin{align}
        \frac{\SC(W,d)}{\SC(B,d)} \geq \mff(m)  \overset{(a)}{\implies}  & 1+ \frac{1}{\frac{1}{n} \sum\limits_{i \in \N}\frac{d(i,B)}{d(B,W)}} \geq \mff(m) \\ \implies & \frac{1}{n} \sum\limits_{i \in \N}\frac{d(i,B)}{d(B,W)} \leq \frac{1}{\mff(m)-1}. \nonumber \\
        \intertext{$(a)$ follows from the triangle inequality: $d(i,W) \leq d(i,B) + d(B,W)$ for all $i \in \N$. Since $\mff(m)$ goes to $\infty$ as $m \rightarrow \infty$, we choose $m$ large enough such that $\mff(m) \geq 2$. Given that $d(B,W) = 1,$ we obtain:}
        &  \frac{1}{n} \sum\limits_{i \in \N} d(i,B) \leq \frac{1}{\mff(m)-1}.\label{eq:distortion_constraint}
    \end{align}
    %


Let  \( \V \) denote the set of voters within a distance of \( \frac{1}{\sqrt{\mff(m)-1}} \)  from B. Similarly, let   \( \C \) denote the set of candidates within a distance of \( \frac{1}{2} \)  from B. The following observation follows from the constraint on the average distance of voters from candidate \( B \) in Equation \eqref{eq:distortion_constraint}.

    \begin{observation} \label{obs:size_V}
        At most $\frac{n}{\sqrt{\mff(m)-1}}$ voters can be in the set $\N \setminus \V$. That is, $|\V| \geq \left(1 - \frac{1}{\sqrt{\mff(m)-1}}\right) n.$
    \end{observation}

    We compute the expected ``truncated'' Borda scores for candidates in \( \C \), considering only the candidates in   \( \C \) and the voters in \( \V \) for ranking and score calculation (hence ``truncated'').  

Let \( B^{(b)} \) be the candidate that achieves the highest expected truncated Borda score. We have: 
\begin{observation}\label{eq:borda_winner_distance}
     The distance between $B$ and $B^{(b)}$ satisfies:    $ d(B^{(b)}, B) \leq \frac{2}{\sqrt{\mff(m)-1}}.$
\end{observation}
Assume the contrary. By the triangle inequality, we have $d(B^{(b)}, i) > |d(B^{(b)},B) - d (i,B)| > \frac{1}{\sqrt{\mff(m)-1}} > d(B, i)$ for every $i \in \V$. Thus, from \eqref{eq:pairwise_probability}, for every $j\in \C\setminus\{B,B^{(b)}\}$ and every $i\in\V$,
\[
\mathbb{P}(B^{(b)} \succ_{\sigma_i} j)
<
\mathbb{P}(B \succ_{\sigma_i} j).
\]
Moreover, for every \(i\in\V\), candidate \(B\) is closer to \(i\) than \(B^{(b)}\), so \(\mathbb{P}(B\succ_{\sigma_i}B^{(b)})>1/2\). Therefore the expected truncated Borda score of \(B\) is strictly larger than that of \(B^{(b)}\), contradicting the choice of \(B^{(b)}\).


     We lower bound the difference between the expected Borda score of candidates $B^{(b)}$ and $W$. 
    \begin{align}
        & \mathbb{E}[\textsc{Borda score of $B^{(b)}$}] - \mathbb{E}[\textsc{Borda score of $W$}] \nonumber\\  
        = & \sum_{i \in \mathcal{N}} \Biggl(\sum_{j \in \mathcal{A} \setminus \{B^{(b)}\}} \mathbbm{P} (B^{(b)} \succ_{\sigma_i} j) - \sum_{j \in \mathcal{A} \setminus \{W\}}\mathbbm{P} (W \succ_{\sigma_i} j)\Biggr)\label{eq:temp_1}\\
        = & \underbrace{\sum_{j \in \C \setminus \{B^{(b)}\}} \sum_{i \in \mathcal{N}} \mathbbm{P} (B^{(b)} \succ_{\sigma_i} j) - \mathbbm{P} (W \succ_{\sigma_i} j)}_{\text{(I)}} \nonumber\\
        + & \underbrace{\sum_{j \in \A \setminus (\C \cup \{W\})} \sum_{i \in \mathcal{N}} \mathbbm{P} (B^{(b)} \succ_{\sigma_i} j) - \mathbbm{P} (W \succ_{\sigma_i} j)}_{\text{(II)}} + \underbrace{\sum_{i \in \mathcal{N}} 2\mathbbm{P} (B^{(b)} \succ_{\sigma_i} W) - 1.}_{\text{(III)}}&    \label{eq:diff_borda_score}
    \end{align}



   Term (I) corresponds to the comparisons of $B^{(b)}$ and $W$ with candidates in $\C.$ Term (II)  corresponds to the comparisons of $B^{(b)}$  and $W$ with candidates not in $\C.$ Whereas Term (III) corresponds to the comparison between $B^{(b)}$  and $W$.
    We lower bound each term separately.
    \begin{align}
(I) \geq\;
&\sum_{j \in  \C \setminus \{B^{(b)}\}}
 \Bigg(
   \sum_{i \in \V}
     \big(
       \mathbbm{P} (B^{(b)} \succ_{\sigma_i} j)
       - \mathbbm{P} (W \succ_{\sigma_i} j)
     \big)
   \;-\;
   \sum_{i \in \mathcal{N} \setminus \V} 1
 \Bigg) \nonumber
\\
\noalign{\vskip2pt}
\end{align}

\noindent\textit{The first term is the expected truncated Borda score of $B^{(b)}$. Since $B^{(b)}$ has the highest expected truncated Borda score among candidates in $\C$, we can lower bound it with the average expected truncated Borda score, which is $\frac{(|\C| -1)|\V|}{2}.$ The third term can be bounded via Observation~\ref{obs:size_V}.}

\begin{align}
\geq\;
&\frac{|\V|(|\C| -1)}{2}
\;-\;
\sum_{j \in  \C \setminus \{B^{(b)}\}}
\sum_{i \in \V}
\frac{1}{1+ \left(\frac{d(i,W)}{d(i,j)}\right)^{\theta}}
\;-\;
\frac{n (|\C| -1)}{\sqrt{\mff(m)-1}}
\\
\noalign{\vskip2pt}
\end{align}

\noindent\textit{By the triangle inequality, $d(i,W) \geq  1-\frac{1}{\sqrt{\mff(m)-1}}$ and $d(i,j) \leq \frac{1}{2} + \frac{1}{\sqrt{\mff(m)-1}}$ for all $i \in \V$ and $j \in \C$.}

\begin{align}
\geq\;
&\frac{|\V|(|\C| -1)}{2}
\;-\;
\frac{|\V|(|\C| -1)}{
  1+ \left(
      \frac{1-\frac{1}{\sqrt{\mff(m)-1}}}{
            \frac{1}{\sqrt{\mff(m)-1}} + \frac{1}{2}}
     \right)^{\theta}}
\;-\;
\frac{n (|\C| -1)}{\sqrt{\mff(m)-1}}
\\
\noalign{\vskip2pt}
\end{align}

\noindent\textit{Since $|\V| \geq \Big(1 - \frac{1}{\sqrt{\mff(m)-1}}\Big) n$ by Observation~\ref{obs:size_V}, and given that $\mff(m)$ grows unbounded as $m \rightarrow \infty$,}

\begin{align}
\geq\; & n (|\C| -1) \Omega_m(1).
\label{eq:bounding_I}
\end{align}

We now bound Term (II) in Equation~\eqref{eq:diff_borda_score}, subject to the distance constraint~\eqref{eq:distortion_constraint}. Since the objective function is additive over candidates \( j \in \A \setminus (\C \cup \{W\}) \), and neither the choice of \( B^{(b)} \) nor the distance constraint couples these candidates, it suffices to lower bound the contribution of an arbitrary fixed outside candidate \(j\) and then sum this bound over all such candidates.
For this fixed candidate \(j\), write
\[
z_j:=\frac{d(B^{(b)},j)}{d(B^{(b)},W)}.
\]
Let \(Z_\theta\) be a sufficiently large constant, depending only on \(\theta\), so that the uniform estimates in Claim~\ref{claim:borda_optim1} apply whenever \(z_j>Z_\theta\). We split according to whether \(z_j\) is larger than this threshold. If \(z_j\) is below this threshold, then \(d(B,j)=O_\theta(1)\), and the elementary bounded-distance argument applies. If \(z_j\) is above this threshold, we apply the optimization argument.




\textbf{Case 1:} \textit{\(z_j\leq Z_\theta\), where \(Z_\theta\) is the fixed threshold from Claim~\ref{claim:borda_optim1}.}

For any candidate \( j \in \A \setminus (\C \cup \{W\}) \) and voter $i \in \V,$ we have:
\begin{align}
    \mathbbm{P} (B^{(b)} \succ_{\sigma_i} j) - \mathbbm{P} (W \succ_{\sigma_i} j) & = \frac{1}{1+ \left(\frac{d(B^{(b)},i)}{d(i,j)}\right)^{\theta}} - \frac{1}{1+ \left(\frac{d(W,i)}{d(i,j)}\right)^{\theta}} \nonumber \\
     & \overset{(a)}{\geq} \frac{1}{1+ \left(\frac{\frac{3}{\sqrt{\mff(m)-1}}}{\frac{1}{2}-\frac{1}{\sqrt{\mff(m)-1}}}\right)^{\theta}} - \frac{1}{1+ \left(\frac{1-\frac{1}{\sqrt{\mff(m)-1}}}{C_\theta+\frac{1}{\sqrt{\mff(m)-1}}}\right)^{\theta}} \quad \\ = & \quad \Omega_m\left(1-\frac{1}{1+C_\theta^{-\theta}}\right).
     \label{eq:bounding_prob_diff_z_bound_individual_voter}
\end{align}

    $(a)$ follows from the triangle inequality. That is, $d(B^{(b)},i) \leq d(B^{(b)},B)+d(i,B) \leq \frac{3}{\sqrt{\mff(m)-1}}$. Also $d(i,j) \geq d(j,B)-d(i,B) \geq \frac{1}{2}-\frac{1}{\sqrt{\mff(m)-1}}$.  Similarly, since \(z_j\leq Z_\theta\), we have \(d(B,j)=O_\theta(1)\); writing this constant as \(C_\theta\), $d(i,j)\leq d(j,B)+d(i,B) \leq C_\theta+\frac{1}{\sqrt{\mff(m)-1}}$. Finally, $d(W,i)\geq d(W,B)-d(i,B) \geq 1 - \frac{1}{\sqrt{\mff(m)-1}}$.

    Since at least $\Big(1-\frac{1}{\sqrt{\mff(m)-1}}\Big)n$
    voters are in $\V$ (Observation~\ref{obs:size_V}), when  \(z_j\leq Z_\theta\), we have:
    \begin{align}\label{eq:bounding_prob_diff_z_bound}
        \sum_{i\in \N} \mathbbm{P} (B^{(b)} \succ_{\sigma_i} j) - \mathbbm{P} (W \succ_{\sigma_i} j) \geq n\Omega_m\left(1-\frac{1}{1+(2C_\theta)^{-\theta}}\right).
    \end{align}







\textbf{Case 2:} \textit{\(z_j>Z_\theta\).}
    

For ease of notation, we define  
\begin{align}
b_i := \frac{d(i,B^{(b)})}{d(B^{(b)},j)}, \quad
w_i := \frac{d(i,W)}{d(B^{(b)},j)}, \quad \\
z := \frac{d(B^{(b)},j)}{d(B^{(b)},W)}, \quad
t_i := \frac{d(i,j)}{d(B^{(b)},j)}.
\end{align}

In this case \(z_j\) is sufficiently large for Claim~\ref{claim:borda_optim1} to apply. We write \(z=z_j\) for notational simplicity.

From Equation~\eqref{eq:distortion_constraint},  Observation~\ref{eq:borda_winner_distance}, and an application of the triangle inequality, we obtain: 
    \begin{align}
\frac{1}{n} \sum_{i \in \N} d(i,B^{(b)})
&\leq \frac{1}{\mff(m)-1}
   + \frac{2}{\sqrt{\mff(m)-1}}
   \label{eq:optimization_II_contraint_first_step} \\
\implies\hspace{-1.0 em}\quad
\frac{1}{n} \sum_{i \in \N} \frac{d(i,B^{(b)})}{d(B^{(b)},W)}
&\leq
   \left(
     \frac{1}{\mff(m)-1}
     + \frac{2}{\sqrt{\mff(m)-1}}
   \right) \times \left(
     1-\frac{2}{\sqrt{\mff(m)-1}}
   \right)^{-1}
   \label{eq:optimization_II_contraint}
\end{align}


For simplicity of notation, we define  
\(
\mfg(m) := \Big(\frac{1}{\mff(m)-1} + \frac{2}{\sqrt{\mff(m)-1}} \Big)^{-1} \Big(1-\frac{2}{\sqrt{\mff(m)-1}}\Big).
\)  
Since \( \mff(m) = \omega(h_{\theta}(m)) \), it follows that \( \mfg(m) = \omega(h_{\theta}(m)) \).  

Now, for a candidate \( j \in \A \setminus (\C \cup \{W\}) \), 
the inner summation in Term (II) is given by:
    \begin{align}{\label{eq:diff_borda_score_objective}}
        & \sum_{i \in \N} \bigg(\frac{1}{1+ \left(\frac{b_i}{t_i} \right)^{\theta}} - \frac{1}{1+ \left(\frac{w_i}{t_i} \right)^{\theta}} \bigg) \\ \quad
        \geq & \quad \sum_{i \in \N}  \bigg( \frac{1}{1+ \left(\frac{b_i}{t_i} \right)^{\theta}} - \frac{1}{1+ \left(\frac{|1/z-b_i|}{t_i} \right)^{\theta}} \bigg).
    \end{align}
    The inequality follows on application of the triangle inequality since $d(i,W) \geq |d(B^{(b)},W) - d(i,B^{(b)})|$ which naturally translates to $w_i \geq |\frac{1}{z}-b_i|.$

We now derive a lower bound for Term (III) in Equation~\eqref{eq:diff_borda_score}. We first define \( a_i \) as  
\(
a_i := \frac{d(B^{(b)},i)}{d(B^{(b)},W)}.
\)
Applying the triangle inequality,  
\(
d(W,i) \geq \left| d(B^{(b)},i) - d(B^{(b)},W) \right|,
\) 
we obtain:  
\begin{equation}\label{eq:diff_prob_III_term}
    (III) \geq \sum_{i \in \N}\Bigg(\frac{2}{1+ \left(\frac{a_i}{|1-a_i|}\right)^{\theta}}-1\Bigg).
\end{equation}

    We consider two optimization problems, denoted by \( \mathcal{E}^{(1)}_{\mfg,z} \) and \( \mathcal{E}^{(2)}_{\mfg} \), and solve them separately. From Equation~\eqref{eq:diff_borda_score_objective}, each outside candidate \(j\in\A\setminus(\C\cup\{W\})\) contributes at least \(\Opt(\mathcal{E}^{(1)}_{\mfg,z_j})\), where \(z_j=d(B^{(b)},j)/d(B^{(b)},W)\). Hence
\[
(II)\geq \sum_{j\in\A\setminus(\C\cup\{W\})}\Opt(\mathcal{E}^{(1)}_{\mfg,z_j}).
\]
Similarly, from Equation~\eqref{eq:diff_prob_III_term}, we have \( (III) \geq \Opt(\mathcal{E}^{(2)}_{\mfg}) \).
    The constraints in these optimization problems follow from Equation \eqref{eq:optimization_II_contraint} and the triangle inequality. 
    We keep the dependence of \( \mathcal{E}^{(1)}_{\mfg,z} \) on the parameter \(z\) explicit, since the proof applies the resulting bound to each value \(z_j=d(B^{(b)},j)/d(B^{(b)},W)\).


\begin{figure}[h]  
\centering
\begin{minipage}{0.57\textwidth}
\begin{equation}\label{eqn:optim_formulation_borda_first}
\mathcal{E}^{(1)}_{\mfg,z} =
    \left\{
    \begin{aligned}
   \min_{ \vect{b}, \vect{t} \in \mathbb{R}_{\geq 0}^n} & \sum_{i \in \N} \frac{1}{1+ \left(\frac{b_i}{t_i} \right)^{\theta}} - \frac{1}{1+ \Big(\frac{\left|\frac{1}{z}-b_i\right|}{t_i} \Big)^{\theta}} \\
    \textrm{s.t.} \quad & \sum_{i \in \N} b_i \leq \frac{n}{z\mfg(m)} \\
      \quad & \max_{i \in \N}|b_i - t_i| \leq 1 \leq \min_{i \in \N} (t_i + b_i) \\
    \end{aligned}
    \right.
\end{equation}
\end{minipage}
\hfill
\begin{minipage}{0.41\textwidth}
\begin{equation}\label{eqn:optim_formulation_borda_second}
\mathcal{E}^{(2)}_{\mfg} =
    \left\{
    \begin{aligned}  \min_{ \vect{a} \in \mathbb{R}_{\geq 0}^n} & \sum_{i \in \N} \left(\frac{2}{1+\left(\frac{a_i}{|1-a_i|}\right)^{\theta}} -1 \right)\\
    \textrm{s.t.} \quad & \sum_{i \in \N} a_i \leq \frac{n}{\mfg(m)} \\
    \end{aligned}
    \right.
\end{equation}
\end{minipage}
\end{figure}

We now state claims that lower bound \( \Opt(\mathcal{E}^{(1)}_{\mfg,z}) \) and \( \Opt(\mathcal{E}^{(2)}_{\mfg}) \).
The proofs require tools from Lagrangian theory and asymptotic analysis of the derivative, and are in Appendix~\ref{sec:proof_claims_in_lemma_boudning_score}.

\gdef\claimbordaoptimone{
 There exists a constant \(Z_\theta<\infty\), depending only on \(\theta\), such that if \(z\geq Z_\theta\) and \(\mfg(m)=\omega(h_\theta(m))\), then the following bounds hold uniformly over all such \(z\): \(\Opt(\mathcal{E}^{(1)}_{\mfg,z}) \geq n\Omega_z(z^{-2})\) if \(\theta<2\), and \(\Opt(\mathcal{E}^{(1)}_{\mfg,z})\geq -n\,o_m(m^{-1})\) otherwise.
}
\global\let\claimone\claimbordaoptimone

\begin{claim}{\label{claim:borda_optim1}}
\claimbordaoptimone
\end{claim}

\gdef\claimbordaoptimtwo{
   If $\mfg(m) = \omega(h_{\theta}(m))$, then $\Opt(\mathcal{E}^{(2)}_{\mfg}) \geq n\Omega_m(1)$.
}
\global\let\claimtwo\claimbordaoptimtwo

\begin{claim}{\label{claim:borda_optim2}}
\claimbordaoptimtwo
\end{claim}

For each fixed outside candidate \(j\), if \(z_j\leq Z_\theta\), Equation~\eqref{eq:bounding_prob_diff_z_bound} gives a nonnegative \(n\Omega_m(1)\)-type lower bound for that candidate's contribution. If \(z_j>Z_\theta\), Claim~\ref{claim:borda_optim1} lower bounds its contribution to Term (II); the constants in the claim are independent of the particular \(z_j\), so summing over all outside candidates gives a total contribution at least \(-n o_m(1)\) in the \(\theta\geq2\) case and a nonnegative contribution in the \(\theta<2\) case. Combining these per-candidate bounds with Claim~\ref{claim:borda_optim2} and Equation~\eqref{eq:bounding_I}, we get that (I) + (II) + (III) in Equation~\eqref{eq:diff_borda_score} is lower bounded by \(n\Omega_m(1)\). This completes the proof.
 \end{proof}

We can now state the upper bound on the metric distortion of \textsc{Borda}.

\newcommand{\lemmabordaub}{
  \(D_m^\theta=O(\distborda)\) as \(m\to\infty\).
}

\begin{lemma}{\label{lemma:borda_upper_bound}}
  \lemmabordaub
\end{lemma}

The proof of Lemma \ref{lemma:borda_upper_bound} analyzes two cases for each non-optimal candidate \( A_j \):  

When the expected number of voters ranking \( A_j \) above \( B \) is less than the threshold \( \frac{n}{m} - \frac{n^{\epsilon + 1/2}}{m} \), we bound the ratio of social costs between \( A_j \) and \( B \) using Observation \ref{lemma:bounding_social_cost_ratio}. This directly provides a bound on \( A_j \)'s contribution to the overall distortion, which approaches zero as \( n \) tends to infinity.

When the expected number of voters ranking \( A_j \) above \( B \) exceeds the threshold, and the ratio of the social costs of \( A_j \) to \( B \) is lower bounded by \( \omega(h_{\theta}(m)) \), we apply Hoeffding's inequality and Lemma~\ref{lemma:bounding_expected_score} to bound the probability that \( A_j \) wins. A full proof is in Appendix \ref{sec:borda_upper_bound_proof}. 

The analysis of metric distortion under \textsc{Borda} is technically intricate, but the core intuition is clear. In the classical setting, \textsc{Borda}'s vulnerability to irrelevant alternatives makes the distortion grow linearly with \(m\). Under probabilistic voting, these irrelevant alternatives are much less effective at inflating the score of a single bad candidate. If the randomness is sufficiently high (\(\theta < 2\)), the distortion no longer scales with \(m\) in the limit. This result supports the classical practitioner intuition that \textsc{Borda} is a reasonable voting rule.

\section{Numerical Evaluations}
\begin{figure}[h]
    \centering
    \begin{subfigure}[b]{0.45\textwidth}
        \centering
        \includegraphics[width=\textwidth]{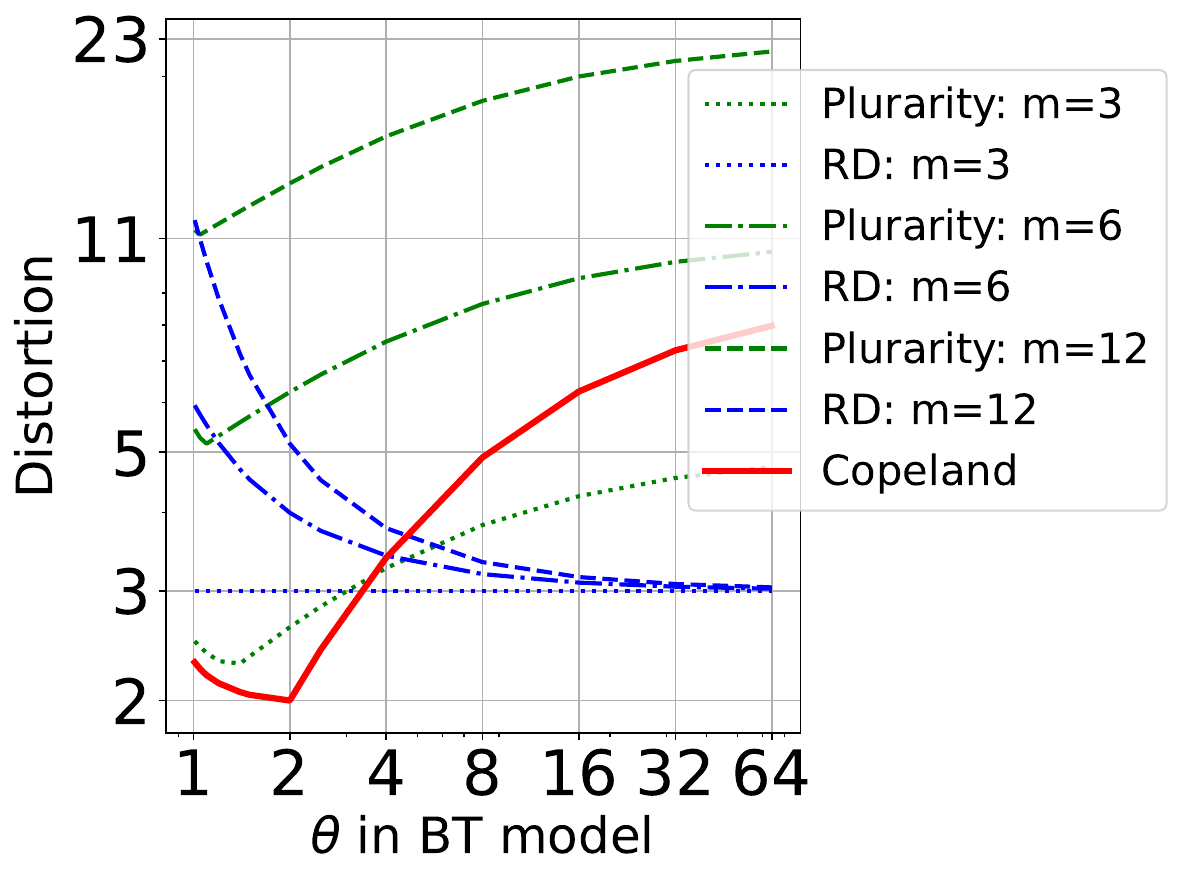}
    \end{subfigure}
    \caption{The distortion bounds on voting rules, varying with the randomness parameter of the PL model, in the limit $n\rightarrow \infty$. Both axes are on the log scale. We plot the upper bound for \textsc{Copeland} (Theorem~\ref{theorem:Copeland_distrotion_m}), the lower bound for \textsc{Random Dictator} (Theorem~\ref{theorem:RD_distortion_lower_bound}), and the upper bound for \textsc{Plurality} (Theorem~\ref{theorem:thm_plurality_distortion_m}).
    }
    \label{fig:dist}
\end{figure}

Recall that higher value of $\theta$ corresponds to lower randomness. From Figure~\ref{fig:dist}, we observe that under sufficient randomness, the more intricate voting rule \textsc{Copeland} outshines the simpler rule \textsc{Random Dictator}, which only looks at a voter's top choice. Moreover, its distortion is independent of $m$ in the limit $n \rightarrow \infty$. This is in contrast to \textsc{Random Dictator}, where the distortion is $\Omega(m^{1/\theta})$ in the PL model, a sharp rate of increase in $m$ for low values of $\theta.$ The metric distortion of \textsc{Plurality} increases linearly in $m$. 
An important observation is regarding the asymptotics when $\theta$ increases. The metric distortion of \textsc{Random Dictator} converges to its value under deterministic voting, i.e., 3. The metric distortion of \textsc{Plurality} also converges to $2m-1$, the same as in deterministic voting. 






\section{Discussion and Future Work}

Although the metric distortion framework has garnered significant research interest, its results often lack interpretability. This paper aims to address some of these limitations by demonstrating that a small and natural modification to the framework yields results that align more closely with intuition.  
We extend the metric distortion framework by incorporating the inherent randomness in voter behavior. Our findings reveal that the original framework is overly pessimistic about two important voting rules—\textsc{Copeland} and \textsc{Borda}. Conversely, the intuitively unsatisfying rule \textsc{Random Dictator}, which achieves a distortion of 3 (matching the best possible performance of any deterministic rule), performs poorly when considered in a more realistic setting—its distortion increases with the number of candidates in our model.
Our framework provides a more nuanced perspective on metric distortion, opening opportunities to revisit the problem with a model that better reflects real-world voter behavior.

\paragraph{Future Work} An interesting direction for future research is extending our analysis to other tournament graph-based voting rules, both weighted and unweighted. Our techniques are well-suited for this class of rules, as they rely on the expected edge weights of the tournament graph.   
Another promising avenue is the design of randomness-aware voting rules, which could be particularly relevant to the emerging paradigm of social choice in reinforcement learning with human feedback (RLHF) \cite{conitzer2024social}. An example of this approach is the work in \citep{liu2023robust}, which applies a learning-theoretic perspective to designing voting rules under the assumption of random voting according to the Mallows model.
Finally, identifying or designing the voting rule that minimizes metric distortion under the PL model would be a significant contribution to this area.











\bibliographystyle{alpha}
\bibliography{refs}
\newpage
\appendix

\newgeometry{left=1in, right=1in, top=1in, bottom=1in}

\addtocontents{toc}{\protect\setcounter{tocdepth}{0}} 

\section{Asymptotic Notation}{\label{sec:asymptotic_notation}}

We use standard asymptotic notation. Unless explicitly stated otherwise, the variable tending to infinity is the one displayed in the surrounding statement; all other parameters are fixed. For example, if \(D_m\) is defined using a \(\limsup_{n\to\infty}\), then a statement such as \(D_m=\Theta(h(m))\) is interpreted as \(m\to\infty\).

\begin{itemize}
    \item \(f(x)=O(g(x))\) as \(x\to\infty\) means that there are constants \(C>0\) and \(x_0\) such that \(|f(x)|\leq C|g(x)|\) for all \(x\geq x_0\).
    \item \(f(x)=\Omega(g(x))\) as \(x\to\infty\) means that there are constants \(c>0\) and \(x_0\) such that \(f(x)\geq c g(x)\) for all \(x\geq x_0\), in the nonnegative regimes considered in this paper.
    \item \(f(x)=\Theta(g(x))\) as \(x\to\infty\) means that both \(f(x)=O(g(x))\) and \(f(x)=\Omega(g(x))\) hold.
    \item \(f(x)=o(g(x))\) as \(x\to\infty\) means that \(\lim_{x\to\infty} f(x)/g(x)=0\).
    \item \(f(x)=\omega(g(x))\) as \(x\to\infty\) means that \(\lim_{x\to\infty} f(x)/g(x)=\infty\).
\end{itemize}

\section{Proof of Lemma \ref{lemma:existence_maxima}}\label{sec:corollary_proof}
\noindent\textbf{Lemma~\ref{lemma:existence_maxima} (Restated):} \lemmauniquemaxima

\begin{proof}
Let
\[
    f(x):=g_{\textsc{mid}}(x)
    \quad\text{and}\quad
    h(x):=\frac{f(x)}{x}.
\]
From the identity \(g(r)+g(1/r)=1\), we have
\[
    f(x)+f(1-x)=1,\qquad x\in(0,1).
\]
In particular, \(f(0)=0\), \(f(1/2)=1/2\), and \(\lim_{x\to1}f(x)=1\). Differentiating the symmetry identity also gives
\[
    f'(x)=f'(1-x).
\]

We first rule out local maxima in \((0,1/2)\). Since \(f\) is strictly convex on \((0,1/2)\) and \(f(0)=0\), the slope \(f(x)/x\) is strictly increasing on \((0,1/2)\). Hence \(h\) has no local maximum in \((0,1/2)\). Moreover, strict convexity and \(f(1/2)=1/2\) imply
\[
    f'(1/2)> \frac{f(1/2)-f(0)}{1/2-0}=1.
\]
Therefore
\[
    h'(1/2+)
    =
    \frac{(1/2)f'(1/2)-f(1/2)}{(1/2)^2}
    >0.
\]

It remains to analyze \(h\) on \((1/2,1)\). Define
\[
    s(x):=x f'(x)-f(x).
\]
Then \(h'(x)=s(x)/x^2\). Since \(f\) is strictly concave on \((1/2,1)\), its derivative is strictly decreasing there, and hence \(s\) is strictly decreasing on \((1/2,1)\). Indeed, at all points where the second derivative is read pointwise, \(s'(x)=x f''(x)\leq0\), and strict concavity rules out a flat interval for \(s\).

We have just shown that
\[
    s(1/2)=\frac12 f'(1/2)-\frac12>0.
\]
On the other hand, regular variation of \(f\) at zero with index \(\beta>1\) implies \(f(x)/x\to0\) as \(x\downarrow0\). Since \(f\) is convex near zero, this also implies \(f'(x)\to0\) as \(x\downarrow0\). By the derivative symmetry \(f'(x)=f'(1-x)\), we get \(f'(x)\to0\) as \(x\uparrow1\). Thus
\[
    \lim_{x\uparrow1}s(x)
    =
    \lim_{x\uparrow1}\bigl(xf'(x)-f(x)\bigr)
    =
    -1.
\]
Because \(s\) is continuous and strictly decreasing on \((1/2,1)\), it has a unique zero \(x^\ast_{\textsc{mid}}\in(1/2,1)\). Therefore \(h'(x)>0\) on \((1/2,x^\ast_{\textsc{mid}})\) and \(h'(x)<0\) on \((x^\ast_{\textsc{mid}},1)\). This proves that \(x^\ast_{\textsc{mid}}\) is the unique local, and hence global, maximum of \(g_{\textsc{mid}}(x)/x\) in \((0,1)\).

\end{proof}

\section{A Chernoff Bound Used in the Upper-Bound Proofs}\label{sec:appendix_threshold_chernoff_bound}

In this section, we state a version of the Chernoff threshold bound that is used in the proofs of the finite $n$ upper bounds for Plurality (Section~\ref{sec:thm_plurality_distortion_m_proof}) and Copeland (Section~\ref{sec:Copeland_distrotion_m_proof}).

\begin{lemma}[Threshold Chernoff Bound]\label{lemma:threshold_chernoff_bound}
Let \(X_1,\ldots,X_n\) be independent Bernoulli random variables, and let
\[
    \rho:=\frac{1}{n}\sum_{i=1}^{n}\mathbb P[X_i=1].
\]
Fix \(q\geq 2\), \(r\geq 1\), and \(\epsilon\in(0,1/2)\). If
\[
    \rho\leq \frac{1}{q}-\frac{n^{-1/2+\epsilon}}{q},
    \quad
    \frac{n}{q}>r,
    \quad\text{and}\quad
    \frac{n^{1/2+\epsilon}}{q}\geq r,
\]
then
\[
    \mathbb P\left[\sum_{i=1}^{n}X_i\geq \frac{n}{q}\right]
    \leq
    (q\rho)^r
    \exp\left(
        \frac{-n^{1/2+\epsilon}+2rq}
        {q(2n^{1/2-\epsilon}-1)}
    \right).
\]
\end{lemma}

\begin{proof}
If \(\rho=0\), the claim is immediate. Otherwise, let \(\mu=n\rho\). The Chernoff bound gives
\[
    \mathbb P\left[\sum_{i=1}^{n}X_i\geq \frac{n}{q}\right]
    \leq
    \left(
        \frac{e^{\frac{1}{q\rho}-1}}
        {(\frac{1}{q\rho})^{1/(q\rho)}}
    \right)^{n\rho}
    =
    (q\rho)^{n/q} e^{n/q-n\rho}.
\]
Since \(n/q>r\), we rewrite this as
\[
    (q\rho)^r
    \left(q\rho\exp\left(-\frac{n\rho}{n/q-r}\right)\right)^{n/q-r}
    e^{n/q}.
\]
The function \(x e^{-x}\) is increasing on \((0,1)\). The condition \(n^{1/2+\epsilon}/q\geq r\) ensures that
\[
    \frac{n\rho}{n/q-r}\leq 1
\]
throughout the range considered. Under the assumed upper bound on \(\rho\), the expression inside the parentheses is maximized at
\(\rho=(1-n^{-1/2+\epsilon})/q\). Hence
\[
    \mathbb P\left[\sum_{i=1}^{n}X_i\geq \frac{n}{q}\right]
    \leq
    (q\rho)^r
    \left(1-n^{-1/2+\epsilon}\right)^{n/q-r}
    \exp\left(\frac{n^{1/2+\epsilon}}{q}\right).
\]
Finally, using
\[
    \log(1+x)\leq \frac{2x}{2+x}
    \quad\text{for }-1<x\leq0
\]
with \(x=-n^{-1/2+\epsilon}\), we obtain
\[
    \left(1-n^{-1/2+\epsilon}\right)^{n/q-r}
    \leq
    \exp\left(
        \frac{-2n^{-1/2+\epsilon}(n/q-r)}
        {2-n^{-1/2+\epsilon}}
    \right).
\]
Combining the two exponential terms gives the claimed bound.
\end{proof}

\section{Proof of the Plurality Upper Bound (Theorem~\ref{theorem:thm_plurality_distortion_m})}

This appendix section proves the upper bound for \textsc{Plurality}. We first show that the fractional voter-type relaxation may be reduced to two active voter types (Subsection~\ref{sec:two_active_voter_types_proof}). We then analyze the resulting two-location program \(\Efrac{\alpha}\) (Subsection~\ref{sec:optimizer_e_alpha_soln}) and use it to prove Theorem~\ref{theorem:thm_plurality_distortion_m}.
\subsection{Two Active Voter Types}\label{sec:two_active_voter_types_proof}

	Fix a finite set \(T\) of feasible voter types, where each type \(t\in T\) is a pair \((b_t,w_t)\) satisfying
	\[
		|b_t-w_t|\leq1,\qquad b_t+w_t\geq1,\qquad b_t,w_t\geq0.
	\]
The fractional voter-type relaxation over this fixed set of types is
\[
\mathcal E_\alpha(T)=
\left\{
\begin{aligned}
	\text{minimize}\quad
	& \frac{\sum_{t\in T}p_t b_t}{\sum_{t\in T}p_t w_t}\\
	\text{subject to}\quad
	& \sum_{t\in T}p_t g\left(\frac{b_t}{w_t}\right)\geq\alpha,\\
	& \sum_{t\in T}p_t=1,\quad p_t\geq0 \quad \forall t\in T.
\end{aligned}
\right.
\]
Here \(g(b_t/w_t)\) is interpreted as \(g(\infty)\) when \(w_t=0\), and the objective-ratio convention is the one stated in the main text.
The relaxation used in the main text optimizes this program also over the feasible type locations.

\begin{lemma}[Two active voter types]\label{lemma:two_active_voter_types}
	For every fixed finite set of feasible voter types \(T\), the optimization over the fractions \(p_t\) in \(\mathcal E_\alpha(T)\) has an optimal basic feasible solution supported on at most two voter types, whenever the optimum is attained. Consequently, the full fractional voter-type relaxation has the same infimum value as the two-location formulation in \eqref{eqn:optim_formulation_two_type}, which is the formulation used in the main text.
\end{lemma}
		
\begin{proof}
	Write \(g_t := g(b_t/w_t)\). We optimize only over the fractions \(p_t\). Feasible solutions with zero denominator do not improve a finite minimum, so it is enough to consider a feasible solution with \(W(p):=\sum_{t\in T}p_t w_t>0\). Define \(q_t := p_t/W(p)\). Then
		\[
			\sum_{t\in T}q_t w_t=1,
			\qquad
			\frac{\sum_t p_t b_t}{\sum_t p_t w_t}
			=
			\sum_t q_t b_t.
		\]
	Since \(p_t=q_t/\sum_s q_s\), the probability constraint
	\(\sum_t p_t g_t\geq \alpha\)
	is equivalent to \(\sum_t q_t(g_t-\alpha)\geq 0\).
	Thus, after fixing the voter types, the fractional optimization over the voter fractions is equivalent to the finite linear program
	\[
	\begin{aligned}
		\text{minimize}\quad & \sum_{t\in T}q_t b_t\\
		\text{subject to}\quad
		& \sum_{t\in T}q_t w_t=1,\\
		& \sum_{t\in T}q_t(g_t-\alpha)\geq0,\\
		& q_t\geq0 \quad \forall t\in T.
	\end{aligned}
	\]
	Every finite linear program with an attained optimum has an optimal basic feasible solution. In any basic feasible solution of this program, the number of positive variables is at most the number of active linearly independent non-negativity-excluding constraints. Here there is always the equality constraint \(\sum_t q_t w_t=1\), and the probability constraint is either inactive or active. Therefore an optimal basic feasible solution has support size at most two.
	
	The transformation between \(p\) and \(q\) preserves support, so there is an optimal choice of voter fractions \(p_t\) supported on at most two of the fixed voter types. If the voter locations are also being optimized, apply the same argument to the finite set of voter types appearing in any feasible solution; replacing only the weights by the basic feasible solution preserves feasibility of the voter locations and weakly improves the objective. Hence every finite-support feasible solution can be weakly improved, without changing its type locations, to one supported on at most two active voter types. Applying this to an approximating sequence handles the case where the full fractional relaxation does not attain its infimum, and shows that the two-location formulation has the same infimum value.
\end{proof}

\subsection{Bounds on the Optimization Problem \(\Efrac{\alpha}\)}
\label{sec:optimizer_e_alpha_soln}

We first state the lemma \ref{lemma:optimizer_mu_alpha_combined} and then restate the optimization problem from equation \eqref{eqn:optim_formulation_final} and recall that $\mu^\ast(\alpha) := \sup\left\{
\mu\in[0,1]: \Opt(\Efrac{\mu,\alpha})\geq 0 \right\}$.

\begin{equation*}
	\Efrac{\mu,\alpha}=
	\left\{
	\begin{aligned}
		\text{minimize} \quad
		& p(b_1-\mu w_1)+(1-p)(b_2-\mu w_2)\\
		\text{subject to} \quad
		& p\,g\left(\frac{b_1}{w_1}\right)
		+(1-p)g\left(\frac{b_2}{w_2}\right) \geq \alpha,\\
		& |b_k-w_k|\leq 1,\quad b_k+w_k\geq 1 \quad k=1,2,\\
		& p\in[0,1],\quad b_k,w_k\geq0 \quad k=1,2.
	\end{aligned}
	\right. 
\end{equation*}

\begin{lemma}[Restatement of Lemma~\ref{lemma:optimizer_mu_alpha_combined}]
	\optimizerepsilonalpha
\end{lemma}

We prove this by proving three separate results: lemma \ref{lemma:optimizer_mu_alpha} in \ref{sec:exact_solution}, claim~\ref{claim:two_point_global_lower_bound} in Section~\ref{sec:simple_lower_bound}  and claim \ref{claim:single_point_dominates}  in Section~\ref{sec:suffienicient_dominant_condition}.

\subsubsection{Exact Solution to \(\mu^\ast(\alpha)\)}
\label{sec:exact_solution}

\begin{lemma}\label{lemma:optimizer_mu_alpha}
	\[
	\mu^\ast(\alpha)
	=
	\min\left\{
	R(g,\alpha),
	\left(\frac{\gm}{\alpha}-1\right)^{-1}
	\right\}.
	\]
\end{lemma}

\begin{proof}
	We next rewrite the linearized program \(\Efrac{\mu,\alpha}\) in terms of the ratio
	\[
	r_k:=\frac{b_k}{w_k}
	\]
	for each active voter type \(k\in\{1,2\}\). This reduction is useful because, once \(r_k\) is fixed, the objective is linear in \(w_k\).
	
	For \(w_k>0\), the triangle constraints
	\[
	|b_k-w_k|\leq1,
	\qquad
	b_k+w_k\geq1
	\]
	are equivalent to
	\[
	\frac{1}{1+r_k}
	\leq
	w_k
	\leq
	\frac{1}{|1-r_k|},
	\]
	where the upper bound is vacuous when \(r_k=1\). The lower bound follows from \(b_k+w_k\geq1\), and the upper bound follows from \(|b_k-w_k|\leq1\). The endpoint cases \(r_k=0\) and \(r_k=\infty\) are interpreted by continuity.
	
	For fixed \(r_k>0\), the contribution of type \(k\) to the objective of \(\Efrac{\mu,\alpha}\) is
	\[
	b_k-\mu w_k
	=
	w_k\left(r_k-{\mu}\right).
	\]
	
	Thus the minimizing choice of \(w_k\) is determined by the sign of \(r_k-\mu\):
	\[
	w_k=
	\begin{cases}
		\dfrac{1}{1+r_k}, & r_k\geq\mu,\\[1.25ex]
		\dfrac{1}{1-r_k}, & r_k<\mu.
	\end{cases}
	\]
	Here the second case uses \(\mu\leq1\), so \(r_k<\mu\) implies \(r_k<1\).
	The corresponding values of \(b_k=r_kw_k\) are \(r_k/(1+r_k)\) and \(r_k/(1-r_k)\), respectively.
	
	Consequently, after optimizing over \(w_k\) for each fixed \(r_k\), the linearized objective may be written using the one-dimensional function
	\[
	\phi_\mu(r)
	:=
	\begin{cases}
		\dfrac{r-\mu}{1+r}, & r\geq\mu,\\[1.25ex]
		\dfrac{r-\mu}{1-r}, & 0\leq r<\mu.
	\end{cases}
	\]
	The two-type linearized program therefore reduces to
	\begin{equation}\label{eq:temp_reduced_r}
		\begin{aligned}
			\text{minimize}\quad
			& p\,\phi_\mu(r_1)+(1-p)\phi_\mu(r_2)\\
			\text{subject to}\quad
			& p\,g(r_1)+(1-p)g(r_2)\geq\alpha,\\
			& p\in[0,1],\quad r_1,r_2\in[0,\infty].
		\end{aligned}
	\end{equation}
	
	We use this reduced problem to identify the threshold at which the optimum changes sign. Since
	\[
	\mu^\ast(\alpha)=\sup\{\mu\in[0,1]:\Opt(\Efrac{\mu,\alpha})\ge0\},
	\]
	it is enough to characterize the feasible solutions with negative objective. Such a solution proves that the chosen value of \(\mu\) is too large.
	
	Finally, to keep the variables bounded, use the same change of variables for both active types:
	\[
	x_k:=\frac{r_k}{1+r_k}\in[0,1],
	\qquad
	r_k=\frac{x_k}{1-x_k}.
	\]
	Let \(\tau_\mu:=\mu/(1+\mu)\). The straddling condition becomes \(x_1>\tau_\mu>x_2\), and
	\[
	g(r_k)
	=
	g\left(\frac{x_k}{1-x_k}\right)
	=
	g_{\textsc{mid}}(x_k).
	\]
	In this bounded variable,
	\[
	\psi_\mu(x):= \phi_\mu\left(\frac{x}{1-x}\right)
	=
	\begin{cases}
		(1+\mu)x-\mu, & x\geq \tau_\mu,\\[1.25ex]
		\dfrac{(1+\mu)x-\mu}{1-2x}, & 0\leq x<\tau_\mu.
	\end{cases}
	\]
	Thus \(\psi_\mu(x)\) is negative exactly for \(x<\tau_\mu\), zero at \(x=\tau_\mu\), and positive for \(x>\tau_\mu\). Therefore, any feasible solution with negative objective must put positive mass below \(\tau_\mu\). If all active mass is below \(\tau_\mu\), then feasibility requires some \(x<\tau_\mu\) with \(g_{\textsc{mid}}(x)\ge\alpha\). By the definition of \(\gm=\sup_{x\in(0,1)}g_{\textsc{mid}}(x)/x\), this implies
	\[
	\tau_\mu>\frac{\alpha}{\gm},
	\qquad\text{or equivalently}\qquad
	\mu>\frac{\alpha}{\gm-\alpha}
	=
	\left(\frac{\gm}{\alpha}-1\right)^{-1},
	\]
	so this case cannot create a negative objective before the one-location threshold. The boundary construction with \(x_2=0\) realizes this threshold; its stationarity equations are derived below in \eqref{eq:temp_local_single_point_1}--\eqref{eq:temp_local_single_point_2}.
	
	It remains to analyze the genuine two-location case. We may relabel the two active locations so that \(x_1>\tau_\mu>x_2\). If \(g_{\textsc{mid}}(x_2)\ge\alpha\), then the lower location alone falls under the preceding case and cannot be the first obstruction before the one-location threshold. Hence, in the genuine two-location case, \(g_{\textsc{mid}}(x_2)<\alpha\); feasibility then forces \(g_{\textsc{mid}}(x_1)>\alpha\). If the probability constraint were slack, we could decrease the mass on the positive-objective location \(x_1\) and increase the mass on the negative-objective location \(x_2\) until the constraint is tight, strictly decreasing the objective unless we have returned to the one-location case. Thus the probability constraint is tight in the genuine two-location obstruction, and the mixing weight satisfies \(0<p<1\). Boundary cases with \(x_2=\tau_\mu\) contribute zero on the lower side and cannot yield a negative two-location obstruction. Boundary cases with \(p=0\) or \(p=1\) are precisely the one-location cases already handled.
	
	If the infimum over genuine two-location obstructions is not attained, take an approximating sequence. Any limit in which \(p\to0\) or \(p\to1\), or in which the lower point ceases to be strictly below \(\tau_\mu\), reduces to the one-location case or cannot keep a negative objective. If the upper point approaches \(\tau_\mu\) while feasibility is preserved, then \(g_{\textsc{mid}}(\tau_\mu)\ge\alpha\), which again implies \(\tau_\mu\ge\alpha/\gm\) and is covered by the one-location threshold. The boundary \(x_1=1\) is ruled out by the normal-cone argument below, while the boundary \(x_2=0\) is the one-location threshold already handled. Thus, after removing these boundary alternatives, any attained genuine two-location threshold is an interior local optimizer of the reduced smooth problem with an active smooth constraint. The KKT conditions below are therefore necessary for the only remaining type of global obstruction.
	
	The Lagrangian is defined as \(\mathcal L(\lambda, p, x_1, x_2) = p(\psi_\mu(x_1)- \lambda g_{\textsc{mid}}(x_1)) + (1-p)(\psi_\mu(x_2)- \lambda g_{\textsc{mid}}(x_2)) + \lambda \alpha\) for the dual variable \(\lambda \geq 0\). We now write the KKT stationarity conditions below.
	
	Since \(p\in(0,1)\), stationarity with respect to \(p\) gives
	\begin{equation}\label{eq:p_stationarity}
		\psi_{\mu}(x_1)-\lambda g_{\textsc{mid}}(x_1)
		=
		\psi_{\mu}(x_2)-\lambda g_{\textsc{mid}}(x_2).
	\end{equation}
	
	For the support locations, the stationarity conditions are
	\[
	0\in
	p\left(
	\partial_x\psi_{\mu}(x_1)
	-
	\lambda\,\partial_x g_{\textsc{mid}}(x_1)
	\right)
	+
	N_{(\tau_\mu,1]}(x_1),
	\]
	and
	\[
	0\in
	(1-p)\left(
	\partial_x\psi_{\mu}(x_2)
	-
	\lambda\,\partial_x g_{\textsc{mid}}(x_2)
	\right)
	+
	N_{[0,\tau_\mu)}(x_2).
	\]
	
	For a one-dimensional interval \(I\), we use \(N_I(x)\) to denote the
	normal cone of \(I\) at \(x\). In the two intervals appearing here, this is
	simply \[
	\begin{aligned}
		N_{(\tau_\mu,1]}(x)
		&=
		\begin{cases}
			\{0\}, & x\in(\tau_\mu,1),\\
			[0,\infty), & x=1,
		\end{cases}
		&
		N_{[0,\tau_\mu)}(x)
		&=
		\begin{cases}
			(-\infty,0], & x=0,\\
			\{0\}, & x\in(0,\tau_\mu).
		\end{cases}
	\end{aligned}
	\]

	The normal cone captures the first-order contribution of the interval
	constraint at boundary points. The upper boundary \(x_1=1\) cannot occur in a genuine two-location optimum: stationarity would require
	\((1+\mu)+c-\lambda g'_{\textsc{mid}}(1)=0\) for some \(c\ge0\), which is impossible because \(g'_{\textsc{mid}}(1)=0\). Hence \(x_1\in(\tau_\mu,1)\), and the stationarity condition at \(x_1\) gives
	\[
	1+\mu-\lambda g'_{\textsc{mid}}(x_1)=0.
	\]

	
	For \(0<x_2<\tau_\mu\), stationarity gives \(\psi'_\mu(x_2)=\lambda g'_{\textsc{mid}}(x_2)\). When \(x_2=0\), the same calculation is replaced by the boundary normal-cone condition. In either case, the \(x_1\)-stationarity equation gives \(\lambda>0\). Hence complementary slackness forces the probability constraint to hold with equality. Consequently, at a threshold solution with objective value zero, the Lagrangian equals the optimization objective.
	
	Since we are solving for $\mu$ such that the optimization objective equals zero, we must have 
	
	\begin{equation}\label{eq:setting_objective_zero}
		0 =
		\psi_\mu(x_1)-\lambda g_{\textsc{mid}}(x_1)+\lambda \alpha
		=
		\psi_\mu(x_2)-\lambda g_{\textsc{mid}}(x_2)+\lambda \alpha .
	\end{equation}

	This follows from setting the Lagrangian equal to zero and applying \eqref{eq:p_stationarity}.

	Eliminating $\lambda$ gives two sets of equations: i) \eqref{eq:temp_local_single_point_1}, \eqref{eq:temp_local_single_point_2}; and ii) \eqref{eq:lambda_equality}, \eqref{eq:temp_local_1}, and \eqref{eq:temp_local_2}. The first set is the boundary case \(x_2=0\), and the second is the interior case \(0<x_2<\tau_\mu\).
	
	Observe that $\psi'_\mu(x) = 1+ \mu$ if $x>\tau_\mu$ and $\frac{1}{(1-2x)^2}$ otherwise.

	\begin{align}
		x_1 &=
		\frac{g_{\textsc{mid}}(x_1)}
		{g'_{\textsc{mid}}(x_1)}
		\label{eq:temp_local_single_point_1}\\
		\frac{\mu}{1+\mu} &=
		\frac{\alpha}{g'_{\textsc{mid}}(x_1)}
		\label{eq:temp_local_single_point_2}.
	\end{align}
	In the boundary case \(x_2=0\), equations \eqref{eq:temp_local_single_point_1}--\eqref{eq:temp_local_single_point_2} give the one-location threshold. Indeed, \eqref{eq:temp_local_single_point_1} says that \(x_1\) maximizes \(g_{\textsc{mid}}(x)/x\), so \(g'_{\textsc{mid}}(x_1)=\gm\). Substituting this into \eqref{eq:temp_local_single_point_2} yields
	\[
	\frac{\mu}{1+\mu}=\frac{\alpha}{\gm},
	\qquad\text{and hence}\qquad
	\mu
	=
	\frac{\alpha}{\gm-\alpha}
	=
	\left(\frac{\gm}{\alpha}-1\right)^{-1}.
	\]
	
	\begin{align}
		\frac{1+\mu}{g'_{\textsc{mid}}(x_1)} &=
		\frac{1-\mu}{(1-2x_2)^2 g'_{\textsc{mid}}(x_2)}
		\label{eq:lambda_equality}\\
		\mu-\frac{(1+\mu)\alpha}{g'_{\textsc{mid}}(x_1)} &=
		(1+\mu)\left(
		x_1-\frac{g_{\textsc{mid}}(x_1)}
		{g'_{\textsc{mid}}(x_1)}
		\right)
		\label{eq:temp_local_1}\\
		&=
		\frac{x_2(1-\mu)}{1-2x_2}
		-
		\frac{(1-\mu)g_{\textsc{mid}}(x_2)}
		{g'_{\textsc{mid}}(x_2)(1-2x_2)^2}.
		\label{eq:temp_local_2}
	\end{align}
	
	The first constraint \eqref{eq:lambda_equality} comes from the $\lambda$ constraint obtained from the stationary condition of $x_1$ and $x_2$ respectively. The second constraint follows from \eqref{eq:setting_objective_zero} and substituting the value of $\lambda$.
	
	Equations \eqref{eq:lambda_equality}, \eqref{eq:temp_local_1}, and \eqref{eq:temp_local_2} can be simplified to obtain
	
	\begin{align}
		\frac{\mu}{1-\mu} &=
		\frac{x_2}{1-2x_2}
		+
		\frac{\alpha-g_{\textsc{mid}}(x_2)}
		{g'_{\textsc{mid}}(x_2)(1-2x_2)^2},
		\label{eq:temp_local_3}\\
		g'_{\textsc{mid}}(x_1) &=
		g'_{\textsc{mid}}(x_2)(1-2x_2)
		+
		2\left(\alpha-g_{\textsc{mid}}(x_2)\right),
		\label{eq:temp_local_4}\\
		g_{\textsc{mid}}(x_1)-\alpha &=
		g'_{\textsc{mid}}(x_1)\left(x_1-\C^{(g,\alpha)}(x_2)\right),
		\label{eq:temp_local_5}\\
		x_2
		&<
		\C^{(g,\alpha)}(x_2),
		\qquad
		x_1>\C^{(g,\alpha)}(x_2),
		\nonumber\\
		0
		&<
		\C^{(g,\alpha)}(x_2)
		<
		1/2.
		\nonumber
	\end{align}
	
	where $\C^{(g,\alpha)}(x_2) = \frac{\alpha - g_{\textsc{mid}}(x_2) + x_2(1-2x_2) g'_{\textsc{mid}}(x_2)}{2(\alpha-g_{\textsc{mid}}(x_2)) +  (1-2x_2) g'_{\textsc{mid}}(x_2)}$.
	
	The first equation \eqref{eq:temp_local_3} follows via elimination of \(x_1\) from \eqref{eq:temp_local_2} via substitution of \(g'_{\textsc{mid}}(x_1)\) from \eqref{eq:temp_local_1}. The second equation \eqref{eq:temp_local_4} follows from \eqref{eq:lambda_equality} on substitution of \(\frac{\mu}{1-\mu}\) computed from \eqref{eq:temp_local_3}. The third equation \eqref{eq:temp_local_5} follows from the first equality in \eqref{eq:temp_local_2} followed by the substitution of \(\frac{\mu}{1+\mu} = \C^{(g,\alpha)}(x_2)\) computed from \eqref{eq:temp_local_3}. 
	
	The inequalities follow from the fact that
	\[
	\tau_{\mu}=\frac{\mu}{1+\mu}=\C^{(g,\alpha)}(x_2)
	\]
	and the nontrivial two-type threshold case has \(0<\mu<1\). Thus \(0<\C^{(g,\alpha)}(x_2)<1/2\), and \(x_1>\tau_\mu>x_2\) gives \(x_2<\C^{(g,\alpha)}(x_2)\) and \(x_1>\C^{(g,\alpha)}(x_2)\).

	The reductions above show that every way for \(\Opt(\Efrac{\mu,\alpha})\) to become negative is accounted for by either the one-location obstruction or by a genuine two-location threshold satisfying the stationarity system. Conversely, any valid solution of \eqref{eq:temp_local_3}, \eqref{eq:temp_local_4}, and \eqref{eq:temp_local_5}, or of the boundary equations \eqref{eq:temp_local_single_point_1} and \eqref{eq:temp_local_single_point_2}, gives a feasible solution of \eqref{eq:temp_reduced_r} with objective value zero at the corresponding \(\mu\); increasing \(\mu\) makes the same feasible solution have negative objective. Therefore, the first value of \(\mu\) at which negative feasible solutions can appear is the minimum over the one-location threshold and the genuine two-location thresholds. This is precisely
	\[
	\mu^\ast(\alpha)
	=
	\min\left\{
	R(g,\alpha),
	\left(\frac{\gm}{\alpha}-1\right)^{-1}
	\right\}.
	\]
	If the two-location threshold is not attained, the same conclusion follows by an approximating sequence, which is why \(R(g,\alpha)\) is defined using an infimum.
	
\end{proof}

\subsubsection{Simpler Lower Bound on \(\mu^\ast(\alpha)\)}
\label{sec:simple_lower_bound}

\begin{claim}[A global lower bound ]
	\label{claim:two_point_global_lower_bound}
	We have
	\[
	R(g,\alpha)
	\ge
	\frac{\overline{L}(g,\alpha)}{1-\overline{L}(g,\alpha)}.
	\]
	Thus,
	\[
	\mu^\ast(\alpha)
	\ge
	\min\left\{
	\frac{\overline{L}(g,\alpha)}{1-\overline{L}(g,\alpha)},
	\left(\frac{\gm}{\alpha}-1\right)^{-1}
	\right\}.
	\]
\end{claim}

\begin{proof}
	Observe, 
	\[
	R(g,\alpha)
	=
	\inf_{(x_1,x_2)\in \mathcal{P}^{(g,\alpha)}}
	\frac{\C^{(g,\alpha)}(x_2)}{1-\C^{(g,\alpha)}(x_2)}.
	\]
	If \(\mathcal X(g,\alpha)=\emptyset\), then the containment argument below shows that \(\mathcal P^{(g,\alpha)}=\emptyset\). Hence \(R(g,\alpha)=+\infty\), and the claimed lower bound is immediate under the convention \(\overline L(g,\alpha)=1/2\).
	Hence assume \(\mathcal X(g,\alpha)\neq\emptyset\). For every \(x\in\mathcal X(g,\alpha)\), write \(A=\alpha-g_{\textsc{mid}}(x)>0\) and \(B=(1-2x)g'_{\textsc{mid}}(x)>0\). Then
	\[
	\C^{(g,\alpha)}(x)-x
	=
	\frac{(1-2x)A}{2A+B}
	>
	0,
	\qquad
	\C^{(g,\alpha)}(x)-\frac12
	=
	-\frac{(1-2x)^2g'_{\textsc{mid}}(x)}
	{2(2A+B)}
	<
	0.
	\]
	Thus \(0<\C^{(g,\alpha)}(x)<1/2\) for every \(x\in\mathcal X(g,\alpha)\), so \(0\leq\overline{L}(g,\alpha)<1/2\).
	If \(\mathcal P^{(g,\alpha)}=\emptyset\), then \(R(g,\alpha)=+\infty\) and the claimed lower bound is immediate. Hence assume \(\mathcal P^{(g,\alpha)}\neq\emptyset\).
	
	We first justify the containment. Let \((x_1,x_2)\in\mathcal P^{(g,\alpha)}\). The constraints give \(x_2\in(0,1/2)\) and \(x_2<\C^{(g,\alpha)}(x_2)<1/2\). To see that \(g_{\textsc{mid}}(x_2)<\alpha\), write \(A=\alpha-g_{\textsc{mid}}(x_2)\) and \(B=(1-2x_2)g'_{\textsc{mid}}(x_2)\). Since \(g_{\textsc{mid}}\) is increasing and \(x_2<1/2\), we have \(B>0\). Now
	\[
	\C^{(g,\alpha)}(x_2)-\frac12
	=
	-\frac{(1-2x_2)^2g'_{\textsc{mid}}(x_2)}
	{2\left(2A+B\right)}.
	\]
	The constraint \(\C^{(g,\alpha)}(x_2)<1/2\) therefore implies \(2A+B>0\). Since also
	\[
	\C^{(g,\alpha)}(x_2)-x_2
	=
	\frac{(1-2x_2)A}{2A+B},
	\]
	the constraint \(x_2<\C^{(g,\alpha)}(x_2)\) gives \(A>0\). Hence \(g_{\textsc{mid}}(x_2)<\alpha\), so \(x_2\in\mathcal X(g,\alpha)\).
	
	Thus, by the definition of \(\overline{L}(g,\alpha)\),
	\[
	\C^{(g,\alpha)}(x_2)\geq \overline{L}(g,\alpha)
	\]
	for every feasible pair. Since the function
	\[
	r\mapsto \frac{r}{1-r}
	\]
	is increasing on \([0,1)\), we obtain
	\[
	R(g,\alpha)
	\ge
	\frac{\overline{L}(g,\alpha)}{1-\overline{L}(g,\alpha)}.
	\]
	Using
	\[
	\mu^\ast(\alpha)
	=
	\min\left\{
	R(g,\alpha),
	\left(\frac{\gm}{\alpha}-1\right)^{-1}
	\right\},
	\]
	we obtain
	\[
	\mu^\ast(\alpha)
	\ge
	\min\left\{
	\frac{\overline{L}(g,\alpha)}{1-\overline{L}(g,\alpha)},
	\left(\frac{\gm}{\alpha}-1\right)^{-1}
	\right\}.
	\]
\end{proof}

\subsubsection{A sufficient condition for $m \gm - 1$ to dominate in $\mu^\ast(\alpha)$}\label{sec:suffienicient_dominant_condition}

\begin{claim}[A sufficient domination condition]
	\label{claim:single_point_dominates}
	Assume \(\mathcal{P}^{(g,\alpha)}\neq \emptyset\) and \(0<\alpha<1\). If
	\[
	\alpha
	\le
	\gm\,\overline{L}(g,\alpha),
	\]
	then
	\[
	\left(\frac{\gm}{\alpha}-1\right)^{-1}
	\le
	R(g,\alpha).
	\]
	Consequently,
	\[
	\mu^\ast(\alpha)
	=
	\left(\frac{\gm}{\alpha}-1\right)^{-1}.
	\]
\end{claim}

\begin{proof}
	By the containment argument in Claim~\ref{claim:two_point_global_lower_bound}, \(\mathcal X(g,\alpha)\neq\emptyset\) and \(0\leq\overline{L}(g,\alpha)<1/2\).
	Since \(x^{\ast}_{\textsc{mid}}\) maximizes \(g_{\textsc{mid}}(x)/x\) over \(x\in(0,1)\), and since \(g_{\textsc{mid}}(1/2)=g(1)=1/2\), we have
	\[
	\gm
	=
	\frac{g_{\textsc{mid}}(x^{\ast}_{\textsc{mid}})}{x^{\ast}_{\textsc{mid}}}
	\geq
	\frac{g_{\textsc{mid}}(1/2)}{1/2}
	=
	1.
	\]
	Thus \(0<\alpha<1\leq\gm\). Hence
	\[
	\left(\frac{\gm}{\alpha}-1\right)^{-1}
	=
	\frac{\alpha}{\gm-\alpha}.
	\]
	The assumed condition is
	\[
	\alpha
	\le
	\gm\,\overline{L}(g,\alpha).
	\]
	Equivalently,
	\[
	\alpha\left(1-\overline{L}(g,\alpha)\right)
	\le
	\overline{L}(g,\alpha)(\gm-\alpha).
	\]
	Since \(\gm-\alpha>0\) and \(1-\overline{L}(g,\alpha)>0\), this is equivalent to
	\[
	\frac{\alpha}{\gm-\alpha}
	\le
	\frac{\overline{L}(g,\alpha)}{1-\overline{L}(g,\alpha)}.
	\]
	Thus
	\[
	\left(\frac{\gm}{\alpha}-1\right)^{-1}
	\le
	\frac{\overline{L}(g,\alpha)}{1-\overline{L}(g,\alpha)}.
	\]
	By Claim~\ref{claim:two_point_global_lower_bound},
	\[
	R(g,\alpha)
	\ge
	\frac{\overline{L}(g,\alpha)}{1-\overline{L}(g,\alpha)}.
	\]
	Therefore
	\[
	\left(\frac{\gm}{\alpha}-1\right)^{-1}
	\le
	R(g,\alpha).
	\]
	Finally, since
	\[
	\mu^\ast(\alpha)
	=
	\min\left\{
	R(g,\alpha),
	\left(\frac{\gm}{\alpha}-1\right)^{-1}
	\right\},
	\]
	we conclude
	\[
	\mu^\ast(\alpha)
	=
	\left(\frac{\gm}{\alpha}-1\right)^{-1}.
	\]
\end{proof}

\subsubsection{Asymptotic Growth of \(\left(R^{(g,1/m)}\right)^{-1}\)}
\label{sec:asymptotic_r_scale_proof}

\begin{claim}[Restatement of Claim~\ref{claim:asymptotic_r_scale}]
	\claimasymptoticrscale
\end{claim}

\begin{proof}
	Let \(\rho(\alpha):=g_{\textsc{mid}}^{-1}(\alpha)\), which is well-defined for all sufficiently small \(\alpha>0\) by the local strict monotonicity and continuity of \(g_{\textsc{mid}}\) near \(0\). We first show that, for all sufficiently small \(\alpha\),
	\[
	\overline{L}^{(g,\alpha)}
	\geq
	c_0\,\rho(\alpha)
	\]
	for a constant \(c_0>0\) independent of \(\alpha\).
	
	Fix such an \(\alpha\), and write \(\rho=\rho(\alpha)\). Since \(\rho(\alpha)\downarrow0\), we may assume \(\rho<1/2\). For any \(x\in(0,\rho)\), we have \(A:=\alpha-g_{\textsc{mid}}(x)>0\) and \(1-2x>0\). The identities
	\[
	\C^{(g,\alpha)}(x)-x
	=
	\frac{(1-2x)A}
	{2A+(1-2x)g'_{\textsc{mid}}(x)}
	\]
	and
	\[
	\C^{(g,\alpha)}(x)-\frac12
	=
	-\frac{(1-2x)^2g'_{\textsc{mid}}(x)}
	{2\left(2A+(1-2x)g'_{\textsc{mid}}(x)\right)},
	\]
	then give \(x<\C^{(g,\alpha)}(x)<1/2\), because the denominator is positive. Hence \(x\in\mathcal X^{(g,\alpha)}\), so \(\mathcal X^{(g,\alpha)}\neq\emptyset\) for all sufficiently small \(\alpha\).
	
	Now take any \(x\in\mathcal X^{(g,\alpha)}\). Since \(g_{\textsc{mid}}(x)\leq\alpha\), monotonicity gives \(0<x\leq\rho\). The equality case cannot occur: if \(x=\rho\), then \(\alpha-g_{\textsc{mid}}(x)=0\), so the identity for \(\C^{(g,\alpha)}(x)-x\) gives \(\C^{(g,\alpha)}(x)=x\), contradicting \(x<\C^{(g,\alpha)}(x)\). Thus \(0<x<\rho\).
	
	Regular variation and convexity imply a local derivative bound: there exists \(K>0\) such that
	\[
	g'_{\textsc{mid}}(x)
	\leq
	K\frac{g_{\textsc{mid}}(x)}{x}
	\]
	for all sufficiently small \(x>0\). Indeed, by convexity,
	\[
	g'_{\textsc{mid}}(x)
	\leq
	\frac{g_{\textsc{mid}}(2x)-g_{\textsc{mid}}(x)}{x},
	\]
	and regular variation bounds \(g_{\textsc{mid}}(2x)/g_{\textsc{mid}}(x)\) for small \(x\). Since convexity and \(g_{\textsc{mid}}(0)=0\) imply that \(g_{\textsc{mid}}(x)/x\) is nondecreasing near \(0\), \(0<x<\rho\) gives
	\[
	\frac{g_{\textsc{mid}}(x)}{x}
	\leq
	\frac{g_{\textsc{mid}}(\rho)}{\rho}
	=
	\frac{\alpha}{\rho}.
	\]
	Thus
	\[
	g'_{\textsc{mid}}(x)
	\leq
	K\frac{\alpha}{\rho}.
	\]
	
	Let \(A=\alpha-g_{\textsc{mid}}(x)\). Since \(x\in\mathcal X^{(g,\alpha)}\), the preceding argument gives \(A>0\). The denominator in \(\C^{(g,\alpha)}(x)\) is bounded above by
	\[
	2A+(1-2x)g'_{\textsc{mid}}(x)
	\leq
	2\alpha+K\frac{\alpha}{\rho}
	\leq
	(K+2)\frac{\alpha}{\rho},
	\]
	where the last inequality uses \(\rho<1\), which holds for all sufficiently small \(\alpha\). For the numerator, convexity and \(g_{\textsc{mid}}(0)=0\) imply \(xg'_{\textsc{mid}}(x)\geq g_{\textsc{mid}}(x)\). Therefore
	\[
	\begin{aligned}
		\alpha-g_{\textsc{mid}}(x)
		+x(1-2x)g'_{\textsc{mid}}(x)
		&\geq \alpha-g_{\textsc{mid}}(x)+(1-2x)g_{\textsc{mid}}(x)\\
		&= \alpha-2xg_{\textsc{mid}}(x)\\
		&\geq \alpha-2\rho\alpha
		\geq \frac{\alpha}{2},
	\end{aligned}
	\]
	where the last inequality is just \(\rho\leq1/4\), which again holds for all sufficiently small \(\alpha\) because \(\rho(\alpha)\downarrow0\). Combining the two estimates gives
	\[
	\C^{(g,\alpha)}(x)
	\geq
	\frac{\rho}{2(K+2)}.
	\]
	Taking the infimum over \(x\in\mathcal X^{(g,\alpha)}\) proves
	\[
	\overline{L}^{(g,\alpha)}
	\geq
	\frac{1}{2(K+2)}g_{\textsc{mid}}^{-1}(\alpha).
	\]
	
	By Claim~\ref{claim:two_point_global_lower_bound},
	\[
	R^{(g,\alpha)}
	\geq
	\frac{\overline{L}^{(g,\alpha)}}
	{1-\overline{L}^{(g,\alpha)}}.
	\]
	Hence
	\[
	\left(R^{(g,\alpha)}\right)
	\geq
	\frac{\overline{L}^{(g,\alpha)}}{1-\overline{L}^{(g,\alpha)}}
	\geq
	{\overline{L}^{(g,\alpha)}}
	=
	\Omega\left(g_{\textsc{mid}}^{-1}(\alpha)\right).
	\]
	Substituting \(\alpha=1/m\) gives the first claim. If \(g_{\textsc{mid}}(x)\sim c x^\beta\), then \(g_{\textsc{mid}}^{-1}(\alpha)\sim(\alpha/c)^{1/\beta}\), and thus \(\left(R^{(g,1/m)}\right)^{-1}=O(m^{1/\beta})\).
\end{proof}

\subsection{Proof of Theorem \ref{theorem:thm_plurality_distortion_m}}{\label{sec:thm_plurality_distortion_m_proof}}

\noindent\textbf{Theorem~\ref{theorem:thm_plurality_distortion_m} (Restated):} \thmplu

\begin{proof}
	Recall that candidate $B \in \mathcal{A}$ minimizes the social cost.  The other candidates are denoted by $\{A_j\}_{j \in [m-1]}$.

	\begin{align}{\label{eqn:plu_distortion_m}}
		\textsc{dist}^{(g)}(\textsc{Plurality},n,m)
		&=
		\sup_{d \in \mathcal{M}(\mathcal{N} \cup \mathcal{A})}
		\Biggl(
			\sum_{j=1}^{m-1}\mathbb{P}[A_j \text{ wins}]
			\frac{\SC(A_j,d)}{\SC(B,d)}
			+ \mathbb{P}[B \text{ wins}]
		\Biggr).
	\end{align}

	For every \( j \in [m-1] \), we now bound the probability of \( A_j \) being the winner. This event requires that at least \( \frac{n}{m} \) voters select \( A_j \) as their top preference, which in turn implies that these voters rank \( A_j \) over \( B \).  
	To formalize this, we define a set of Bernoulli random variables \( \{Y_{i,j}\}_{i=1}^{n} \), where each \( Y_{i,j} \) represents the event that voter \( i \) ranks candidate \( A_j \) above \( B \).  
	Recalling Equation~\ref{eq:pairwise_probability}, we have:  
	\[
	Y_{i,j} \sim \text{Ber}\left(g\left(\frac{d(i,B)}{d(i,A_j)}\right)\right).
	\]  
	
	Therefore, we obtain the following bound on the probability of \( A_j \) being the winner:  
	\begin{equation}{\label{eq:boudning_prob_winner}}
		\mathbb{P}[A_j \text{ wins}] \leq \mathbb{P}\left(\sum_{i \in \N} Y_{i,j} \geq \frac{n}{m}\right).
	\end{equation}
	
	Define \( \alpha_j n \) as the expected number of voters ranking \( A_j \) above \( B \), where  
	\begin{equation}{\label{eqn:alpha_defn}}
		\alpha_j :=
		\frac{1}{n} \sum_{i \in \N} \mathbb{E}[Y_{i,j}]
		=
		\frac{1}{n}\sum_{i \in \N}
		g\left(\frac{d(i,B)}{d(i,A_j)}\right),
		\quad \forall j \in [m-1].
	\end{equation}

		Now, we apply Lemma~\ref{lemma:threshold_chernoff_bound} to the sum of Bernoulli random variables for every \(j\in[m-1]\) satisfying \(\alpha_j \leq \frac{1}{m}-\frac{n^{-1/2+\epsilon}}{m}\). In the notation of the lemma, we take \(q=m\), \(r=1\), and \(\rho=\alpha_j\). Since \(n\ge m^2\) and \(m\ge2\), the lemma's remaining conditions \(n/m>1\) and \(n^{1/2+\epsilon}/m\ge1\) both hold.
	\begin{align}{\label{eq:chernoff_bound_winning_prob}}
		\text{If }\alpha_j \leq \frac{1}{m} - \frac{n^{(-1/2+\epsilon)}}{m},\text{ then}\nonumber\\
		\mathbbm{P}[A_j \text{ wins}]
		&\leq \mathbb{P}\left(\sum_{i \in \N} Y_{i,j} \geq \frac{n}{m}\right) \nonumber\\
		&\leq {m\alpha_j}\exp\left(\frac{- n^{(\frac{1}{2}+\epsilon)} +2m}{(2n^{(\frac{1}{2}-\epsilon)}-1)m}\right). 
	\end{align}

		

		
		
		Let $S := \{j \in [m-1]: \alpha_j < \frac{1}{m} - \frac{n^{(-1/2+\epsilon)}}{m}\}$ i.e. $S$ denotes the indices of candidates for which $\alpha_j$ is less than $\frac{1}{m} - \frac{n^{(-1/2+\epsilon)}}{m}$.
		Now using Observation \ref{lemma:bounding_social_cost_ratio}, lemma \ref{lemma:optimizer_mu_alpha_combined} and the fact that $\alpha_j \geq \frac{1}{m} - \frac{n^{(-1/2+\epsilon)}}{m} \text{ for every $j \in [m-1]\setminus S$}$, we have the following bound:
	\begin{align}{\label{eq:bounding_Ratio_alpha_notS}}
		\frac{\SC(A_j,d)}{\SC(B,d)}
		\leq \max\left(
			\frac{m\gm}{1-n^{-(1/2-\epsilon)}} - 1,
			R\left(g,\frac{1- n^{-(\frac{1}{2}-\epsilon)}}{m}\right)^{-1}
		\right).
	\end{align}
		We now have
		\begin{align}{\label{eqn:bounding_plu_distortion_m_case_1}}
			\textsc{dist}^{(g)}(\textsc{Plurality},n,m)  
			& = \sup_{d \in \mathcal{M}(\mathcal{N} \cup \mathcal{A})} \Biggl( \sum_{j \in [m-1]\setminus S}\left(\mathbb{P}[A_j \text{ wins}] \frac{\SC(A_j,d)}{\SC(B,d)}\right)
			+ \mathbb{P}[B \text{ wins}] \nonumber\\
			& \qquad\qquad + \sum_{j \in S}\left(\mathbb{P}[A_j \text{ wins}] \frac{\SC(A_j,d)}{\SC(B,d)} \right) \Biggr)  \nonumber\\
			&  \overset{(a)}{\leq} \max\left(\max_{j \in [m-1]\setminus S} \frac{\SC(A_j,d)}{\SC(B,d)},1\right) \nonumber\\
			& \qquad + \sum_{j \in S} \Biggl(
			m\alpha_j\,\Delta_g(\alpha_j)
			\exp\left(\frac{- n^{(\frac{1}{2}+\epsilon)} +2m}{(2n^{(\frac{1}{2}-\epsilon)}-1)m}\right)\Biggr) \nonumber  \\
			&  \overset{(b)}{\leq}  m(m-1)\Gamma_g\left(\frac{1-n^{-(1/2-\epsilon)}}{m}\right)\exp\left(\frac{- n^{(\frac{1}{2}+\epsilon)} +2m}{(2n^{(\frac{1}{2}-\epsilon)}-1)m}\right) \nonumber\\
			& \qquad +  \max\left(
			\frac{m\gm}{(1-n^{-(1/2-\epsilon)})} - 1,
			R\left(g,\frac{1- n^{-(\frac{1}{2}-\epsilon)}}{m}\right)^{-1}
			\right). \nonumber
		\end{align}
		
		$(a)$ follows from the following observations. 
		
		\begin{itemize} 
			\item Apply Observation \ref{lemma:bounding_social_cost_ratio} to bound $\frac{\SC(A_j,d)}{\SC(B,d)}$. Since $\alpha_j \leq \frac{1}{m} - \frac{n^{(-1/2+\epsilon)}}{m}$  $ \forall j \in S$, apply Equation \eqref{eq:chernoff_bound_winning_prob} to bound $\mathbbm{P}[A_j \text{ wins}]$.
			\item $\sum\limits_{j \in [m-1]\setminus S}\left(\mathbb{P}[A_j \text{ wins}] \frac{\SC(A_j,d)}{\SC(B,d)}\right) + \mathbb{P}[B \text{ wins}] \leq \max\limits\left(\max\limits_{j \in [m-1]\setminus S} \frac{\SC(A_j,d)}{\SC(B,d)},1\right).$
		\end{itemize}
		
		$(b)$ follows from the fact that $|S| \leq m-1$, the definition of \(\Gamma_g\), and applying Equation \eqref{eq:bounding_Ratio_alpha_notS}. 
	\end{proof}

\section{Plurality Lower Bounds (Theorems~\ref{theorem:thm_plurality_distortion_m_lower_bound} and~\ref{theorem:thm_plurality_distortion_m_PL_lower_bound})}
\label{sec:thm_plurality_distortion_comb_m_lower_bound_proof}

This appendix section proves the lower bounds for \textsc{Plurality}. We give two constructions. The first applies to the adversarial probabilistic-voting model with pairwise marginals induced by \(g\), and proves Theorem~\ref{theorem:thm_plurality_distortion_m_lower_bound}. The second specializes to the PL model and proves Theorem~\ref{theorem:thm_plurality_distortion_m_PL_lower_bound}. We also prove Observation~\ref{obs:generic_lb}, which gives a two-candidate lower bound for any well-behaved deterministic voting rule under probabilistic voting.

\subsection{Proof of Theorem~\ref{theorem:thm_plurality_distortion_m_lower_bound}}
\label{sec:thm_plurality_distortion_m_lower_bound_proof}
\noindent\textbf{Theorem~\ref{theorem:thm_plurality_distortion_m_lower_bound} (Restated):} \thmplulb

\begin{proof}
We give a construction in a Euclidean metric space in \(\mathbb R^3\). Candidate \(W\) is placed at \((1,0,0)\). The other \(m-1\) candidates \(B_1,B_2,\ldots,B_{m-1}\) are placed equidistantly on the circle of radius \(\epsilon\) in the \(y\)-\(z\) plane centered at the origin. For every voter on the \(x\)-axis, the distance to every \(B_j\) is the same. Hence, after conditioning on whether \(W\) is ranked above the \(B_j\)'s, we may rank the \(B_j\)'s uniformly at random among themselves; this satisfies the pairwise probability rule among them because their distances from such a voter are equal.

Fix \(\delta>0\), and choose a feasible two-type solution \((p,b_1,w_1,b_2,w_2)\) for \(\Efrac{1/m}\) whose objective value is at most \(\Opt(\Efrac{1/m})+\delta=\mu^\ast(1/m)+\delta\).
Thus
\[
p\,g\left(\frac{b_1}{w_1}\right)
+(1-p)g\left(\frac{b_2}{w_2}\right)
\ge
\frac1m
\]
and
\[
\frac{p b_1+(1-p)b_2}{p w_1+(1-p)w_2}
\le
\mu^\ast(1/m)+\delta.
\]
By the reduction in Appendix~\ref{sec:optimizer_e_alpha_soln}, the active types may be chosen on the line through \(W\) and the center of the \(B_j\)'s. Concretely, for each \(k\in\{1,2\}\), either \(b_k+w_k=1\), in which case we place the type-\(k\) voters at \((b_k,0,0)\), or \(w_k-b_k=1\), in which case we place them at \((-b_k,0,0)\). In both cases the distance to \(W\) is \(w_k\), and the distance to the origin is \(b_k\).

Let \(\beta_\zeta:=\zeta/(m-1)\), where \(0<\zeta<m-1\). A fraction \(\beta_\zeta\) of voters is placed at \(W\). The remaining voters are split between the two active types with fractions \((1-\beta_\zeta)p\) and \((1-\beta_\zeta)(1-p)\), up to the usual integer rounding.

A voter at \(W\) ranks \(W\) first with probability one. A voter of active type \(k\) ranks \(W\) first with probability
\[
q_{k,\epsilon}
:=
g\left(
\frac{\sqrt{b_k^2+\epsilon^2}}{w_k}
\right),
\]
with the convention that \(q_{k,\epsilon}=1\) when \(w_k=0\). The voter ranks \(W\) last with probability \(1-q_{k,\epsilon}\); conditional on either event, the candidates \(B_1,\ldots,B_{m-1}\) are uniformly randomly permuted in the remaining positions.

For each active type, the distance to every \(B_j\) is \(\sqrt{b_k^2+\epsilon^2}\), and the distance to \(W\) is \(w_k\). Hence this ranking distribution satisfies the pairwise probability criterion. By continuity of \(g\),
\[
\lim_{\epsilon\to0}
\left(
p q_{1,\epsilon}
+(1-p)q_{2,\epsilon}
\right)
\ge
\frac1m.
\]
Therefore, for every fixed \(\zeta>0\), choosing \(\epsilon>0\) sufficiently small gives a limiting expected plurality-score fraction for \(W\) strictly larger than \(1/m\). In fact, as \(\epsilon\to0\) this limiting fraction is at least \(\beta_\zeta+(1-\beta_\zeta)/m=(1+\zeta)/m\).
The law of large numbers then implies that \(W\) wins with probability tending to one as \(n\to\infty\).

For every \(j\in[m-1]\), after first taking \(n\to\infty\), the social-cost ratio in this construction is
\[
\frac{\SC(W,d)}{\SC(B_j,d)}
=
\frac{
(1-\beta_\zeta)\left(pw_1+(1-p)w_2\right)
}{
(1-\beta_\zeta)
\left(
p\sqrt{b_1^2+\epsilon^2}
+(1-p)\sqrt{b_2^2+\epsilon^2}
\right)
+\beta_\zeta\sqrt{1+\epsilon^2}
}.
\]
Letting \(\epsilon\to0\) and then \(\zeta\to0\), this ratio tends to
\[
\frac{p w_1+(1-p)w_2}{p b_1+(1-p)b_2}
\ge
\frac{1}{\mu^\ast(1/m)+\delta}.
\]
Finally, letting \(\delta\to0\) gives
\[
\lim_{n\to\infty}
\textsc{dist}^{(g)}(\textsc{Plurality},n,m)
\ge
\left(\mu^\ast(1/m)\right)^{-1}.
\]
By Lemma~\ref{lemma:optimizer_mu_alpha_combined},
\[
\left(\mu^\ast(1/m)\right)^{-1}
=
\max\left\{
m\gm-1,
\left(R(g,1/m)\right)^{-1}
\right\}.
\]
This proves the theorem.
\end{proof}

\subsection{Proof of Observation \ref{obs:generic_lb}}{\label{sec:proof_obs_generic_lb}}

We first define well-behaved deterministic voting rules for the case of two candidates.
\begin{definition} [Well-Behaved Deterministic Voting Rule for Two Candidates] \label{def:wellbehaved}
    Let \(p_{W \succ B}\) represent the fraction of voters who rank candidate \(W\) over \(B\). A well-behaved deterministic voting rule satisfies the following criteria:

\begin{itemize}\label{list:voting_rule_properties}
    \item \label{item:deterministic} The voting rule is a deterministic function \(\mathbb F\) of \(p_{W \succ B}\), the fraction of voters ranking \(W\) over \(B\).
    \item \label{item:discontinuities} The function \(\mathbb F: [0,1] \rightarrow \A\) has at most finitely many discontinuities.
\end{itemize}
\end{definition}

Virtually all well-studied deterministic voting rules are well-behaved for two candidates.

\noindent\textbf{Observation~\ref{obs:generic_lb} (Restated):} \obsgenericlb

\begin{proof}
Let \(D\subset[0,1]\) be the finite set of discontinuities of \(\mathbb F\). Fix \(\eta>0\), and choose a continuity point \(p_\eta\in(1/2,1/2+\eta)\setminus D\). Since \(\mathbb F\) takes values in the two-point set \(\{B,W\}\), continuity at \(p_\eta\) implies that \(\mathbb F\) is locally constant around \(p_\eta\). Hence, if \(p_{W\succ B}^{(n)}\to p_\eta\) in probability, then \(\mathbb F(p_{W\succ B}^{(n)})\to\mathbb F(p_\eta)\) in probability. The winner indicator is bounded, so the corresponding expectations also converge.

We now build two perturbed metric instances with the same limiting value \(p_\eta\). Let \(\beta_\eta:=2p_\eta-1\), so \(\beta_\eta\downarrow0\) as \(\eta\downarrow0\). Fix \(\delta>0\), and choose an approximately optimal two-type solution \((p,b_1,w_1,b_2,w_2)\) for \(\Efrac{1/2}\), with the probability constraint tight, whose objective value is at most \(\mu^\ast(1/2)+\delta\). Thus
\[
p\,g\left(\frac{b_1}{w_1}\right)
+(1-p)g\left(\frac{b_2}{w_2}\right)
=
\frac12,
\qquad
\frac{p b_1+(1-p)b_2}{p w_1+(1-p)w_2}
\le
\mu^\ast(1/2)+\delta.
\]
Such near-optimal tight solutions are obtained from the threshold characterization in Appendix~\ref{sec:optimizer_e_alpha_soln}; if the optimum is not attained, take a sequence approaching it.

For the first instance, place \(B\) at \(0\) and \(W\) at \(1\). Realize each active type on the line so that type \(k\) voters have distance \(b_k\) from \(B\) and \(w_k\) from \(W\). A fraction \(\beta_\eta\) of voters is placed at \(W\), and the remaining voters are split between the two active types in proportions \(p\) and \(1-p\), up to integer-rounding error. The limiting expected fraction ranking \(W\) above \(B\) is \(\beta_\eta+(1-\beta_\eta)/2=p_\eta\), so \(p_{W\succ B}^{(n)}\to p_\eta\) in probability. If \(\mathbb F(p_\eta)=W\), then this instance makes the rule choose \(W\) with probability tending to one, while
\[
\frac{\SC(W,d)}{\SC(B,d)}
\to
\frac{p w_1+(1-p)w_2}
{p b_1+(1-p)b_2}
\ge
\frac{1}{\mu^\ast(1/2)+\delta}
\]
after sending \(\eta\downarrow0\).

For the second instance, interchange the roles of the candidates in the active-type placement: type \(k\) voters have distance \(b_k\) from \(W\) and \(w_k\) from \(B\). Again place a fraction \(\beta_\eta\) of voters at \(W\), and split the remaining voters between the two active types. The limiting fraction ranking \(W\) above \(B\) is still \(p_\eta\). If \(\mathbb F(p_\eta)=B\), then the rule chooses \(B\) with probability tending to one, and the distortion tends, as \(\eta\downarrow0\), to at least the same reciprocal \(1/(\mu^\ast(1/2)+\delta)\).

Thus, for every \(\eta\) and \(\delta\), one of the two instances gives limiting distortion at least \(1/(\mu^\ast(1/2)+\delta)\), up to an error that vanishes with \(\eta\). Letting \(\delta\downarrow0\) and using Lemma~\ref{lemma:optimizer_mu_alpha_combined},
\[
\left(\mu^\ast(1/2)\right)^{-1}
=
\max\left\{
2\gm-1,
\left(R^{(g,1/2)}\right)^{-1}
\right\},
\]
which proves the observation.
\end{proof}

\subsection{Proof of Theorem \ref{theorem:thm_plurality_distortion_m_PL_lower_bound}}{\label{sec:thm_plurality_distortion_m_PL_lower_bound_proof}}

\noindent\textbf{Theorem~\ref{theorem:thm_plurality_distortion_m_PL_lower_bound} (Restated):} \thmplulbpl

\begin{proof}
    The proof is by an example in a Euclidean metric space in $\mathbb{R}^3$. One candidate \(C\) is at $(1,0,0)$. The other $m-1$ candidates are ``good'' and are equidistantly placed on a circle of radius $\epsilon$ on the $y$-$z$ plane centered at $(0,0,0)$. We call them  $\mathcal{G}:=\{G_1,G_2,\ldots,G_{m-1}\}$. 

We present a construction below for every $\epsilon,\zeta>0$ for a chosen value of $\hat{x} = 1- m^{-\frac{1}{\theta}}$.

Let
\[
    q:=
    \frac{(1-\hat{x})^{-\theta}}
    {(m-1)\left(\sqrt{\hat{x}^2+\epsilon^2}\right)^{-\theta}
    +(1-\hat{x})^{-\theta}}
    \quad\text{and}\quad
    a:= \frac{1}{m-1} \left(1 - \frac{1+\zeta}{m q} \right).
\]
Each of the $m-1$ candidates in $\mathcal{G}$ has $\floor{an}$ voters overlapping with it. 
The remaining voters (we call them ``ambivalent'') are placed at $(\hat{x},0,0).$ 
%
%
%
%
%
Clearly, each voter overlapping with a candidate votes for it as the most preferred candidate with probability one. Each ambivalent voter votes according to the PL model. Recall that $\mathbb{P}[A_j \text{~is top-ranked in~} \sigma_i]  = \frac{d(i,A_j)^{-\theta}}{\sum_{A_k \in \A} d(i,A_k)^{-\theta} }$ \cite{azari2012random}.

\begin{itemize}
    \item With probability $q$, vote for candidate $C$ as the top choice.
    \item With probability $1-q$, vote for any of the candidates in $\mathcal{G}$ as the top choice uniformly at random.
\end{itemize}

Since $\lim_{n \to \infty} \floor{an}/n =a$ and that the distance of a candidate in $\mathcal{G}$ from any non-ambivalent voter is at most $2\epsilon$, we have that for every $j \in [m-1]$,
\begin{align}\label{eq:bound_ratio_social_cost_PL}
    \lim_{n \to \infty} \frac{\SC(C,d)}{\SC(G_j,d)} \geq & \frac{(1-\hat{x})(1-(m-1)a)+(m-1)a \sqrt{1+ \epsilon^2}}{(1-(m-1)
    a)\sqrt{(\hat{x})^2+\epsilon^2} +  2(m-2)a\epsilon}\\
    & = \frac{(mq-(1+\zeta))\sqrt{1+ \epsilon^2} + (1+\zeta)(1-\hat{x}) }{(1+\zeta)\sqrt{(\hat{x})^2+\epsilon^2} + 2(m-2)a\epsilon}.
\end{align}

Clearly every candidate in $\mathcal{G}$ minimizes the social cost and now we show that $\lim\limits_{n \to \infty} \mathbb{P}[C \text{ wins}] =1$.

	Let Bernoulli random variables $\{Y_i\}_{i=1}^{n}$ denote the events that voter $i \in \N$ ranks candidate $C$ at the top. Here, $\sum_{i \in \N} \mathbb{P}[Y_i=1] = q(n - (m-1)\floor{a n})$ and thus
	\[
	\lim_{n \to \infty} \frac{\sum_{i=1}^{n} \mathbb{P}[Y_i=1]}{n} = \frac{1+\zeta}{m}.
	\]
	For every good candidate \(G_j\), its limiting expected top-score fraction is
	\[
	    a+\frac{1-q}{m-1}\left(1-(m-1)a\right)
	    =
	    \frac{1}{m}-\frac{\zeta}{m(m-1)}.
	\]
	Thus \(C\)'s limiting expected top-score fraction is separated from that of every good candidate by \(\zeta/(m-1)\). Since the top-score of each candidate is a sum of independent bounded random variables, the law of large numbers and a union bound over the \(m-1\) good candidates imply that \(\lim_{n \to \infty}\mathbb{P}[C \text{ wins}]=1\). Moreover, since the bound in Equation \eqref{eq:bound_ratio_social_cost_PL} holds true for every $\zeta$ and $\epsilon$, we can make them arbitrarily small.
 \begin{align}
    \lim_{n \to \infty} \textsc{dist}^{\theta}_{\textsc{PL}}(\textsc{Plurality},n,m) & \geq \frac{m(1-\hat{x})^{-\theta}}{\hat{x}((m-1)\hat{x}^{-\theta} + (1-\hat{x})^{-\theta})} -1\\
    &= \frac{m m}{(1-m^{-\frac{1}{\theta}})((m-1)(1-m^{-\frac{1}{\theta}})^{-\theta} + m) } = \Omega(m) .
\end{align}

The last equality follows by substituting \(\hat{x}=1-m^{-1/\theta}\).

\end{proof}

\section{Proof of the Copeland Upper Bound (Theorem~\ref{theorem:Copeland_distrotion_m})}

This appendix section proves the upper bound for \textsc{Copeland} and records the related lower-bound construction. We first justify the three-candidate fractional program \(\mathcal T_{\alpha_1,\alpha_2}\) and show that finite-support feasible solutions can be compressed to at most three active voter types (Subsection~\ref{sec:copeland_three_active_voter_types_proof}). We then use this program to prove Theorem~\ref{theorem:Copeland_distrotion_m}. Finally, Subsection~\ref{sec:copeland_lower_bound_proof} gives the corresponding lower-bound construction in terms of \(\mathcal T_{1/2,1/2}\).

\subsection{Three Active Voter Types for the Copeland Program}
\label{sec:copeland_three_active_voter_types_proof}

We first justify the fractional formulation \(\mathcal{T}_{\alpha_1,\alpha_2}\) used in the main text. For a finite electorate, define \(\mathcal{T}^{(n)}_{\alpha_1,\alpha_2}\) as follows. The variables \(b_i,y_i,w_i\) represent the distances from voter \(i\) to \(B,Y,W\), respectively, and \(D_{BY},D_{YW},D_{BW}\) represent the three candidate-candidate distances.
\[
\mathcal{T}^{(n)}_{\alpha_1,\alpha_2}
=
    \left\{
    \begin{aligned}
    \text{minimize}\quad
        & \frac{\sum_{i\in\N} b_i}{\sum_{i\in\N} w_i}\\
    \text{subject to}\quad
        & \sum_{i\in\N} g\left(\frac{y_i}{w_i}\right)\geq n\alpha_1,\qquad
          \sum_{i\in\N} g\left(\frac{b_i}{y_i}\right)\geq n\alpha_2,\\
        & |b_i-y_i|\leq D_{BY}\leq b_i+y_i \quad \forall i\in\N,\\
        & |y_i-w_i|\leq D_{YW}\leq y_i+w_i \quad \forall i\in\N,\\
        & |b_i-w_i|\leq D_{BW}\leq b_i+w_i \quad \forall i\in\N,\\
        & D_{BY},D_{YW},D_{BW} \text{ satisfy the triangle inequalities},\\
        & b_i,y_i,w_i\geq0 \quad \forall i\in\N.
    \end{aligned}
    \right.
\]

Now fix a finite set \(S\) of feasible three-candidate voter types. A type \(t\in S\) is a triple \((b_t,y_t,w_t)\) that can arise as the distances from a voter to candidates \(B,Y,W\), together with common candidate-candidate distances \(D_{BY},D_{YW},D_{BW}\). The fractional program over this fixed set is
\[
    \mathcal T_{\alpha_1,\alpha_2}(S)=
    \left\{
    \begin{aligned}
    \text{minimize}\quad
        & \frac{\sum_{t\in S}p_t b_t}{\sum_{t\in S}p_t w_t}\\
    \text{subject to}\quad
        & \sum_{t\in S}p_t g\left(\frac{y_t}{w_t}\right)\geq\alpha_1,\\
        & \sum_{t\in S}p_t g\left(\frac{b_t}{y_t}\right)\geq\alpha_2,\\
        & \sum_{t\in S}p_t=1,\quad p_t\geq0 \quad \forall t\in S.
    \end{aligned}
    \right.
\]

This reduction is the direct analogue of the two-active-type reduction for \textsc{Plurality} in Lemma~\ref{lemma:two_active_voter_types}. The only difference is that the Copeland two-hop case has two pairwise probability constraints instead of one, so the basic feasible support bound increases from two to three.

\begin{lemma}[Three active voter types]\label{lemma:copeland_three_active_voter_types}
For every fixed finite set of feasible three-candidate voter types \(S\), the optimization over the fractions \(p_t\) in \(\mathcal T_{\alpha_1,\alpha_2}(S)\) has an optimal basic feasible solution supported on at most three voter types, whenever the optimum is attained. Consequently, the full fractional relaxation has the same optimal value as the three-location formulation \eqref{eqn:copeland_three_candidate_fractional}.
\end{lemma}

\begin{proof}
Let \(q_t^{(1)}:=g(y_t/w_t)\) and \(q_t^{(2)}:=g(b_t/y_t)\). As in the plurality reduction, it is enough to consider feasible solutions with \(W(p):=\sum_{t\in S}p_t w_t>0\). Set \(r_t:=p_t/W(p)\). Then \(\sum_t r_t w_t=1\), the objective becomes \(\sum_t r_t b_t\), and the two probability constraints become
\[
    \sum_{t\in S}r_t(q_t^{(1)}-\alpha_1)\geq0,
    \qquad
    \sum_{t\in S}r_t(q_t^{(2)}-\alpha_2)\geq0.
\]
Thus, after fixing the voter types, the fractional optimization over voter fractions is equivalent to the linear program
\[
    \begin{aligned}
    \text{minimize}\quad & \sum_{t\in S}r_t b_t\\
    \text{subject to}\quad
        & \sum_{t\in S}r_t w_t=1,\\
        & \sum_{t\in S}r_t(q_t^{(1)}-\alpha_1)\geq0,\qquad
          \sum_{t\in S}r_t(q_t^{(2)}-\alpha_2)\geq0,\\
        & r_t\geq0 \quad \forall t\in S.
    \end{aligned}
\]
Every attained optimum has an optimal basic feasible solution. A basic feasible solution has support size at most the number of active non-negativity-excluding constraints: the denominator-normalizing equality and the two probability inequalities. Hence the support size is at most three. The transformation between \(p\) and \(r\) preserves support, and the candidate-candidate distance witnesses \(D_{BY},D_{YW},D_{BW}\) are unchanged. Optimizing also over type locations therefore preserves the same three-active-type bound.
\end{proof}

\begin{observation}[Restatement of Observation~\ref{obs:copeland_three_candidate_ratio}]
\obscopelandthreecandidateratio
\end{observation}

\begin{proof}
For a given triple \(W,Y,B\), set \(b_i=d(i,B)\), \(y_i=d(i,Y)\), and \(w_i=d(i,W)\). The metric triangle inequalities among \(B,Y,W\) and voter \(i\) give the feasibility constraints in \(\mathcal{T}^{(n)}_{\alpha_1,\alpha_2}\). The two probability constraints are exactly
\[
    \sum_{i\in\N}g\left(\frac{y_i}{w_i}\right)
    =
    \sum_{i\in\N}\mathbbm P[W\succ_{\sigma_i}Y]\geq n\alpha_1,
    \qquad
    \sum_{i\in\N}g\left(\frac{b_i}{y_i}\right)
    =
    \sum_{i\in\N}\mathbbm P[Y\succ_{\sigma_i}B]\geq n\alpha_2.
\]
Thus this candidate triple induces a feasible finite solution with objective \(\SC(B,d)/\SC(W,d)\). Passing to the fractional relaxation and then using Lemma~\ref{lemma:copeland_three_active_voter_types} gives
\[
    \Opt(\mathcal{T}_{\alpha_1,\alpha_2})
    \leq
    \frac{\SC(B,d)}{\SC(W,d)}.
\]
Taking reciprocals proves the observation.
\end{proof}

\begin{claim}[Restatement of Claim~\ref{claim:copeland_program_product_lower_bound}]
\claimcopelandprogramproductlowerbound
\end{claim}

\begin{proof}
Consider any feasible solution of \(\mathcal{T}_{\alpha_1,\alpha_2}\), and write \(B:=\sum_{\ell=1}^{3}p_\ell b_\ell\), \(Y:=\sum_{\ell=1}^{3}p_\ell y_\ell\), and \(W:=\sum_{\ell=1}^{3}p_\ell w_\ell\). The projection of this solution onto the pair \((Y,W)\) is feasible for \(\Efrac{\alpha_1}\) \footnote{While Equation \eqref{eqn:optim_formulation_two_type} is formulated for two active voter types, we show in Appendix \ref{sec:two_active_voter_types_proof} that increasing the number of types does not reduce the objective value.}, since \(\sum_{\ell=1}^{3}p_\ell g(y_\ell/w_\ell)\geq\alpha_1\). Therefore \(Y/W\geq\Opt(\Efrac{\alpha_1})\), with the same ratio convention used in the plurality program. Similarly, the projection onto the pair \((B,Y)\) is feasible for \(\Efrac{\alpha_2}\), and hence \(B/Y\geq\Opt(\Efrac{\alpha_2})\). Multiplying the two inequalities gives
\[
    \frac{B}{W}
    =
    \frac{B}{Y}\cdot\frac{Y}{W}
    \geq
    \Opt(\Efrac{\alpha_1})\Opt(\Efrac{\alpha_2}).
\]
Since the argument applies to every feasible solution of \(\mathcal{T}_{\alpha_1,\alpha_2}\), taking the infimum over feasible solutions proves the product lower bound. The reciprocal upper bound follows immediately in the nondegenerate cases used in the distortion bound.
\end{proof}

We next record two small technical facts used in the Copeland lower-bound construction. The construction first works with constraints \(1/2+\zeta\), so that the two relevant pairwise victories hold with positive slack and hence with probability tending to one after rounding to a finite electorate. To recover the advertised \(1/2\)-threshold bound, we must then let \(\zeta\downarrow0\). The next lemma supplies a voter type that can boost both pairwise constraints while preserving any fixed nondegenerate candidate geometry; the following lemma uses it to prove right-continuity of the three-candidate program at the majority threshold.

\begin{lemma}[Strict booster type for a fixed candidate geometry]\label{lemma:copeland_strict_booster_type}
Fix candidate-candidate distances \(D_{BY},D_{YW},D_{BW}>0\) satisfying the triangle inequalities. Then there is a feasible voter type \((b,y,w)\) for this same candidate geometry such that
\[
    g\left(\frac{y}{w}\right)>1/2,
    \qquad
    g\left(\frac{b}{y}\right)>1/2.
\]
\end{lemma}

\begin{proof}
Choose \(s>0\) so small that
\[
    s\leq D_{BY},\qquad s\leq D_{YW},\qquad 2s\leq D_{BW}.
\]
Then choose \(R\) large enough that
\[
    D_{BY}\leq 2R+3s,\qquad
    D_{YW}\leq 2R+s,\qquad
    D_{BW}\leq 2R+2s.
\]
Set
\[
    w=R,\qquad y=R+s,\qquad b=R+2s.
\]
The inequalities above imply
\[
    |b-y|\leq D_{BY}\leq b+y,\qquad
    |y-w|\leq D_{YW}\leq y+w,\qquad
    |b-w|\leq D_{BW}\leq b+w,
\]
so \((b,y,w)\) is jointly feasible with the fixed candidate geometry. Moreover \(y>w\) and \(b>y\), so \(y/w>1\) and \(b/y>1\). By monotonicity and \(g(1)=1/2\), both displayed pairwise probabilities are strictly larger than \(1/2\).
\end{proof}

We now apply the booster type to show that imposing a strict majority \(1/2+\zeta\) changes the program value by a vanishing amount as \(\zeta\downarrow0\). This is the step used at the end of Subsection~\ref{sec:copeland_lower_bound_proof} to pass from the slack lower-bound construction back to \(\mathcal{T}_{1/2,1/2}\).

\begin{lemma}[Right-continuity at the majority threshold]\label{lemma:copeland_T_right_continuity}
Let \(V(\alpha):=\Opt(\mathcal{T}_{\alpha,\alpha})\). Then
\[
    \lim_{\zeta\downarrow0}V(1/2+\zeta)=V(1/2).
\]
\end{lemma}

\begin{proof}
Since increasing \(\alpha\) only tightens the feasible set, \(V(\alpha)\) is nondecreasing, and hence
\[
    \liminf_{\zeta\downarrow0}V(1/2+\zeta)\geq V(1/2).
\]
It remains to prove the reverse inequality for the one-sided limit. Work first with the full fractional relaxation, where the support size is unrestricted, before applying Lemma~\ref{lemma:copeland_three_active_voter_types}. Fix \(\varepsilon>0\) and take an \(\varepsilon\)-optimal fractional solution for \(V(1/2)\). By an arbitrarily small perturbation of the common candidate-candidate distances and the associated voter distances, if necessary, we may assume that the witnessing candidate geometry is nondegenerate; the constraints and objective change continuously under this perturbation.

By Lemma~\ref{lemma:copeland_strict_booster_type}, for this same candidate geometry there is a feasible booster type \((b^\circ,y^\circ,w^\circ)\) satisfying
\[
    q_1^\circ:=g\left(\frac{y^\circ}{w^\circ}\right)>1/2,
    \qquad
    q_2^\circ:=g\left(\frac{b^\circ}{y^\circ}\right)>1/2.
\]
Mix this type into the fractional solution with mass \(\lambda\), and scale the old masses by \(1-\lambda\). If the original two constraint values are at least \(1/2\), then after mixing they are at least
\[
    (1-\lambda)\frac12+\lambda q_r^\circ
    =
    \frac12+\lambda\left(q_r^\circ-\frac12\right),
    \qquad r=1,2.
\]
Thus, for every sufficiently small \(\zeta>0\), choosing
\[
    \lambda
    =
    \max_{r=1,2}
    \frac{\zeta}{q_r^\circ-1/2}
    =
    O(\zeta)
\]
makes both constraints at least \(1/2+\zeta\). To see that the objective changes continuously, write the original fractional solution as having average distances \((B_0,Y_0,W_0)\). After mixing in the booster type, the objective is
\[
    \frac{(1-\lambda)B_0+\lambda b^\circ}
         {(1-\lambda)W_0+\lambda w^\circ}.
\]
The denominator is positive for all sufficiently small \(\lambda\), and this ratio converges to \(B_0/W_0\) as \(\lambda\downarrow0\). Thus the perturbation and mixture costs vanish as their sizes go to zero. Hence
\[
    \limsup_{\zeta\downarrow0}V(1/2+\zeta)\leq V(1/2)+\varepsilon.
\]
Finally, Lemma~\ref{lemma:copeland_three_active_voter_types} compresses the resulting fractional solution back to at most three active voter types without increasing the objective. Sending \(\varepsilon\downarrow0\) proves the lemma.
\end{proof}

\subsection{Proof of Theorem~\ref{theorem:Copeland_distrotion_m}}
\label{sec:Copeland_distrotion_m_proof}

\noindent\textbf{Theorem~\ref{theorem:Copeland_distrotion_m} (Restated):} \thmcop

\begin{proof}[Proof of Theorem \ref{theorem:Copeland_distrotion_m}]
    Recall that $B \in \mathcal{A}$ minimizes the social cost, and $\{A_j\}_{j \in [m-1]}$ denotes the set $\A \setminus \{B\}$.
    \begin{align}{\label{eqn:cop_distortion_m}}
    \textsc{dist}^{(g)}(\textsc{Copeland},n,m)
    &=
    \sup_{d \in \mathcal{M}(\mathcal{N} \cup \mathcal{A})}
    \Biggl(
        \sum_{j=1}^{m-1}\mathbb{P}[A_j \text{ wins}]
        \frac{\SC(A_j,d)}{\SC(B,d)}
        + \mathbb{P}[B \text{ wins}]
    \Biggr).
    \end{align}
    Consider a Copeland winner $W$. As noted by prior work \cite{anshelevich2015approximating}, $W$ must be in the uncovered set of the tournament graph, and one of the following two cases must be true.
    \begin{itemize} 
        \item The oriented edge from $W$ to $B$ is present.
        \item There exists a candidate $Y \in \mathcal{A}$ such that the oriented edges from $W$ to $Y$ and from $Y$ to $B$ are present.
    \end{itemize}
%

    For every $j \in [m-1]$, we now bound the probability of $A_j$ being the winner. 
    For every $j \in [m-1]$, we define Bernoulli random variables $\{Y_{i,j}\}_{i=1}^{n}$ denoting the event that voter $i$ ranks candidate $A_j$ over candidate $B$.
    From Equation \ref{eq:pairwise_probability}, we have that $Y_{i,j} \sim \text{Bern}\left(g\left(\frac{d(i,B)}{d(i,A_j)}\right)\right)$. For every distinct $j,k \in [m-1]$, we define Bernoulli random variables $\{Z_{i,j,k}\}_{i=1}^{n}$ denoting the event that voter $i$ ranks candidate $A_j$ over $A_k$. Then $Z_{i,j,k} \sim \text{Bern}\left(g\left(\frac{d(i,A_k)}{d(i,A_j)}\right)\right)$.

Observe that    
    \begin{align}{\label{eq:bounding_prob_winner_cop}}
        \mathbb{P}[A_j \text{ wins}]
        &\leq \mathbb{P}\Biggl(
            \left\{\sum_{i \in \N} Y_{i,j} \geq \frac{n}{2}\right\}
            \bigcup _{k \in [m-1]\setminus \{j\}}
            \left(
                \left\{\sum_{i \in \N} Z_{i,j,k} \geq \frac{n}{2}\right\}
                \cap
                \left\{\sum_{i \in \N} Y_{i,k} \geq \frac{n}{2}\right\}
            \right)
        \Biggr).
    \end{align}

     Let $n\alpha_j$ denote the expected value of the random variable $\sum_{i \in \N} Y_{i,j}$, i.e., the expected number of voters who rank candidate $A_j$ over $B$. 
\begin{equation}{\label{eqn:alpha_defn_cop}}
    \alpha_j :=
    \frac{1}{n}\sum_{i \in \N} \mathbbm{E}[Y_{i,j}]
    =
    \frac{1}{n}\sum_{i \in \N}
    g\left(\frac{d(i,B)}{d(i,A_j)}\right),
    \quad \forall j \in [m-1].
\end{equation}

Let $n\beta_{j,k}$ denote the expected value of the random variable $\sum_{i \in \N} Z_{i,j,k}$, i.e., the expected number of voters who rank candidate $A_j$ over $A_k$. 
\begin{equation}{\label{eqn:beta_defn_cop}}
    \beta_{j,k} :=
    \frac{1}{n}\sum_{i \in \N} \mathbbm{E}[Z_{i,j,k}]
    =
    \frac{1}{n}\sum_{i \in \N}
    g\left(\frac{d(i,A_k)}{d(i,A_j)}\right),
    \quad \forall j,k \in [m-1],\ j\neq k.
\end{equation}

We next apply Lemma~\ref{lemma:threshold_chernoff_bound} with \(q=2\) and \(r=2\). For the \(Y_{i,j}\) variables we take \(\rho=\alpha_j\), and for the \(Z_{i,j,k}\) variables we take \(\rho=\beta_{j,k}\). Since \(n\ge16\), the lemma's remaining conditions \(n/2>2\) and \(n^{1/2+\epsilon}/2\ge2\) both hold, so the lemma gives the following two bounds.

\begin{align}{\label{eq:chernoff_bounding_temp1}}
    \mathbb{P}\left(\sum_{i \in \N} Y_{i,j} \geq \frac{n}{2}\right)
    \leq
    \left(2\alpha_j\right)^2
    \exp\left(\frac{- n^{(\frac{1}{2}+\epsilon)} +8}{2(2n^{(\frac{1}{2}-\epsilon)}-1)}\right)
    \quad
    \text{if }\alpha_j \leq \frac{1}{2} - \frac{n^{(-1/2+\epsilon)}}{2}.
\end{align}
\begin{align}{\label{eq:chernoff_bounding_temp2}}
\mathbb{P}\left(\sum_{i \in \N} Z_{i,j,k} \geq \frac{n}{2}\right)
\leq
\left(2\beta_{j,k}\right)^2
\exp\left(\frac{- n^{(\frac{1}{2}+\epsilon)} +8}{2(2n^{(\frac{1}{2}-\epsilon)}-1)}\right)
\quad
\text{if }\beta_{j,k} \leq \frac{1}{2} - \frac{n^{(-1/2+\epsilon)}}{2}.
\end{align}

Consider two exhaustive cases on candidate $A_j$ and define an event $E_j$ for every $j \in [m-1]$. We compute the expected fraction of votes on pairwise comparisons. The event $E_j$ denotes the existence of an at-most two hop directed path from a candidate $A_j$ to candidate $B$ for \textsc{Copeland} such that the expected fraction of votes on all edges along that path exceed $\frac{1}{2}-\frac{n^{(-1/2+\epsilon)}}{2}$ . Recall that we only considered one hop path for the case of \textsc{Plurality} in the proof of Theorem \ref{theorem:thm_plurality_distortion_m}.


\begin{equation}{\label{eq:E_j_defn}}
    E_j:= \left(\alpha_j \geq \frac{1}{2}-\frac{n^{(-1/2+\epsilon)}}{2}\right) \bigcup_{k \in [m-1]\setminus \{j\}} \left( \left(\beta_{j,k} \geq \frac{1}{2}-\frac{n^{(-1/2+\epsilon)}}{2}\right) \bigcap \left(\alpha_k \geq \frac{1}{2}-\frac{n^{(-1/2+\epsilon)}}{2}\right)\right). 
\end{equation}

If $E_j$ holds true, we can directly upper bound the ratio of the social cost of candidate $A_j$ to the social cost of candidate $B$ using Observation \ref{lemma:bounding_social_cost_ratio}, which in turn provides a bound on the distortion. If $E_j$ does not hold, we apply the union bound and Chernoff's bound to upper bound the probability of $A_j$ being the winner. By multiplying this probability bound with the ratio of social costs obtained from Observation \ref{lemma:bounding_social_cost_ratio}, we derive a bound on the distortion.

Define $S := \{j \in [m-1]: E_j \text{ is not true}\}$. Furthermore, for each \(j\in[m-1]\), define
\[
    \mathcal{K}_1(j):=
    \{k \in [m-1]\setminus\{j\}: \alpha_k \geq \beta_{j,k}\},
    \quad
    \mathcal{K}_2(j):=
    \{k \in [m-1]\setminus\{j\}: \alpha_k < \beta_{j,k}\}.
\]
Thus, \(\mathcal{K}_2(j)\) is the complement of \(\mathcal{K}_1(j)\) in \([m-1]\setminus\{j\}\). 



From Equations~\eqref{eq:chernoff_bounding_temp1} and \eqref{eq:chernoff_bounding_temp2}, both of the following conditions \ref{label:itemone} and \ref{label:itemtwo} are 
satisfied for every $j \in S$.
\begin{enumerate}  
    \item $\mathbb{P}\left(\sum_{i \in \N} Y_{i,j} \geq \frac{n}{2}\right)  \leq \left(2\alpha_j\right)^2 \exp\left(\frac{- n^{(\frac{1}{2}+\epsilon)} +8}{2(2n^{(\frac{1}{2}-\epsilon)}-1)}\right)$ \label{label:itemone}

    \item For every $k \in [m-1] \setminus \{j\}$, \label{label:itemtwo}

    $\mathbb{P}\left(\sum_{i \in \N} Z_{i,j,k} \geq \frac{n}{2}\right) \leq \left(2\beta_{j,k}\right)^2 \exp\left(\frac{- n^{(\frac{1}{2}+\epsilon)} +8}{2(2n^{(\frac{1}{2}-\epsilon)}-1)}\right) \text{ if $k \in \mathcal{K}_1(j)$}$

and, 
    $\mathbb{P}\left(\sum_{i \in \N} Y_{i,k} \geq \frac{n}{2}\right) \leq \left(2\alpha_k\right)^2 \exp\left(\frac{- n^{(\frac{1}{2}+\epsilon)} +8}{2(2n^{(\frac{1}{2}-\epsilon)}-1)}\right) \text{ if $k \in \mathcal{K}_2(j)$}$.
%
    %
    %
\end{enumerate}
Every Copeland winner $W$ must either have an oriented edge to $B$ or there exists a $Y\in \mathcal{A}$ such that the oriented edges from $W$ to $Y$ and from $Y$ to $B$ are present. Each such oriented edge implies that the corresponding weak pairwise-majority event has occurred.
Using union bound for every $j \in S$, we have
\begin{align}{\label{eq:probability_bound}}
     \mathbb{P}[A_j \text{ wins}] 
     \leq & ~ \mathbbm{P}\left[\sum_{i \in \N} Y_{i,j} \geq \frac{n}{2}\right] \nonumber\\
     & + \sum_{k \in [m-1]\setminus \{j\}} \mathbbm{P}\left[\left(\sum_{i \in \N} Y_{i,k} \geq \frac{n}{2}\right) \cap \left(\sum_{i \in \N} Z_{i,j,k} \geq \frac{n}{2}\right)\right] \text{ if $j \in S$} \nonumber\\
     \leq & \left(2\alpha_j\right)^2 \exp\left(\frac{- n^{(\frac{1}{2}+\epsilon)} +8}{2(2n^{(\frac{1}{2}-\epsilon)}-1)}\right) \nonumber\\
     & + \sum_{k \in \mathcal{K}_2(j)} \left({2\alpha_k}\right)^2 \exp\left(\frac{- n^{(\frac{1}{2}+\epsilon)} +8}{2(2n^{(\frac{1}{2}-\epsilon)}-1)}\right) \nonumber \\
     & + \sum_{k \in \mathcal{K}_1(j)} \left(2\beta_{j,k}\right)^2 \exp\left(\frac{- n^{(\frac{1}{2}+\epsilon)} +8}{2(2n^{(\frac{1}{2}-\epsilon)}-1)}\right) \text{ if $j \in S$}. 
\end{align}

We next bound the probability-weighted contribution of \(A_j\). The one-hop path from \(A_j\) to \(B\) gives the direct social-cost bound \(\Opt(\Efrac{\alpha_j})^{-1}\). For a two-hop path through \(A_k\), the two edge constraints are imposed on the same electorate, so we use Observation~\ref{obs:copeland_three_candidate_ratio}: the social-cost ratio is at most \(\Opt(\mathcal T_{\beta_{j,k},\alpha_k})^{-1}\), not a product of two separately chosen one-hop ratios. Claim~\ref{claim:copeland_program_product_lower_bound} then gives a product-type relaxation of this shared-program bound.

If one of the thresholds below is zero, the corresponding Chernoff prefactor is zero and the term may be omitted. Write
\[
    H_n:=
    \exp\left(
    \frac{- n^{(\frac{1}{2}+\epsilon)} +8}
    {2(2n^{(\frac{1}{2}-\epsilon)}-1)}
    \right).
\]
Applying the union bound in probability-weighted form and using Claim~\ref{claim:copeland_program_product_lower_bound} gives

	\begin{align}{\label{eq:bounding_expr_part1}}
	    \mathbb{P}[A_j \text{ wins}] \frac{\SC(A_j,d)}{\SC(B,d)}
	    \leq & ~4H_n
	    \Bigg(
	        \alpha_j^2\Opt(\Efrac{\alpha_j})^{-1} \nonumber\\
	        &\quad + \sum_{k\in\mathcal K_1(j)}
	        \beta_{j,k}^2
	        \Opt(\Efrac{\beta_{j,k}})^{-1}
	        \Opt(\Efrac{\alpha_k})^{-1} \nonumber\\
	        &\quad + \sum_{k\in\mathcal K_2(j)}
	        \alpha_k^2
	        \Opt(\Efrac{\beta_{j,k}})^{-1}
	        \Opt(\Efrac{\alpha_k})^{-1}
	    \Bigg) \nonumber\\
	    \leq & ~4mH_n
	    \max\left\{\Gamma_g(\eta_n),\Gamma_g(\eta_n)^2\right\}
	    \text{ if $j \in S$}.
	\end{align}
	
	The last inequality uses monotonicity of \(\Opt(\Efrac{t})\) in \(t\): increasing \(t\) only tightens the pairwise-support constraint in \(\Efrac{t}\), so the minimum social-cost ratio cannot decrease. If \(k\in\mathcal K_1(j)\), then \(\beta_{j,k}\le\alpha_k\), so the corresponding two-hop summand is at most \((\beta_{j,k}\Opt(\Efrac{\beta_{j,k}})^{-1})^2\le\Gamma_g(\eta_n)^2\). If \(k\in\mathcal K_2(j)\), then \(\alpha_k<\beta_{j,k}\), and monotonicity gives \(\Opt(\Efrac{\beta_{j,k}})^{-1}\leq \Opt(\Efrac{\alpha_k})^{-1}\); hence that summand is at most \((\alpha_k\Opt(\Efrac{\alpha_k})^{-1})^2\leq\Gamma_g(\eta_n)^2\). The direct summand is at most \(\Gamma_g(\eta_n)\), and there are at most \(m\) summands.

Recall that for every $j \in [m-1] \setminus S$, $E_j$ is satisfied. Let us further denote 

\[
    \hat{E}_j :=
    \left\{\alpha_j \geq \frac{1}{2}-\frac{n^{(-1/2+\epsilon)}}{2}\right\},
    \quad
    \hat{D}_{j,k} :=
    \left\{\beta_{j,k} \geq \frac{1}{2}-\frac{n^{(-1/2+\epsilon)}}{2}\right\}.
\]








Observe that $E_j$ being satisfied implies either a) $\hat{E}_j$ is satisfied or b) $\exists k \in [m-1]\setminus \{j\}$ s.t $\hat{E}_k$ and $\hat{D}_{j,k}$ are satisfied. We consider both cases separately. 

Let \(\eta_n:=\frac{1-n^{-(1/2-\epsilon)}}{2}\). Suppose \(\hat{E}_j\) is satisfied for some \(j \in [m-1]\setminus S\). Then Observation~\ref{lemma:bounding_social_cost_ratio} gives
\begin{align}{\label{eq:bounding_expr_part2}}
    & \frac{\SC(A_j,d)}{\SC(B,d)}  
    \leq \Opt(\Efrac{\eta_n})^{-1}.
\end{align}
%
%
Now we consider case (b) where \(\hat{E}_k\) and \(\hat{D}_{j,k}\) are both satisfied for some \(k \in [m-1]\setminus \{j\}\). Observation~\ref{obs:copeland_three_candidate_ratio} gives
%
\begin{align}{\label{eq:bounding_expr_part3}}
    & \frac{\SC(A_j,d)}{\SC(B,d)}  
    \leq \Opt(\mathcal{T}_{\eta_n,\eta_n})^{-1}.
\end{align}
Now combining Equations \eqref{eq:bounding_expr_part1}, \eqref{eq:bounding_expr_part2}, and \eqref{eq:bounding_expr_part3}, we have for any metric space $d \in \mathcal{M}(\mathcal{N}\cup \mathcal{A})$,
%
%
%
%
%
    \begin{align}{\label{eqn:bounding_copeland_distortion_case_2}}
    &\textsc{dist}^{(g)}(\textsc{Copeland},n,m) \nonumber\\
    &\leq \Biggl( \sum_{j \in S}\left(\mathbb{P}[A_j \text{ wins}] \frac{\SC(A_j,d)}{\SC(B,d)}\right) + \mathbb{P}[B \text{ wins}] \nonumber\\
    &\qquad + \sum_{j \in [m-1] \setminus S}\left(\mathbb{P}[A_j \text{ wins}] \frac{\SC(A_j,d)}{\SC(B,d)} \right) \Biggr)\nonumber\\
    &\overset{(a)}{\leq}  4(m-1) m  \exp\left(\frac{- n^{(\frac{1}{2}+\epsilon)} +8}{2(2n^{(\frac{1}{2}-\epsilon)}-1)}\right)\max\left\{\Gamma_g(\eta_n),\Gamma_g(\eta_n)^2\right\} \nonumber\\
    &\qquad + \max\left(\max_{j \in [m-1]\setminus S} \frac{\SC(A_j,d)}{\SC(B,d)},1\right) \nonumber \\
     &\overset{(b)}{\leq}  4(m-1) m  \exp\Big(\frac{- n^{(\frac{1}{2}+\epsilon)} +8}{2(2n^{(\frac{1}{2}-\epsilon)}-1)}\Big) \max\left\{\Gamma_g(\eta_n),\Gamma_g(\eta_n)^2\right\} \nonumber\\
     &\qquad + \max \left\{\Opt(\Efrac{\eta_n})^{-1},\Opt(\mathcal{T}_{\eta_n,\eta_n})^{-1}\right\}.
    \end{align}

$(a)$ follows from Equation \eqref{eq:probability_bound} and the fact that $\sum_{j \in [m-1]\setminus S}\left(\mathbb{P}[A_j \text{ wins}] \frac{\SC(A_j,d)}{\SC(B,d)}\right) + \mathbb{P}[B \text{ wins}] \leq \max\left(\max\limits_{j \in [m-1]\setminus S} \frac{\SC(A_j,d)}{\SC(B,d)},1\right)$. 

$(b)$ follows from combining Equations  \eqref{eq:bounding_expr_part1}, \eqref{eq:bounding_expr_part2}, and \eqref{eq:bounding_expr_part3}.
\end{proof}

\subsection{Copeland Lower Bound}
\label{sec:copeland_lower_bound_proof}

\begin{claim}[Restatement of Claim~\ref{claim:copeland_lower_bound}]
	\claimcopelandlowerbound
\end{claim}

\begin{proof}
	Fix \(\zeta>0\) and \(\delta>0\). Choose a feasible solution \(\{(p_\ell,b_\ell,y_\ell,w_\ell)\}_{\ell=1}^{3}\) of \(\mathcal{T}_{1/2+\zeta,1/2+\zeta}\) whose objective value is at most \(\Opt(\mathcal{T}_{1/2+\zeta,1/2+\zeta})+\delta\). Realize the three jointly feasible voter types in a metric space with candidates \(B,Y,W\). For an \(n\)-voter electorate, place \(\lfloor p_\ell n\rfloor\) voters of type \(\ell\), and distribute the remaining \(O(1)\) voters arbitrarily.
	
		The expected fractions of voters ranking \(W\) above \(Y\) and \(Y\) above \(B\) are at least \(1/2+\zeta-o_n(1)\). Hence, for all sufficiently large \(n\), both expectations are at least \(1/2+\zeta/2\). Chernoff's bound implies that \(W\) beats \(Y\) and \(Y\) beats \(B\) with probability tending to one. We label the three candidates so that the fixed arbitrary tie-breaking rule for \textsc{Copeland} selects \(W\) whenever the resulting three-candidate tournament ties all three candidates. Thus \(W\) is a Copeland winner with probability tending to one.
	
	The social-cost ratio converges to
	\[
	\frac{\sum_{\ell=1}^{3}p_\ell w_\ell}
	{\sum_{\ell=1}^{3}p_\ell b_\ell}
	\geq
	\left(\Opt(\mathcal{T}_{1/2+\zeta,1/2+\zeta})+\delta\right)^{-1}.
	\]
	Letting \(n\to\infty\), then \(\delta\downarrow0\), and finally using Lemma~\ref{lemma:copeland_T_right_continuity} as \(\zeta\downarrow0\), gives the claimed lower bound.
\end{proof}



\section{Proof of Claims \ref{claim:borda_optim1} and \ref{claim:borda_optim2} from the proof of Lemma \ref{lemma:bounding_expected_score}}{\label{sec:proof_claims_in_lemma_boudning_score}}

We prove the two optimization claims used to lower bound Terms (II) and (III) in Equation~\eqref{eq:diff_borda_score}. Claim~\ref{claim:borda_optim1} controls Term (II): it lower bounds the optimization problem \(\mathcal{E}^{(1)}_{\mfg,z}\), which captures the contribution of candidates outside \(\C \cup \{W\}\). Claim~\ref{claim:borda_optim2} controls Term (III): it lower bounds \(\mathcal{E}^{(2)}_{\mfg}\), which captures the pairwise comparison between \(B^{(b)}\) and \(W\). Both proofs use the same broad strategy: replace the original objective by a one-dimensional surrogate, choose an appropriate Lagrange multiplier, and then show that the resulting Lagrangian is minimized only at controlled boundary points.

\noindent\textbf{Claim~\ref{claim:borda_optim1} (Restated):} \claimbordaoptimone


\begin{proof}
We proceed in four steps. First, we reduce the optimization problem to a scalar objective \(g^{(1)}\) and restrict attention to the active domain \(\mathcal B\). Second, we replace \(g^{(1)}\) by a signed lower-bound approximation. Third, we choose a Lagrange multiplier that rules out interior stationary points. Finally, we compare the remaining boundary and kink values.
Throughout the proof, \(z\) is treated as a parameter. All constants may depend on \(\theta\), but not on \(m,n\), or \(z\), once \(z\geq Z_\theta\).

\smallskip
\noindent\textbf{Reduction to a scalar objective.}
To solve the optimization problem $\mathcal{E}^{(1)}_{\mfg,z}$, we first define the function $g^{(1)}(b)$.

\begin{equation}{\label{eq:g_1_b}}
  g^{(1)}(b) = \min\limits_{t: |b-t|\leq 1\leq b+t}\left(\frac{1}{1+ \left(\frac{b}{t}\right)^{\theta}} - \frac{1}{1+ \left(\frac{|\frac{1}{z}-b|}{t}\right)^{\theta}}\right)  
\end{equation}

Observe that solving $\mathcal{E}^{(1)}_{\mfg,z}$ is equivalent to minimising $\sum_{i \in \N} g^{(1)}(b_i)$ subject to $\sum_{i \in \N} b_i \leq \frac{n}{z\mfg(m)}.$





Now we minimize the above function over $\rC = \{t: |b-t| \leq 1 \leq b+t\}$. By differentiating the function in Equation \eqref{eq:g_1_b} with respect to $t$, we observe that it has a unique minimum at $t = \sqrt{b|b-\frac{1}{z}|}$ if $b>\frac{1}{2z}$ and a unique maximum at $t = \sqrt{b|b-\frac{1}{z}|}$ otherwise. Thus, we can divide it into three cases using the constraint $\rC$, which requires $t$ to be bounded in $[|1-b|,1+b]$.

\begin{equation}{\label{eq:cases_g_1}}
g^{(1)}(b) =
\begin{cases}
\left(\frac{1}{1+ \left(\frac{b}{b+1} \right)^{\theta}} - \frac{1}{1+ \left(\frac{|\frac{1}{z}-b|}{b+1} \right)^{\theta}} \right) & \text{if } b \leq \frac{1}{2z} \\
\left(\frac{1}{1+ \left(\frac{b}{|b-1|} \right)^{\theta}} - \frac{1}{1+ \left(\frac{|\frac{1}{z}-b|}{|b-1|} \right)^{\theta}} \right) & \text{ else if } b + \sqrt{b\left|\frac{1}{z}-b \right|} \leq 1 \\
\left(\frac{1}{1+ \left(\frac{b}{|b-\frac{1}{z}|} \right)^{\theta/2}} - \frac{1}{1+ \left(\frac{|\frac{1}{z}-b|}{|b|} \right)^{\theta/2}}\right) & \text{otherwise}
\end{cases}
\end{equation}



\smallskip
\noindent\textbf{The three forms of \(g^{(1)}\).}
Under both the first and second cases, the terms $1-b$ and $1+b$ exceed $\sqrt{b\left|\frac{1}{z}-b \right|}$ (the local extremum). In the first case, since the local extremum is a maximum, the function in Equation \eqref{eq:g_1_b} monotonically decreases with $t$ over the interval $[1-b,1+b]$.  Conversely, in the second case, since the local extremum is a minimum, the function monotonically increases with $t$ over the interval $[1-b,1+b]$. In the third case, the minimum point $t=\sqrt{b|b-\frac{1}{z}|}$ lies within the interval $[|1-b|,1+b]$. These behaviors explain the corresponding terms in Equation \eqref{eq:cases_g_1}.

\smallskip
\noindent\textbf{Restriction to the active domain.}
Due to the negative and monotonically increasing nature of $g^{(1)}(b)$ when $b + \sqrt{b \left|\frac{1}{z}-b \right|} > 1$, 
we can simplify our analysis by focusing only on a restricted domain. Indeed, if a feasible solution has such a value \(b_i\), replacing it by the boundary value \(\bar b\) satisfying \(\bar b+\sqrt{\bar b|\frac{1}{z}-\bar b|}=1\) weakly decreases the objective and also weakly decreases \(\sum_i b_i\), so the budget constraint remains feasible. Thus, up to taking the closure and then a limiting argument, in the optimization problem $\mathcal{E}^{(1)}_{\mfg,z}$ we can restrict $b_i$ to the set $\mathcal{B}$, which is defined as follows:
$$\mathcal{B} = \left\{b: b + \sqrt{b\left|\frac{1}{z}-b \right|} < 1,\; b > 0 \right\}.$$ Furthermore, when \(z\geq5\), if \(b\in\mathcal B\), then \(b\leq 2/3\). We will choose \(Z_\theta\geq5\), so this bound holds throughout the rest of the proof.

\smallskip
\noindent\textbf{Signed approximation.}
We next replace \(g^{(1)}\) by an order-equivalent approximation on \(\mathcal{B}\). The three regions below all lead to the same approximating function \(g^{(1)}_{\text{apr}}\), up to a multiplicative error \(1+\varepsilon(z)\), where \(\varepsilon(z)\to0\) as \(z\to\infty\), uniformly on the relevant region. Because \(g^{(1)}_{\text{apr}}\) changes sign at \(b=1/z\), we use this approximation only through the signed one-sided lower bound stated after the three cases.



\begin{itemize}
    \item When $b \geq \frac{1}{z}$ and $b \in \mathcal{B}$, we have

\begin{align}
   \left(\frac{1}{1+ \left(\frac{b}{|b-1|} \right)^{\theta}} - \frac{1}{1+ \left(\frac{|\frac{1}{z}-b|}{|b-1|} \right)^{\theta}} \right) = & \frac{(|b-\frac{1}{z}|^{\theta}-b^{\theta})|b-1|^{\theta}}{ \left(|b-1|^{\theta}+ {b}^{\theta}\right) \left({|b-1|}^{\theta}+ \left({|\frac{1}{z}-b|}^{\theta} \right)\right)}\nonumber\\
    = & \underbrace{\frac{(|b-\frac{1}{z}|^{\theta}-b^{\theta})(1-b)^{\theta}}{ \left((1-b)^{\theta}+ {b}^{\theta}\right)^2}}_{g^{(1)}_{\text{apr}}(b)}
    \times \eta_+(b)
\end{align}

where
\[
\eta_+(b)=
\frac{(1-b)^\theta+b^\theta}
{(1-b)^\theta+\left|b-\frac{1}{z}\right|^\theta}.
\]
Observe that
\[
    \sup_{b\in\mathcal B: b\geq 1/z}|\eta_+(b)-1|\to0
    \qquad\text{as }z\to\infty.
\]
Indeed, \(b\leq 2/3\) on \(\mathcal B\), so the denominator of \(\eta_+(b)\) is bounded below by a positive constant, while
\[
    \left|b^\theta-\left|b-\frac{1}{z}\right|^\theta\right|
    =O(z^{-1})
\]
uniformly on this region. At the endpoint \(b=1/z\), this difference is only \(z^{-\theta}\). Thus this case also differs from \(g^{(1)}_{\text{apr}}(b)\) by a multiplicative \(1+o_z(1)\) factor.
\item When $\frac{1}{2z} <  b < \frac{1}{z}$ and $b \in \mathcal{B}$, we have

\begin{align}
    \left(\frac{1}{1+ \left(\frac{b}{b-1} \right)^{\theta}} - \frac{1}{1+ \left(\frac{|\frac{1}{z}-b|}{|b-1|} \right)^{\theta}} \right) = & \frac{(|b-\frac{1}{z}|^{\theta}-b^{\theta})|b-1|^{\theta}}{ \left(|b-1|^{\theta}+ {b}^{\theta}\right) \left({|b-1|}^{\theta}+ \left({|\frac{1}{z}-b|}^{\theta} \right)\right)}\nonumber\\
    = & \underbrace{\left(\left|\frac{b}{1-b}-\frac{1}{z (1-b)}\right|^{\theta}-\left(\frac{b}{1-b}\right)^{\theta}\right)}_{g^{(1)}_{\text{apr}}(b)} \times \eta(b)
\end{align}

Observe that
\[
    \sup_{0\leq b\leq 1/z}|\eta(b)-1|\to0
    \qquad\text{as }z\to\infty.
\]

\item When $b \leq \frac{1}{2z}$ and $b \in \mathcal{B}$, we have

\begin{align}
    \left(\frac{1}{1+ \left(\frac{b}{b+1} \right)^{\theta}} - \frac{1}{1+ \left(\frac{|\frac{1}{z}-b|}{b+1} \right)^{\theta}} \right) = & \frac{(|b-\frac{1}{z}|^{\theta}-b^{\theta})(b+1)^{\theta}}{ \left((b+1)^{\theta}+ {b}^{\theta}\right) \left({(b+1)}^{\theta}+ \left({|\frac{1}{z}-b|}^{\theta} \right)\right)}\nonumber\\
    = & \underbrace{\left(\left|\frac{b}{1+b}-\frac{1}{z (1+b)}\right|^{\theta}-\left(\frac{b}{1+b}\right)^{\theta}\right)} \times \zeta(b) \nonumber\\
    = & \underbrace{\left(\left|\frac{b}{1-b}-\frac{1}{z (1-b)}\right|^{\theta}-\left(\frac{b}{1-b}\right)^{\theta}\right)}_{g^{(1)}_{\text{apr}}(b)} \times \zeta(b) \\ & \hspace{ 6 em}\times \left(\frac{1-b}{1+b}\right)^{\theta}
\end{align}

Observe that
\[
    \sup_{b\in\mathcal B: b\leq 1/(2z)}
    \left|
    \zeta(b)\left(\frac{1-b}{1+b}\right)^\theta-1
    \right|
    \to0
    \qquad\text{as }z\to\infty.
\]

\end{itemize}

Combining the three cases, there is a function \(\varepsilon(z)\to0\) as \(z\to\infty\), depending only on \(\theta\), such that, uniformly over \(b\in\mathcal B\),
\[
    g^{(1)}(b)
    \geq
    (1-\varepsilon(z))\bigl(g^{(1)}_{\text{apr}}(b)\bigr)_+
    -(1+\varepsilon(z))\bigl(g^{(1)}_{\text{apr}}(b)\bigr)_-,
\]
where \(x_+=\max\{x,0\}\) and \(x_-=\max\{-x,0\}\). This is the signed version of the \(1+o_z(1)\) approximation: the positive part is multiplied by the smaller factor, while the negative part is multiplied by the larger factor, which is the direction required for a lower bound.

\smallskip
\noindent\textbf{Lagrangian surrogate.}
Let
\[
    \widetilde g^{(1)}(b)
    :=
    (1-\varepsilon(z))\bigl(g^{(1)}_{\text{apr}}(b)\bigr)_+
    -(1+\varepsilon(z))\bigl(g^{(1)}_{\text{apr}}(b)\bigr)_- .
\]
The preceding display gives \(g^{(1)}(b)\geq \widetilde g^{(1)}(b)\) uniformly on \(\mathcal B\). Hence it is sufficient to minimize \(\sum_{i=1}^{n}\widetilde g^{(1)}(b_i)\) subject to \(\sum_{i=1}^{n}b_i \leq \frac{n}{z\mfg(m)}\). The function \(\widetilde g^{(1)}\) differs from \(g^{(1)}_{\text{apr}}\) only by the constants \(1-\varepsilon(z)\) and \(1+\varepsilon(z)\) on the two sides of \(b=1/z\), so all derivative and boundary estimates below are unchanged up to constants depending only on \(\theta\), after increasing \(Z_\theta\) so that \(\varepsilon(z)\) is bounded away from \(1\) for all \(z\geq Z_\theta\). We therefore carry out the calculations with \(g^{(1)}_{\text{apr}}\), absorbing these factors into the constants \(C_\theta\) and \(E_\theta\).


Define
\[
    \widetilde F^{(1)}(\vect b,\lambda)
    :=
    \sum_{i\in\N}
    \left(\widetilde g^{(1)}(b_i)+\lambda b_i
    -\lambda\frac{1}{z\mfg(m)}\right).
\]
From Lagrangian duality applied to this signed surrogate, we have the exact lower bound

\begin{align}
    \Opt(\mathcal{E}^{(1)}_{\mfg,z})
    &\geq \min_{\vect{b} \in \mathcal{B}^n } \max_{\lambda \geq 0} \widetilde F^{(1)}(\vect{b},\lambda) \nonumber\\
    &\geq \max_{\lambda \geq 0} \min_{\vect{b} \in \mathcal{B}^n } \widetilde F^{(1)}(\vect{b},\lambda).
\end{align} 
For readability, the remaining derivative and boundary calculations are written for
\[
    F^{(1)}_{\text{apr}}(\vect{b},\lambda)
    =
    \sum_{i\in\N}
    \left(g^{(1)}_{\text{apr}}(b_i)+\lambda b_i
    -\lambda\frac{1}{z\mfg(m)}\right).
\]
The same calculations apply to \(\widetilde F^{(1)}\) after changing only the constants \(C_\theta\) and \(E_\theta\), because \(\widetilde g^{(1)}\) rescales the positive and negative pieces of \(g^{(1)}_{\text{apr}}\) by factors bounded above and below by positive constants for all \(z\geq Z_\theta\). All constants in the estimates below depend only on \(\theta\), and the estimates hold uniformly for all \(z\geq Z_\theta\).

\smallskip
\noindent\textbf{Choice of multiplier.}
We choose the multiplier \(\lambda=\lambda_1^\ast\) to dominate the largest possible negative slope of \(g^{(1)}_{\text{apr}}\) on \(\mathcal{B}\). Equivalently, \(\lambda_1^\ast\) has the order of the slope obtained from the tangent-line calculation for the limiting curve \(y=g^{(1)}_{\text{apr}}(x)\). With this choice, \(g^{(1)}_{\text{apr}}(b)+\lambda b\) has no interior stationary point for every \(z\geq Z_\theta\), after increasing \(Z_\theta\) if necessary, so its minimum is attained on the boundary of \(\mathcal{B}\).

\begin{equation}
    \lambda^{\ast}_1 =
\begin{cases}
\frac{C_{\theta}}{z}, & \text{if } \theta \geq 2, \\
\frac{C_{\theta}}{z^{\theta - 1}}, & \text{if } \theta < 2.
\end{cases} \text{ for $C_{\theta}$ depending only on \(\theta\)}  
\end{equation}


Observe that since each $b_i$ is independent of each other, it is sufficient to minimize $g^{(1)}_{\text{apr}}(b_i)+ \lambda b_i$ for each $b_i$ separately. Since $g^{(1)}_{\text{apr}}(b)$ is a piecewise differentiable function, it is sufficient to check the stationary points of $g^{(1)}_{\text{apr}}(b)+ \lambda b$ and thus, we consider two cases namely, $b \leq \frac{1}{z}$ and $b > \frac{1}{z}$.

\smallskip
\noindent\textbf{Stationary-point check.}

\begin{itemize} [leftmargin = 5pt]
    \item Case 1: $0 \leq b \leq \frac{1}{z}$.

    \begin{align}
        & \frac{d}{db} g^{(1)}_{\text{apr}}(b) \\ = & \frac{d}{db} \Biggl(\left(\frac{1}{z(1-b)}-\frac{b}{1-b}\right)^{\theta} - \left(\frac{b}{1-b} \right)^{\theta}\Biggr)\\
        = & \Biggl(\theta \left(\frac{1}{z}-1\right) \left(\frac{1+c}{z}-c\right)^{\theta-1} - \theta c^{\theta-1}\Biggr){(1-b)^{-2}} \\ & \hspace{17 em}\text{ where $c = \frac{b}{(1-b)}$}
    \end{align}

    As \(z\to\infty\), we have
    \[
    \lim_{z \to \infty}
    \max_{0<b<1/z}
    \left(-\frac{d}{db} g^{(1)}_{\text{apr}}(b)\right) z^{\theta-1}
    =
    \begin{cases}
    \theta 2^{2-\theta}, & \theta<2,\\
    \theta, & \theta\geq 2.
    \end{cases}
    \]
    This follows by replacing \((1+c)\) and \((1/z-1)\) by their limiting values \(1\) and \(-1\), respectively, and maximizing
    \(\left(\frac{1}{z}-c\right)^{\theta-1}+c^{\theta-1}\)
    over \(c\in[0,1/z]\). The maximum is attained in the interior when \(\theta<2\), and at the endpoints when \(\theta\geq2\). Thus, choosing \(C_\theta\) sufficiently large ensures that there is no solution to
    \(\frac{d}{db}g^{(1)}_{\text{apr}}(b)=-\lambda_1^\ast\)
    in this region for every \(z\geq Z_\theta\), after increasing \(Z_\theta\) if necessary.




    \item Case 2: $b > \frac{1}{z}$ and $b \in \mathcal{B}$.

    Now observe that $g^{(1)}_{\text{apr}}(b) = \left(\frac{\left(\frac{b}{1-b}-\frac{1}{z(1-b)}\right)^{\theta} - \left(\frac{b}{1-b} \right)^{\theta}}{\left(1+ \left(\frac{b}{1-b}\right)^{\theta}\right)^2}\right) = - \frac{\theta (\frac{b}{1-b})^{\theta-1} \frac{1}{z(1-b)}}{\left(1+ \left(\frac{b}{1-b}\right)^{\theta}\right)^2} + \epsilon_z(b)$, where $\sup\limits_{b \in \mathcal{B}} |\epsilon_z(b)|= o_z(\frac{1}{z})$. This follows from Taylor expansion and bounding $\frac{b}{1-b}$ by 2 since $b \leq 2/3$.

    We now record the derivative scale of the remainder. This is the only place where differentiating the approximation error is needed. Let \(c=b/(1-b)\), \(\Delta=(1+c)/z\), \(q=c-\Delta\), and
    \[
        R(c)=q^\theta-c^\theta+\theta c^{\theta-1}\Delta .
    \]
    Then \(\epsilon_z(b)=R(c)/(1+c^\theta)^2\), and since \(dc/db=(1+c)^2\),
    \[
        \epsilon_z'(b)
        =(1+c)^2\left(
        \frac{R'(c)}{(1+c^\theta)^2}
        -\frac{2\theta c^{\theta-1}R(c)}{(1+c^\theta)^3}
        \right),
    \]
    where
    \[
        R'(c)=\theta\left(1-\frac{1}{z}\right)
        \left(q^{\theta-1}-c^{\theta-1}\right)
        +\frac{\theta(\theta-1)c^{\theta-2}(1+c)}{z}.
    \]
    The two regimes in the bound below arise from the factor \(c^{\theta-2}\). When \(\theta\geq 2\), the map \(x\mapsto x^{\theta-1}\) has bounded derivative on \([0,2]\), so the mean value theorem gives \(|q^{\theta-1}-c^{\theta-1}|=O(z^{-1})\), and the second term in \(R'(c)\) is also \(O(z^{-1})\). When \(1<\theta<2\), the map \(x\mapsto x^{\theta-1}\) is \((\theta-1)\)-Hölder on \([0,2]\), so
    \[
        |q^{\theta-1}-c^{\theta-1}|
        \leq |\Delta|^{\theta-1}
        =O(z^{1-\theta}).
    \]
    For \(1<\theta<2\), the second term in \(R'(c)\) has the same
    worst-case order: since \(c^{\theta-2}\) is decreasing in \(c\), it is
    maximized near the lower endpoint \(c\asymp z^{-1}\), giving
    \(c^{\theta-2}/z=O(z^{1-\theta})\).
    Since \(c\in[1/(z-1),2]\) in this case, and \(\theta\) is fixed throughout the Borda analysis, we obtain
    \[
        \sup_{b\in\mathcal B: b>1/z}|\epsilon_z'(b)|
        =
        \begin{cases}
        O(z^{1-\theta}), & 1<\theta<2,\\
        O(z^{-1}), & \theta\geq 2.
        \end{cases}
    \]

    Now we have, \begin{align}
        \frac{d}{db} g^{(1)}_{\text{apr}}(b) & = \Biggl(-\frac{\theta((\theta-1) c^{\theta-2}+ \theta c^{\theta-1})}{z(1+c^{\theta})^2} + \frac{2\theta^2 c^{2\theta-2}(1+c)}{z(1+c^{\theta})^3} + \epsilon_z'(b)\Biggr) \frac{1}{(1-b)^2} \nonumber\\
        & \hspace{17 em} \text{ where $c= \frac{b}{1-b}$} 
    \end{align}

    \begin{itemize}
        \item When $\theta\geq 2$, we bound the most negative derivative over \(c\in [0,2]\), since $b \leq 2/3$. The quantity
        \[
        -\frac{\theta((\theta-1)c^{\theta-2}+ \theta c^{\theta-1})}{(1+c^{\theta})^2}
        + \frac{2\theta^2 c^{2\theta-2}(1+c)}{(1+c^{\theta})^3}
        \]
        is bounded above on this interval by a constant depending only on \(\theta\). Hence
        \(\max_{1/z<b<2/3} -\frac{d}{db}g^{(1)}_{\text{apr}}(b)=O(1/z)\), and choosing \(C_\theta\) sufficiently large ensures that there is no solution to
        \(\frac{d}{db} g^{(1)}_{\text{apr}}(b)=-C_\theta/z\).

        \item When $\theta<2$, the term \(c^{\theta-2}\) dominates near the lower boundary \(c\asymp 1/z\), while the bound on \(\epsilon_z'(b)\) above is of the same order and can be absorbed into the constant. Hence
        \[
        \max_{1/z<b<2/3}
        \left(-\frac{d}{db}g^{(1)}_{\text{apr}}(b)\right)z^{\theta-1}
        =
        O(1).
        \]
        Thus, choosing \(C_\theta\) sufficiently large ensures that there is no solution to
        \(\frac{d}{db} g^{(1)}_{\text{apr}}(b)=-C_\theta/z^{\theta-1}\)
        for every \(z\geq Z_\theta\), after increasing \(Z_\theta\) if necessary.

    \end{itemize}




\end{itemize} 

\smallskip
\noindent\textbf{Boundary and kink check.}
For every \(z\geq Z_\theta\), after increasing \(Z_\theta\) if necessary, there are no stationary points in either differentiable region, so the minimum of $g^{(1)}_{\text{apr}}(b) + \lambda b$ must occur at the boundaries or at the kink \(b=1/z\). The kink cannot be the minimizer: since \(g^{(1)}_{\text{apr}}(1/z)=0\), its value is \(\lambda/z\), which is \(\Theta(z^{-2})\) when \(\theta\geq2\) and \(\Theta(z^{-\theta})\) when \(1<\theta<2\); choosing \(C_\theta\) sufficiently large makes this value no smaller than the value at \(b=0\). The upper boundary \( b + \sqrt{b\left|\frac{1}{z}-b \right|} = 1\) has \(b=\Theta(1)\). There, \(g^{(1)}_{\text{apr}}(b)=-\Theta(z^{-1})\), whereas \(\lambda b=\Theta(z^{-1})\) when \(\theta\geq2\) and \(\lambda b=\Theta(z^{1-\theta})\) when \(1<\theta<2\). By taking \(C_\theta\) sufficiently large, this boundary value is also no smaller than the value at \(b=0\). We again consider the two cases on $\theta$.

\smallskip
\noindent\textbf{Final lower bound.}
\begin{itemize} [leftmargin = 5pt]
    \item For \( \theta \geq 2 \) and \(z\geq Z_\theta\), thus
\[
\min\limits_{b \in \mathcal{B}} g^{(1)}_{\text{apr}}(b) + \lambda b = g^{(1)}_{\text{apr}}(0)  = \Theta(z^{-\theta})
\]

Thus, the surrogate Lagrangian bound gives constants \(c_1,c_2>0\), depending only on \(\theta\), such that
\begin{align}
\Opt(\mathcal{E}^{(1)}_{\mfg, z})
\geq
n \left(c_1 z^{-\theta}-\frac{c_2}{\mfg(m)}z^{-2}\right).
\label{eq:borda_opt1_explicit_constants}
\end{align}
When \(\theta=2\), the right-hand side is
\(n z^{-2}(c_1-c_2/\mfg(m))\), which is nonnegative for sufficiently large \(m\). When \(\theta>2\), the minimum over \(z>0\) of the right-hand side in \eqref{eq:borda_opt1_explicit_constants} is \(-O_m(\mfg(m)^{-\theta/(\theta-2)})\). Since
\[
\mfg(m)=\omega_m(h_{\theta}(m))
=\omega_m(m^{1-2/\theta}),
\]
we have \(\mfg(m)^{-\theta/(\theta-2)}=o_m(m^{-1})\). Therefore, uniformly over \(z\),
\[
\Opt(\mathcal{E}^{(1)}_{\mfg, z}) \geq -n o_m(m^{-1}).
\]

    \item For \( \theta < 2 \) and \(z\geq Z_\theta\), thus
\[
\min\limits_{b \in \mathcal{B}} g^{(1)}_{\text{apr}}(b) + \lambda b = g^{(1)}_{\text{apr}}(0) = \Theta(z^{-\theta}),
\]
and therefore, by the same surrogate Lagrangian bound, there are constants \(c_3,c_4>0\), depending only on \(\theta\), such that
\begin{align*}
   \Opt(\mathcal{E}^{(1)}_{\mfg, z})
   & \geq n \left(c_3-\frac{c_4}{\mfg(m)}\right)z^{-\theta}
   = n\Omega_m(z^{-\theta})
   \geq n\Omega_m(z^{-2}),
 \end{align*}
where the last inequality uses \(z\geq Z_\theta\geq1\) and \(\theta<2\).
This completes the proof.
\end{itemize}

\end{proof}

\noindent\textbf{Claim~\ref{claim:borda_optim2} (Restated):} \claimbordaoptimtwo

\begin{proof}
We again use weak duality. The proof first reduces the domain to the closure of \((0,1)\), then chooses the tangent multiplier \(\lambda^\ast\), verifies that the resulting scalar objective is minimized at \(a=0\) and \(a=a^\ast\), and finally uses \(\mfg(m)=\omega(h_\theta(m))\).


\smallskip
\noindent\textbf{Weak-duality setup.}
We first define
\[
    F^{(2)}(\vect{a},\lambda)
    :=
    \sum_{i=1}^{n}
    \left(
        \frac{2}{1+\left(\frac{a_i}{|1-a_i|}\right)^{\theta}}
        -1+\lambda a_i-\lambda\frac{1}{\mfg(m)}
    \right).
\]
By weak duality,
\begin{equation}{\label{eq:lagrangian_eq1}}
    \Opt (\mathcal{E}^{(2)}_{\mfg})
    \geq
    \min_{\vect{a} \in \mathbb{R}^n_{+}} \max_{\lambda \geq 0} F^{(2)} (\vect{a},\lambda)
    \geq
    \max_{\lambda \geq 0} \min_{\vect{a} \in \mathbb{R}^n_{+}} F^{(2)} (\vect{a},\lambda).
\end{equation}

Observe that it is sufficient to restrict $a_i$ to the closure of the interval \((0,1)\). Indeed, for every ${a}>1,$ there exists \(\hat{a}\in(1/2,1)\) such that \(\frac{a}{|1-a|}=\frac{\hat a}{|\hat a-1|}\), and the endpoint \(a=0\) is included by continuity.

\smallskip
\noindent\textbf{Tangent multiplier.}
We lower bound this quantity by choosing a value of $\lambda$. Before setting $\lambda$, define $a^{\ast}$ by
\begin{equation}
    \frac{d}{da}\left(\frac{2}{1+\left(\frac{a}{1-a}\right)^{\theta}}-1\right)\bigg|_{a=a^{\ast}} = -\frac{1 - \left(\frac{2}{1+\left(\frac{a^{\ast}}{1-a^{\ast}}\right)^{\theta}} -1\right)}{a^{\ast}}
\end{equation}

We construct $a=a^{\ast}$ such that the tangent to the function $g^{(2)}(a) = \frac{2}{1+\left(\frac{a}{1-a}\right)^{\theta}}-1$ at $a=a^{\ast}$ passes through $(0,g^{(2)}(0))$. Further, we construct $\lambda^{\ast}$ as the modulus of the slope of the tangent.

\smallskip
\noindent\textbf{Shape and endpoint check.}
We now justify that this tangent construction gives the required minimum. Let \(x=a/(1-a)\), so \(a=x/(1+x)\), and define
\[
    S(x):=-\frac{d}{da}g^{(2)}(a)
    =
    \frac{2\theta x^{\theta-1}(1+x)^2}{(1+x^\theta)^2}.
\]
The tangent equation above is equivalent to
\[
    x_\ast^\theta=\theta(1+x_\ast)-1,
    \qquad x_\ast=\frac{a^\ast}{1-a^\ast}.
\]
The equation has a unique solution \(x_\ast>1\), since \(x^\theta-\theta x-(\theta-1)\) is strictly increasing for \(x>1\), is negative at \(x=1\), and tends to infinity. Moreover,
\[
    \frac{S'(x)}{S(x)}
    =
    \frac{(\theta-1)+(\theta+1)x-(\theta+1)x^\theta-(\theta-1)x^{\theta+1}}
    {x(1+x)(1+x^\theta)}.
\]
The numerator is positive for \(0<x<1\) and negative for \(x>1\); hence \(S\) increases on \((0,1)\) and decreases on \((1,\infty)\). Since \(S(0^+)=S(\infty)=0\), \(S(1)=2\theta\), and \(\lambda^\ast=S(x_\ast)<S(1)\), the equation \(S(x)=\lambda^\ast\) has exactly two solutions, one on each side of \(1\). Therefore \(g^{(2)}(a)+\lambda^\ast a\) has one interior local maximum and one interior local minimum on \((0,1)\), with the latter attained at \(a=a^\ast\). Moreover, by construction,
\[
    g^{(2)}(0)=1
    \quad\text{and}\quad
    g^{(2)}(a^\ast)+\lambda^\ast a^\ast=1.
\]
The remaining endpoint satisfies \(\liminf_{a\to1^-}(g^{(2)}(a)+\lambda^\ast a)=\lambda^\ast-1\). Using the tangent equation, \(\lambda^\ast=2x_\ast^{\theta-1}/\theta\); since \(x^\theta-\theta x-(\theta-1)<0\) at \(x=\theta^{1/(\theta-1)}\), we have \(x_\ast>\theta^{1/(\theta-1)}\), and therefore \(\lambda^\ast\geq2\). Hence the minimum of \(g^{(2)}(a)+\lambda^\ast a\) over the closure of \((0,1)\) is \(1\), attained at \(a=0\) and \(a=a^\ast\).

\smallskip
\noindent\textbf{Conclusion.}
Therefore, Equation \eqref{eq:lagrangian_eq1} gives
\[
    \Opt(\mathcal{E}^{(2)}_{\mfg})
    \geq n\left(1-\frac{\lambda^\ast}{\mfg(m)}\right).
\]
Since \(\lambda^\ast\) depends only on \(\theta\), while \(\mfg(m)=\omega(h_\theta(m))\) and \(h_\theta(m)\geq1\), we have \(\mfg(m)\to\infty\). Hence \(\Opt(\mathcal{E}^{(2)}_{\mfg})\geq n\Omega_m(1)\).
\end{proof}

\section{Proof of Lemma \ref{lemma:borda_upper_bound}}{\label{sec:borda_upper_bound_proof}}

\noindent\textbf{Lemma~\ref{lemma:borda_upper_bound} (Restated):} \lemmabordaub



\begin{proof}[Proof of Lemma \ref{lemma:borda_upper_bound}]






    \begin{align} \label{eq:distortion_split_borda}
    \textsc{dist}^{\theta}_{PL}\nonumber(\textsc{Borda},n,m) = & \sup_{\substack{\N: |\N| =n \\ \A: |\A| =m}} \sup_{d \in \mathcal{M}(\N \cup \mathcal{A})} \frac{\mathbb{E}_{\sigma_{\N} \sim \textsc{PL}_{\theta}}[\SC(f(\sigma_{\N}),d)]}{\min\limits_{A \in \A}\SC(A,d)}\\ 
    = & \sup_{\substack{\N: |\N| =n \\ \A: |\A| =m}} \sup_{d \in \mathcal{M}(\N \cup \mathcal{A})} \sum\limits_{j=1}^{m}\mathbb{P}(\textsc{ $A_j$ is the Borda winner} ) \frac{\SC(A_j,d)}{\min\limits_{A \in \mathcal{A}}\SC(A,d)} \nonumber\\
     \end{align}

To prove Lemma \ref{lemma:borda_upper_bound}, assume the contrary, that
\(D_m^\theta = \omega(\distborda)\),
and derive a contradiction. Given any metric space \( d(\cdot) \) over the set of candidates and voters, let \( B \) be a candidate that minimizes the social cost. We bound \( \mathbb{P}(\textsc{$A_j$ is the winner}) \frac{\SC(A_j, d)}{\SC(B, d)} \) by considering the two cases below. Let \(  \alpha_j n\) denote the expected number of voters who rank candidate \( A_j \) above candidate \( B \), following Equation \eqref{eqn:alpha_defn} in the proof of Theorem \ref{theorem:thm_plurality_distortion_m} (Appendix \ref{sec:thm_plurality_distortion_m_proof}). Additionally, let \( \epsilon \in (0,1/2) \).

\begin{itemize}
    \item \textbf{Case 1:} \( \alpha_j \leq \frac{1}{m} - \frac{n^{(-1/2 + \epsilon)}}{m} \).
    
    The probability that \( A_j \) is the Borda winner is at most the probability that \( A_j \) is ranked above \( B \) by at least \( \frac{n}{m} \) voters. Indeed, for any voter who ranks \(B\) above \(A_j\), candidate \(B\)'s per-voter Borda score exceeds \(A_j\)'s by at least \(1\); for any voter who ranks \(A_j\) above \(B\), candidate \(A_j\)'s per-voter Borda score exceeds \(B\)'s by at most \(m-1\). Thus, if fewer than \(n/m\) voters rank \(A_j\) above \(B\), the total Borda score of \(A_j\) is strictly below that of \(B\). Using Equation \eqref{eq:chernoff_bound_winning_prob} (Appendix \ref{sec:thm_plurality_distortion_m_proof}), this probability is at most
    \[
    m \alpha_j \exp\left( \frac{- n^{(1/2 + \epsilon)} + 2m}{(2 n^{(1/2 - \epsilon)} - 1)m} \right).
    \]
    Applying Observation \ref{lemma:bounding_social_cost_ratio}, Lemma~\ref{lemma:optimizer_mu_alpha_combined} we can further bound \( \frac{\SC(A_j, d)}{\SC(B, d)} \) by $\Delta_g(\alpha_j)$.
    
    Thus, we have:
    \[
        \mathbb{P}(\textsc{$A_j$ is the Borda winner}) \frac{\SC(A_j, d)}{\SC(B, d)} \leq m\alpha_j\Delta_g(\alpha_j) \exp\left( \frac{- n^{(1/2 + \epsilon)} + 2m}{(2 n^{(1/2 - \epsilon)} - 1)m} \right).
    \]

    As \( n \to \infty \), this upper bound approaches zero as $\alpha_j\Delta_g(\alpha_j)$ is bounded since $\Gamma_g(\eta)$ is finite (Equation \eqref{eq:Gamma_g_eta_defn}).

    \item \textbf{Case 2:} \( \alpha_j \geq \frac{1}{m} - \frac{n^{(-1/2 + \epsilon)}}{m} \) but \( \frac{\SC(A_j, d)}{\SC(B, d)} \geq D_m^\theta - 1 \).

    Using Observation \ref{lemma:bounding_social_cost_ratio} and Lemma~\ref{lemma:optimizer_mu_alpha_combined} coupled with the fact that \( \alpha_j \geq \frac{1}{m} - \frac{n^{(-1/2 + \epsilon)}}{m} \) for every \( j \in [m-1] \setminus S \), we obtain:
    \[
        \frac{\SC(A_j, d)}{\SC(B, d)} \leq \max\left( \frac{m \gm}{1 - n^{-(1/2 - \epsilon)}} - 1, R\left(g, \frac{1-n^{-(1/2-\epsilon)}}{m}\right)^{-1} \right)
    \]

    Note that candidates \( A_j \) and \( B \) satisfy the conditions of Lemma \ref{lemma:bounding_expected_score}. We select constants \( m_0 \) and \( c \) as specified in the lemma, assuming \(D_m^\theta = \omega(\distborda)\).

    Thus the expected Borda score of $A_j$ is at least $cn$ lower than the expected Borda score of some other candidate, say \(C_j\). Recall that the Borda score of a candidate $C$ is given by $\sum_{i \in \N} \sum_{A \in \mathcal{A} \setminus C} \mathbbm{1} (C \succ_{\sigma_i} A)$, where $\sigma_i$ denotes the ranking by voter $i$. For each voter, the contribution to the score difference between \(C_j\) and \(A_j\) lies in an interval of length at most \(2m\). Since voters rank independently, Hoeffding's inequality gives, after adjusting the absolute constant in \(c\),
    \[
    \mathbb{P}(\textsc{$A_j$ is the winner}) \leq \exp\left( \frac{-c^2n}{2m^2} \right)
    \]
    for any \( m > m_0 \). Therefore, for \( m \geq m_0 \), we conclude:
    \[
        \frac{\SC(A_j, d)}{\SC(B, d)} \mathbb{P}(\textsc{$A_j$ is the winner}) \leq \exp\left( \frac{-c^2n}{2m^2} \right) \left( \frac{m \gm}{1 - n^{-(1/2 - \epsilon)}} - 1 + R\left(g, \frac{1-n^{-(1/2-\epsilon)}}{m}\right)^{-1} \right).
    \]

\end{itemize}

    Let
    \[
        H_{m,n}:=
        \exp\left( \frac{- n^{(1/2 + \epsilon)} + 2m}{(2 n^{(1/2 - \epsilon)} - 1)m} \right).
    \]
    Combining the two cases, for any \(m>m_0\), we have: \footnote{For each candidate \( A_j \), either
\[
\frac{\SC(A_j, d)}{\SC(B, d)} \leq D_m^\theta - 1,
\]
or this inequality is not satisfied. The cumulative contribution, meaning probability times social-cost ratio, of all candidates in the former case is at most \(D_m^\theta - 1\).}

    \begin{align}
        & \textsc{dist}^{\theta}_{PL}(\textsc{Borda},n,m) \nonumber\\
        \leq\;& D_m^\theta -1 + \exp\left(\frac{-c^2n}{2m^2}\right) \left( \frac{m \gm}{1 - n^{-(1/2 - \epsilon)}} - 1 + R\left(g, \frac{1-n^{-(1/2-\epsilon)}}{m}\right)^{-1} \right) \nonumber\\
        & \hspace{7 em} +m(m-1)\Gamma_g\left(\frac{1-n^{-(1/2-\epsilon)}}{m}\right)H_{m,n}.
    \end{align}

    Taking the limsup as \(n \to \infty\), the last two terms vanish. \footnote{The last term vanishes since \(\Gamma_g(\eta)\) is finite for the relevant \(\eta\) values (Equation \eqref{eq:Gamma_g_eta_defn}) and \(m\) is fixed while \(H_{m,n}\to0\).} Thus,
    \[
        D_m^\theta
        \leq
        D_m^\theta -1,
    \]
    a contradiction.  

\end{proof}

\section{Proof of Theorem \ref{theorem:RD_distortion_upper_bound}}{\label{sec:RD_distortion_upper_bound_proof}}

\noindent\textbf{Theorem~\ref{theorem:RD_distortion_upper_bound} (Restated):} \thmrdub

\begin{proof}
    The probability that voter \(i\) votes for candidate \(W\) as her top candidate is upper bounded by \(g\left(\frac{d(i,B)}{d(i,W)}\right)\), which is the probability that \(W\) is ranked above \(B\). Therefore, under \textsc{Random Dictator}, the probability of \(W\) winning satisfies
    \begin{equation}{\label{eq:bounding_prob_winner}}
        \mathbb{P}[W \text{ wins}] \leq \frac{1}{n} \left(\sum_{i \in \N} g\left(\frac{d(i,B)}{d(i,W)}\right)\right).
    \end{equation}

	Recall that we define the set of candidates in \(\mathcal{A} \setminus \{B\}\) as \(\{A_1,A_2,\ldots,A_{m-1}\}\). In the rest of the analysis, denote \(d(i,A_j)\) by \(y_{i,j}\) for all \(j \in [m-1]\), \(d(i,B)\) by \(b_i\) for every \(i \in \N\), and \(d(B,A_j)\) by \(z_j\). For every metric \(d\), we bound the distortion as follows.
    \begin{align}
        \textsc{dist}^{(g)}(\textsc{Random Dictator},n,m) \leq & \sum_{j=1}^{m-1} \left(\mathbb{P}[A_j \text{ wins}] \frac{\sum_{i \in \N} y_{i,j}}{\sum_{i \in \N} b_i}\right) + (1- \sum_{j=1}^{m-1}\mathbb{P}[A_j \text{ wins}]) \\
        = & \sum_{j=1}^{m-1} \mathbb{P}[A_j \text{ wins}] \left(\frac{\sum_{i \in \N} y_{i,j}}{\sum_{i \in \N} b_i} -1\right) + 1\\
        & \overset{(a)}{\leq} \sum_{j=1}^{m-1} \frac{1}{n}\left(\sum_{i \in \N} g\left(\frac{b_i}{y_{i,j}}\right)\right) \frac{\sum_{i \in \N} (y_{i,j}-b_i)}{\sum_{i \in \N} b_i} +1\\
        & {\leq} \sum_{j=1}^{m-1} \frac{1}{n}\left(\sum_{i \in \N} g\left(\frac{b_i/z_j}{y_{i,j}/z_j}\right)\right) \frac{\sum_{i \in \N} (y_{i,j}/z_j-b_i/z_j)}{\sum_{i \in \N} b_i/z_j} +1\\
        & \overset{(d)}{\leq} \sum_{j=1}^{m-1}\frac{ \left(\sum_{i \in \N} g\left(\frac{b_i/z_j}{y_{i,j}/z_j}\right)\right)}{\sum_{i \in \N} b_i/z_j} +1\\
        & \overset{(e)}{\leq} (m-1) \frac{g\left(\frac{x^{\ast}_{\textsc{mid}}}{1-x^{\ast}_{\textsc{mid}}}\right)}{x^{\ast}_{\textsc{mid}}} +1 = (m-1)\gm + 1.
    \end{align}

$(a)$ follows from Equation \eqref{eq:bounding_prob_winner}.

$(d)$ follows from \(y_{i,j}-b_i \leq z_j\), which follows from the triangle inequality.

$(e)$ follows by considering the two cases \(\frac{b_i}{z_j} \leq 1\) and \(\frac{b_i}{z_j} \geq 1\). 

When \(\frac{b_i}{z_j} \leq 1\), the triangle inequality gives \(\frac{y_{i,j}}{z_j} \geq 1 - \frac{b_i}{z_j}\). Similarly, when \(\frac{b_i}{z_j} \geq 1\), it gives \(\frac{y_{i,j}}{z_j} \geq \frac{b_i}{z_j}-1\). Thus,
%
\begin{align}
\frac{ g\left(\frac{b_i/z_j}{y_{i,j}/z_j}\right)}{b_i/z_j}
&\leq
\max\left\{
\sup_{x \in (0,1)}\frac{g(\frac{x}{1-x})}{x},
\sup_{x \in (1,\infty)}\frac{g(\frac{x}{x-1})}{x}
\right\}
\quad \forall i \in \N, \nonumber\\
\frac{ \sum_{i =1}^{n}g\left(\frac{b_i/z_j}{y_{i,j}/z_j}\right)}{\sum_{i =1}^{n} b_i/z_j}
&\leq
\max\left\{
\frac{g\left(\frac{x^{\ast}_{\textsc{mid}}}{1-x^{\ast}_{\textsc{mid}}}\right)}{x^{\ast}_{\textsc{mid}}},
1
\right\}.
\end{align}

The last inequality follows from the fact that  $\frac{g(\frac{x}{x-1})}{x} \leq 1$ when $x \geq 1$. Further, we have $\gm \geq 1$ for all valid $g$.
\end{proof}

\addtocontents{toc}{\protect\setcounter{tocdepth}{2}} 

\end{document}